\documentclass[nohyperref]{article}

\usepackage{mfirstuc}
\MFUnocap{the}\MFUnocap{and}\MFUnocap{a}
\usepackage{microtype}
\usepackage{graphicx}
\usepackage{subcaption}
\usepackage{booktabs} 
\usepackage{hyperref}

\newcommand{\alglinelabel}{%
  \addtocounter{ALC@line}{-1}
  \refstepcounter{ALC@line}
  \label
}

\usepackage{amsmath,lineno,amssymb,mathtools,url,verbatim}
\usepackage{nccmath}
\usepackage{pgfplots}
\pgfplotsset{grid style={red}}
\usepackage{filecontents}
\usepackage{microtype}
\usepackage{float}
\usetikzlibrary{matrix}
\usepackage{tikz-cd}
\usepackage[algo2e]{algorithm2e} 
\usepackage{algorithmic}
\usepackage[inline]{enumitem}
\usepackage{tabularx}
\usetikzlibrary{decorations.markings}
\usepackage{url} 
\usepackage{graphicx}
\usepackage{tikz}
\usetikzlibrary{automata, positioning, arrows}
\usetikzlibrary{matrix}
\usetikzlibrary{arrows,positioning,shapes.geometric}
\usetikzlibrary{calc,trees,positioning,arrows,chains,shapes.geometric,%
    decorations.pathreplacing,decorations.pathmorphing,shapes,%
   shapes.symbols, trees,backgrounds}
      \tikzstyle{blockzm1} = [rectangle, draw, fill=white, 
    text width=1.5em, text centered, rounded corners, minimum height=1.5em, minimum width = 1.5em]
    
      \tikzstyle{blockzm2} = [rectangle, draw, fill=white, 
    text width=1.0em, text centered, rounded corners, minimum height=1.0em, minimum width = 1.0em]
    
      \tikzstyle{blockzm1r} = [rectangle, draw, fill=red!70, 
    text width=1.5em, text centered, rounded corners, minimum height=1.5em, minimum width = 1.5em]
      \tikzstyle{blockzm1o} = [rectangle, draw, fill=orange!80, 
    text width=1.5em, text centered, rounded corners, minimum height=1.5em, minimum width = 1.5em]
      \tikzstyle{blockzm1y} = [rectangle, draw, fill=yellow!50, 
    text width=1.5em, text centered, rounded corners, minimum height=1.5em, minimum width = 1.5em]
      \tikzstyle{blockzm1g} = [rectangle, draw, fill=green!60, 
    text width=1.5em, text centered, rounded corners, minimum height=1.5em, minimum width = 1.5em]
        \tikzstyle{blockzm1p} = [rectangle, draw, fill=blue!50, 
    text width=1.5em, text centered, rounded corners, minimum height=1.5em, minimum width = 1.5em]
    
   \tikzstyle{blockz} = [rectangle, draw, fill=white, 
    text width=1.75em, text centered, rounded corners, minimum height=1.75em, minimum width = 1.75em]
     \tikzstyle{blockzflexi} = [rectangle, draw, fill=white, 
    text centered, rounded corners]
    \tikzstyle{blockz2} = [rectangle, draw, fill=white, 
    text width=2em, text centered, rounded corners, minimum height=2em, minimum width = 2em]
        \tikzstyle{blockzL} = [rectangle, draw, fill=white, 
    text width=4em, text centered, rounded corners, minimum height=3em, minimum width = 4em]
    
       \tikzstyle{vertex}=[circle,fill=black!25,minimum size=20pt,inner sep=0pt]
\tikzstyle{selected vertex} = [vertex, fill=red!24]
\tikzstyle{edge} = [draw,thick,-]
\tikzstyle{weight} = [font=\small]
\tikzstyle{selected edge} = [draw,line width=5pt,-,red!50]
\tikzstyle{ignored edge} = [draw,line width=5pt,-,black!20]
    
\tikzset{
  treenode/.style = {align=center, inner sep=0pt, text centered,
    font=\sffamily},
  arn_n/.style = {treenode, circle, white, font=\sffamily\bfseries, draw=black,
    fill=black, text width=1.5em},
  arn_r/.style = {treenode, circle, red, draw=red, 
    text width=1.5em, very thick},
  arn_x/.style = {treenode, rectangle, draw=black,
    minimum width=0.5em, minimum height=0.5em}
}
\usepackage{tikz-cd}
\usepackage{verbatim}
\usepackage{hyperref}

\usepackage[accepted]{icml2023}

\usepackage{amsthm}
\usepackage{thmtools}
\usepackage{thm-restate}

\usepackage{amsmath}
\usepackage{amssymb}
\usepackage{mathtools}

\numberwithin{algorithm}{section}

\newtheorem{dfn}[algorithm]{Definition}

\newtheorem{thmm}[algorithm]{Theorem}

\theoremstyle{definition}
\newtheorem{exa}[algorithm]{Example}
\theoremstyle{definition}

\theoremstyle{plain}

\newtheorem{cor}[algorithm]{Corollary}
\newtheorem{lem}[algorithm]{Lemma}

\usepackage{mathtools}
\DeclarePairedDelimiter\ceil{\lceil}{\rceil}

\newcommand{\li}[1]{\#\mathrm{Line}[#1]}
\newcommand{\Es}{\mathcal{E}}
\newcommand{\C}{\mathcal{C}}
\newcommand{\AP}{\mathcal{P}}
\newcommand{\Child}{\mathrm{Children}}
\newcommand{\Desc}{\mathrm{Descendants}}
\newcommand{\nxt}{\mathrm{Next}}

\newcommand{\R}{\mathbb{R}}

\newcommand{\Z}{\mathbb{Z}}

\newcommand{\T}{\mathcal{T}}
\newcommand{\Sd}{\mathcal{S}}

\newcommand{\be}{\beta}

\newcommand{\la}{\lambda}

\newcommand{\NN}{\mathrm{NN}}

\newcommand{\rad}{\mathrm{diam}}
\newcommand{\diam}{\mathrm{diam}}

\newcommand{\bs}{\hfill $\blacksquare$}

\newenvironment{hproof}{%
  \proof}{\endproof}

\usepackage{array} 
\usepackage{varwidth} 
\newcolumntype{M}{>{\begin{varwidth}{3cm}}l<{\end{varwidth}}}

\usepackage[textsize=tiny]{todonotes}

\icmltitlerunning{A new compressed cover tree for $k$-nearest neighbors}

\begin{document}

\twocolumn[
\icmltitle{\capitalisewords{A new near-linear time algorithm for $k$-Nearest neighbor search using a compressed cover tree}}




\begin{icmlauthorlist}
\icmlauthor{Yury Elkin}{1}
\icmlauthor{Vitaliy Kurlin}{1}
\end{icmlauthorlist}

\icmlaffiliation{1}{Department of Computer Science, University of Liverpool, Liverpool, United Kingdom}

\icmlcorrespondingauthor{Vitaliy Kurlin}{vitaliy.kurlin@gmail.com}

\icmlkeywords{nearest neighbor search, parameterized complexity, metric space, cover tree}

\vskip 0.3in
]



\printAffiliationsAndNotice{} 

\begin{abstract}
Given a reference set $R$ of $n$ points and a query set $Q$ of $m$ points in a metric space, this paper studies an important problem of finding $k$-nearest neighbors of every point $q \in Q$ in the set $R$ in a near-linear time. 
In the paper at ICML 2006, Beygelzimer, Kakade, and Langford introduced a cover tree on $R$ and attempted to prove that this tree can be built in $O(n\log n)$ time while the nearest neighbor search can be done in $O(n\log m)$ time with a hidden dimensionality factor.
This paper fills a substantial gap in the past proofs of time complexity 
by defining a simpler compressed cover tree on the reference set $R$. 
The first new algorithm constructs a compressed cover tree in $O(n \log n)$ time.
The second new algorithm finds all $k$-nearest neighbors of all points from $Q$ using a compressed cover tree in time $O(m(k+\log n)\log k)$ with a hidden dimensionality factor depending on point distributions of the given sets $R,Q$ but not on their sizes.  
\end{abstract}


\section{\capitalisewords{The neighbor search, overview of results}}
\label{sec:text:intro}

In the modern formulation, the $k$-nearest neighbor problem is to find all $k\geq 1$ nearest neighbors in a given reference set $R$ for all points from another given query set $Q$.

Both sets belong to a common ambient space $X$ with a distance metric $d$ satisfying all metric axioms.
The simplest example of $X$ is $\R^n$ with the Euclidean metric. 
A query set $Q$ can be a point or a finite subset of a reference set~$R$. 

The \emph{exact} $k$-nearest neighbor problem asks for all true (exact) $k$-nearest neighbors in $R$ for every point $q\in Q$. 
\smallskip

Another (probabilistic) version of the $k$-nearest neighbor search \citet{har2006fast,manocha2007empirical} aims to find exact $k$-nearest neighbors with a given probability. 
The approximate version \citet{arya1993approximate,krauthgamer2004navigating,andoni2018approximate, wang2021comprehensive} of the nearest neighbor search looks for an $\epsilon$-approximate neighbor $r\in R$ of every query point $q \in Q$ such that $d(q,r) \leq (1+\epsilon)d(q,\NN(q))$ , where $\epsilon>0$ is fixed and $\NN(q)$ is the exact first nearest neighbor of $q$. 
\smallskip

\begin{restatable}[diameter and aspect ratio]{dfn}{dfnaspectratio}
\label{dfn:radius+d_min}
For any finite set $R$ with a metric $d$, the \emph{diameter} is $\rad(R) = \max\limits_{p \in R}\max\limits_{q \in R}d(p,q)$. 
The \emph{aspect ratio} is $\Delta(R) = \dfrac{\rad(R)}{d_{\min}(R)}$, where $d_{\min}(R)$ is the shortest distance between points of $R$.
\end{restatable}


\begin{restatable}[$k$-nearest neighbor set $\NN_k$]{dfn}{dfnkNearestNeighbor}\label{dfn:kNearestNeighbor}
For any point $q\in Q$, let $d_1\leq\dots\leq d_{|R|}$ be ordered
 distances from $q$ to all points of a reference set $R$ whose size (number of points) is denoted by $|R|$.
For any $k\geq 1$, the $k$-nearest neighbor set $\NN_k(q;R)$ consists of all $u\in R$ with $d(q,u)\leq d_k$. 
\end{restatable}

\noindent
For $Q=R=\{0,1,2,3\}$, the point $q=1$ has ordered distances $d_1=0<d_2=1=d_3<d_4=2$. 
The nearest neighbor sets are $\NN_1(1;R)=\{1\}$,
$\NN_2(1;R)=\{0,1,2\}=\NN_3(1;R)$, $\NN_4(1;R)=R$. 
So 0 can be a 2nd neighbor of 1, then 2 becomes a 3rd neighbor of 1, or these neighbors of $0$ can be found in a different order.

\begin{restatable}[all $k$-nearest neighbors search]{prob}{proknn}\label{pro:knn}
Let $Q,R$ be finite subsets of query and reference sets in a metric space $(X,d)$.
For any fixed $k\geq 1$, design an algorithm to exactly find $k$ distinct points from $\NN_k(q;R)$ for all $q\in Q$ so that the parametrized worst-case time complexity is near-linear in time $\max\{|Q|,|R|\}$, where hidden constants may depend on structures of $Q,R$ but not on their sizes $|Q|,|R|$. 
\end{restatable}

\noindent 
In a metric space, let $\bar B(p,t)$ be the closed ball with a center $p$ and a radius $t\geq 0$.
The notation $|\bar B(p,t)|$ denotes the number (if finite) of points in the closed ball. 
Definition~\ref{dfn:expansion_constant} recalls the expansion constant $c$ from \citet{beygelzimer2006cover} and introduces the new minimized expansion constant $c_m$, which is a discrete analog of the doubling dimension \citet{cole2006searching}. 

\begin{restatable}[expansion constants $c$ and $c_m$]{dfn}{dfnexpansionconstant}\label{dfn:expansion_constant}
A subset $R$ of a metric space $(X,d)$ is called \emph{locally finite} if the set $\bar{B}(p,t) \cap R$ is finite for all $p \in X$ and $t \in \R_{+}$. 
	Let $R$ be a locally finite set in a metric space $X$. 
\smallskip

	The \emph{expansion constant} $c(R)$ is the smallest $c(R)\geq 2$ such that $|\bar{B}(p,2t)|\leq c(R) \cdot |\bar{B}(p,t)|$ for any point $p\in R$ and $t\geq 0$, see \citet{beygelzimer2006cover}. 
\smallskip
 
Introduce the new \emph{minimized expansion constant} $c_m(R) = \lim\limits_{\xi \rightarrow 0^{+}}\inf\limits_{R\subseteq A\subseteq X}\sup\limits_{p \in A,t > \xi}\dfrac{|\bar{B}(p,2t) \cap A|}{|\bar{B}(p,t) \cap A|}$, where $A$ is a locally finite set which covers $R$.
\end{restatable}

\begin{restatable}{lem}{lemexpansionconstantproperty}\label{lem:expansion_constant_property}
For any finite sets $R\subseteq U$ in a metric space, we have that $c_m(R) \leq c_m(U)$ and $c_m(R) \leq c(R)$.
\end{restatable}

\noindent
Note that both $c(R), c_m(R)$ are always defined when $R$ is finite. 
We show below that a single outlier can make the expansion constant $c(R)$ as large as $O(|R|)$.
\smallskip

In the Euclidean line $\R$, 
The set $R=\{1,2,\dots,n,2n+1\}$ of $|R|=n+1$ points has $c(R)=n+1$ because $\bar B(2n+1;n)=\{2n+1\}$ is a single point, while $\bar B(2n+1;2n)=R$ is the full set of $n+1$ points. 
On the other hand, the same set $R$ can be extended to a larger uniform set $A=\{1,2,\dots,2n-1,2n\}$ whose expansion constant is $c(A)=2$.
So the minimized constant of the original set $R$ is much smaller: $c_m(R) \leq c(A)=2<c(R)=n+1$.
\smallskip

\noindent
The constant $c$ from \citet{beygelzimer2006cover} equals $2^{\text{dim}_{KR}}$ from \citet[Section~2.1]{krauthgamer2004navigating}. 
\smallskip

In \citet[Section~1.1]{krauthgamer2004navigating} the doubling dimension $2^{\text{dim}}$ is defined as a minimum value $\rho$ such that any set $X$ can be covered by $2^{\rho}$ sets whose diameters are half of the diameter of $X$.
The past work \citet{krauthgamer2004navigating} proves that  $2^{\text{dim}} \leq 2^{n}$ for any subset of $\R^n$.
\smallskip

Theorem \ref{thm:normed_space_exp_constant} in appendix~\ref{sec:minimized_exp_constant} will prove that $c_m(R) \leq 2^{n}$ for any a finite subset $R\subset\R^{n}$, so $c_m(R)$ mimics $2^{\text{dim}}$.
\medskip

\noindent
\textbf{Navigating nets}.
In 2004, \citet[Theorem~2.7]{krauthgamer2004navigating} claimed that a navigating net can be constructed in time 
$O\big(2^{O(\text{dim}_{KR}(R)} |R| (\log|R|) \log(\log|R|)\big)$ and
all $k$-nearest neighbors of a query point $q$ can be found in time $O(2^{O(\text{dim}_{KR}(R \cup \{q\})}(k + \log|R|)$, where $\text{dim}_{KR}(R \cup \{q\})$ is the expansion constant defined above. 
The paper above sketched a proof of \citet[Theorem~2.7]{krauthgamer2004navigating} in one sentence and skipped pseudo-codes.
Unfortunately, the authors didn't reply to our request for these details. 
\medskip

\noindent
\textbf{Modified navigating nets} \citet{cole2006searching} were used in 2006 to claim the time $O(\log(n) + (1/\epsilon)^{O(1)})$ to find $(1+\epsilon)$-approximate neighbors.
The proof and pseudo-code were skipped for this claim and for the construction of the modified navigating net for the claimed time $O(|R| \cdot \log(|R|))$. 
\smallskip

\noindent
\textbf{Cover trees}. 
In 2006, \citet{beygelzimer2006cover} introduced a cover tree inspired by the navigating nets \citet{krauthgamer2004navigating}. 
This cover tree was designed to prove a worst-case time for the nearest neighbor search in terms of the size $|R|$ of a reference set $R$ and the expansion constant $c(R)$ from Definition \ref{dfn:expansion_constant}. 
Assume that a cover tree is already constructed on the set $R$. Then \citet[Theorem~5]{beygelzimer2006cover} claimed that a nearest neighbor of any query point $q \in Q$ could be found in time $O(c(R)^{12} \cdot \log|R|)$. 
\smallskip

\noindent
In 2015, \citet[Section~5.3]{curtin2015improving} pointed out that the proof of \citet[Theorem~5]{beygelzimer2006cover} contains a crucial gap, now also confirmed by a specific example in \citet[Counterexample~5.2]{elkin2022counterexamples}. 
\smallskip

The time complexity result of the cover tree construction algorithm \citet[Theorem~6]{beygelzimer2006cover} had a similar issue, the gap of which is exposed rigorously in \citet[Counterexample~4.2]{elkin2022counterexamples}.
\smallskip

\noindent 
\textbf{Further studies in cover trees.} A noteworthy paper on cover trees \citet{kollar2006fast} introduced a new probabilistic algorithm for the nearest neighbor search, as well as corrected the pseudo-code of the cover tree construction algorithm of \citet[Algorithm~2]{beygelzimer2006cover}.
Later in 2015, a new, more efficient implementation of cover tree was introduced in \citet{izbicki2015faster}. However, no new time-complexity results were proven. 
\smallskip

Another study \citet{jahanseir2016transforming} explored connections between modified navigating nets \citet{cole2006searching} and cover trees \citet{beygelzimer2006cover}.
\smallskip

Several papers \citet{beygelzimer2006coverExtend, ram2009linear, curtin2015plug} studied the possibility of solving $k$-nearest neighbor Problem \ref{pro:knn} by using cover trees on both sets $Q,R$, 
see \citet[Section~6]{elkin2022counterexamples}.
\medskip 
\noindent
\textbf{New contributions}.
This work corrects the past gaps of the single-tree approach \citet{beygelzimer2006cover}, which were discovered in \citet{elkin2022counterexamples} by using a new compressed cover tree $\T(R)$ from Definition \ref{dfn:cover_tree_compressed}.
\begin{itemize}
    \item Theorem \ref{thm:construction_time}  and Corollary \ref{cor:construction_time_KR} estimate the time to build a compressed cover tree, which corrects the proof of \citet[Theorem~6]{beygelzimer2006cover}.
    \item Theorem \ref{thm:knn_KR_time} and Corollary~\ref{cor:cover_tree_knn_miniziminzed_constant_time} estimate the time to find all $k$-nearest neighbors as in Problem~\ref{pro:knn}.
These advances correct and generalize \citet[Theorem~5]{beygelzimer2006cover}.
\end{itemize}

\begin{table*}
	\centering
	\caption{Building data structures with hidden $c_m(R)$ or dimensionality constant $2^{\text{dim}}$ \citet[Section~1.1]{krauthgamer2004navigating}.}
        \vskip 0.15in
	\begin{tabular}{|V{35mm}|V{5.5cm}|V{25mm}|V{35mm}|}
		\hline
		Data structure    & claimed time complexity   & space & proofs \\
		\hline
		Navigating nets \citet{krauthgamer2004navigating} & $O\big(2^{O(\text{dim})} \cdot |R| \cdot \log(\Delta) \cdot \log(\log((\Delta ))\big)$
		& $O(2^{O(\text{dim})}|R|)$  & \citet[Theorem~2.5]{krauthgamer2004navigating} \\
		\hline
		Compressed cover tree [Definition \ref{dfn:cover_tree_compressed}] & $O\big(c_m(R)^{O(1)}\cdot|R|\log(\Delta(R))\big)$    & $O(|R|)$ Lemma~\ref{lem:linear_space_cover_tree}     & Theorem \ref{thm:construction_time}  \\
		\hline           
	\end{tabular}
	\label{table:text:dim:construction}  
\end{table*}
\begin{table*}
	\label{table:text:KR:construction}
	\centering
	\caption{Results for building data structures with the hidden classical expansion constant $c(R)$ of Definition \ref{dfn:expansion_constant} or KR-type constant $2^{\text{dim}_{KR}}$ \citet[Section~2.1]{krauthgamer2004navigating}.}
        \vskip 0.15in
	\begin{tabular}{|V{35mm}|V{6cm}|V{25mm}|V{40mm}|}
		\hline
		Data structure    & claimed time complexity   & space & proofs \\
		\hline
		Navigating nets \citet{krauthgamer2004navigating} & $O\big(2^{O(\text{dim}_{KR})} \cdot |R| \log(|R|) \log(\log|R| )\big)$, \citet[Theorem~2.6]{krauthgamer2004navigating} & $O(2^{O(\text{dim})}|R|)$ & Not available \\
		\hline
		Cover tree \citet{beygelzimer2006cover}   &  $O(c(R)^{O(1)} \cdot |R| \cdot \log|R|)$, \citet[Theorem~6]{beygelzimer2006cover} & $O(|R|)$ & \citet[Counterexample~4.2]{elkin2022counterexamples} shows that the past proof is incorrect \\ 
		\hline
		Compressed cover tree [Definition \ref{dfn:cover_tree_compressed}] &     $O\big(c(R)^{O(1)} \cdot |R| \cdot \log|R| \big)$  &  $O(|R|)$ Lemma~\ref{lem:linear_space_cover_tree}     & Corollary \ref{cor:construction_time_KR} \\
		\hline       
	\end{tabular}
\end{table*}
\begin{table*}
	\label{table:text:KR:knearest}
	\centering
	\caption{Results for exact $k$-nearest neighbors of one query point $q \in X$ using the hidden classical expansion constant $c(R)$ of Definition \ref{dfn:expansion_constant} or KR-type constant $2^{\text{dim}_{KR}}$  \citet[Section~2.1]{krauthgamer2004navigating} and assuming that all data structures are already built. Note that the dimensionality factor $2^{\text{dim}_{KR}}$ is equivalent  to $c(R)^{O(1)}$. }
        \vskip 0.15in
	\begin{tabular}{|V{3.0cm}|V{5cm}|V{22mm}|V{50mm}|}
		\hline
		Data structure     & claimed time complexity   & space  & proofs \\
		\hline
		Navigating nets \citet{krauthgamer2004navigating} & $O\big(2^{O(\text{dim}_{KR})}(\log(|R|) + k)\big)$ for $k\geq 1$ \citet[Theorem~2.7]{krauthgamer2004navigating} & $O(2^{O(\text{dim})} |R|)$ & Not available \\
		\hline
		Cover tree \citet{beygelzimer2006cover}   &  $O\big(c(R)^{O(1)}\log|R|\big)$ for $k=1$ \citet[Theorem~5]{beygelzimer2006cover} & $O(|R|)$ & \citet[Counterexample~5.2]{elkin2022counterexamples} shows that the past proof is incorrect \\ 
		\hline
		Compressed cover tree, Definition \ref{dfn:cover_tree_compressed} &     $O\big(c(R \cup \{q\})^{O(1)} \cdot \log(k) \cdot (\log(|R|) + k)\big)$  & $O(|R|)$, Lemma~\ref{lem:linear_space_cover_tree}      & Theorem \ref{thm:knn_KR_time} \\
		\hline            
	\end{tabular}
\end{table*}
\begin{table*}
	\centering
	\caption{Results for exact $k$-nearest neighbors of one point $q$ using hidden $c_m(R)$ or dimensionality constant $2^{\text{dim}}$ \citet[Section~1.1]{krauthgamer2004navigating}  assuming that all structures are built.}
        \vskip 0.15in
	\begin{tabular}{|V{35mm}|V{55mm}|V{22mm}|V{40mm}|}
		\hline
		Data structure     & claimed time complexity & space & proofs \\
		\hline
		Navigating nets \citet{krauthgamer2004navigating} & $O\big(2^{O(\text{dim})} \cdot \log(\Delta) + |\bar{B}(q,O(d(q,R))|\big)$ for $k = 1$  & $O(2^{O(\text{dim})}|R|)$  & a proof outline in \citet[Theorem~2.3]{krauthgamer2004navigating} \\
		\hline
		Compressed cover tree, Definition~\ref{dfn:cover_tree_compressed} &     $O\big(\log(k)\cdot (c_m(R)^{O(1)} \log(|\Delta|) + |\bar{B}(q,O(d_k(q,R))| ) \big )$  & $O(|R|)$, Lemma~\ref{lem:linear_space_cover_tree}     & Corollary~\ref{cor:cover_tree_knn_miniziminzed_constant_time} \\
		\hline          
	\end{tabular}
	\label{table:text:dim:knearest}
\end{table*}

\begin{figure*}
	\centering
	\begin{tikzpicture}[align=center, node distance = 1.0cm, scale = 0.45]

    \node (scaleinf) {$l = \infty$};
    \node [below = 12.5pt of scaleinf](scalemidinf) {};
     \node[below=25pt of scaleinf] (scale6) {$l = 5$};
    \node[below of =scale6] (scale5) {$l = 4$};
    \node[below of =scale5] (scale4) {$l = 3$};
	\node[below of =scale4] (scale3) {$l = 2$};
	\node[below of =scale3] (scale2) {$l = 1$};
	\node[below=25pt of scale2] (scale1) {$l = -\infty$};

 	\node [blockzm2,  right=70pt of scaleinf ] (node1x) {$r$};
 	\node [blockzm2, below=25pt of node1x ] (node1a) {$r$};
 	\node [blockzm2, below of = node1a ] (node1b) {$r$};
 	\node [blockzm2,  below of = node1b ] (node1c) {$r$};
 	\node [blockzm2,   below of = node1c ] (node1d) {$r$};
 	\node [blockzm2,  below of = node1d ] (node1e) {$r$};
 	\node [blockzm2,   below =25pt of node1e ] (node1f) {$r$};
 	
 	\draw[dashed, ->] (node1x) -> (node1a);
 	    \draw[->] (node1a) -> (node1b);
 	     \draw[->] (node1b) -> (node1c);
 	      \draw[->] (node1c) -> (node1d);
 	       \draw[->] (node1d) -> (node1e);
 	      \draw[dashed, ->] (node1e) -> (node1f);

 		\node [blockzm2, right=10pt of node1b ] (node2b) {$p_{4}$};
 	\node [blockzm2,  below of = node2b ] (node2c) {$p_{4}$};
 	\node [blockzm2,   below of = node2c ] (node2d) {$p_{4}$};
 	\node [blockzm2,  below of = node2d ] (node2e) {$p_{4}$};
 	\node [blockzm2,   below =25pt of node2e ] (node2f) {$p_{4}$};
 	
 	    \draw[->] (node1a) -> (node2b);
 	    
 	    \draw[->] (node2b) -> (node2c);
 	      \draw[->] (node2c) -> (node2d);
 	       \draw[->] (node2d) -> (node2e);
 	      \draw[dashed, ->] (node2e) -> (node2f);

 		\node [blockzm2, right = 10pt of node2c ] (node3c) {$p_{3}$};
 	\node [blockzm2,   below of = node3c ] (node3d) {$p_{3}$};
 	\node [blockzm2,  below of = node3d ] (node3e) {$p_{3}$};
 	\node [blockzm2,   below =25pt of node3e ] (node3f) {$p_{3}$};
 	
 	        \draw[->] (node2b) -> (node3c);
 	      \draw[->] (node3c) -> (node3d);
 	       \draw[->] (node3d) -> (node3e);
 	      \draw[dashed, ->] (node3e) -> (node3f);

 		\node [blockzm2,   left = 10 pt of  node1d ] (node4d) {$p_{2}$};
 	\node [blockzm2,  below of = node4d ] (node4e) {$p_{2}$};
 	\node [blockzm2,   below =25pt of node4e ] (node4f) {$p_{2}$};

            \draw[->] (node1c)   -> (node4d);
 	       \draw[->] (node4d) -> (node4e);
 	      \draw[dashed, ->] (node4e) -> (node4f);

 	\node [blockzm2,  left = 10 pt of node4e ] (node5e) {$p_{1}$};
 	\node [blockzm2,   below =25pt of node5e ] (node5f) {$p_{1}$};
 	
 	\draw[->] (node4d) -> (node5e);
 	 \draw[dashed, ->] (node5e) -> (node5f);
 	
 	\node[rectangle, draw, fill=white, 
    text width=1.0em, text centered, rounded corners, minimum height=50pt, minimum width = 1.0em, right = 220pt of scalemidinf] (ynode1a) {$r$};
    
    	\node[rectangle, draw, fill=white, 
    text width=1.0em, text centered, rounded corners, minimum height=50pt, minimum width = 1.0em, below = 10pt of ynode1a] (ynode1b) {$r$};
    
    	\node[rectangle, draw, fill=white, 
    text width=1.0em, text centered, rounded corners, minimum height=80pt, minimum width = 1.0em, below = 10pt of ynode1b] (ynode1c) {$r$};

    \node[blockzm2, above right = -18 pt and 10 pt of ynode1b] (ynode2a) {$p_4$};
    
    	\node[rectangle, draw, fill=white, 
    text width=1.0em, text centered, rounded corners, minimum height=110pt, minimum width = 1.0em, below = 10pt of ynode2a] (ynode2b) {$p_4$};
    
    \node[rectangle, draw, fill=white, 
    text width=1.0em, text centered, rounded corners, minimum height=110pt, minimum width = 1.0em, right = 10pt of ynode2b] (ynode3a) {$p_3$};

        \node[blockzm2, above left = -15 pt and 10 pt of ynode1c] (ynode4a) {$p_2$};
    
    	\node[rectangle, draw, fill=white, 
    text width=1.0em, text centered, rounded corners, minimum height=55pt, minimum width = 1.0em, below = 10pt of ynode4a] (ynode4b) {$p_2$};
    
    \node[rectangle, draw, fill=white, 
    text width=1.0em, text centered, rounded corners, minimum height=55pt, minimum width = 1.0em, left = 10pt of ynode4b] (ynode5a) {$p_1$};

    \draw[->] (ynode1a) -> (ynode1b);
    \draw[->] (ynode1b) -> (ynode1c);
    \draw[->] (ynode1a) -> (ynode2a);
    \draw[->] (ynode2a) -| (ynode3a);
    \draw[->] (ynode2a) -> (ynode2b);
    
      \draw[->] (ynode1b) -> (ynode4a);
     \draw[->] (ynode4a) -| (ynode5a);
     \draw[->] (ynode4a) -> (ynode4b);
    

 	\node[blockzm2, right = 285pt of scale6] (xnode1) {$r$};
	\node[blockzm2, below right = 10 pt and 10 pt of xnode1] (xnode2) {$p_4$};
 	\node[blockzm2, below right = 70 pt and 10 pt of xnode1] (xnode4) {$p_2$};
 	\node[blockzm2, below right = 10 pt and 8 pt of xnode4] (xnode5) {$p_1$};
 	\node[blockzm2, below right = 12 pt and 8 pt of xnode2] (xnode3) {$p_3$};
 	
 	\draw[->]	(xnode1) -> (xnode2);
 	\draw[->]	(xnode2) -> (xnode3);
 	\draw[->]	(xnode1) -> (xnode4);
 	\draw[->]	(xnode4) -> (xnode5);

\end{tikzpicture}
	\caption{A comparison of past cover trees and a new compressed cover tree in Example \ref{exa:implicitexplicitexample}. \textbf{Left:} an implicit cover tree contains infinite repetitions. \textbf{Middle:} an explicit cover tree. \textbf{Right:} a compressed cover tree from Definition \ref{dfn:cover_tree_compressed} includes each given point exactly once.  }
	\label{fig:text:tripleexample}
\end{figure*}
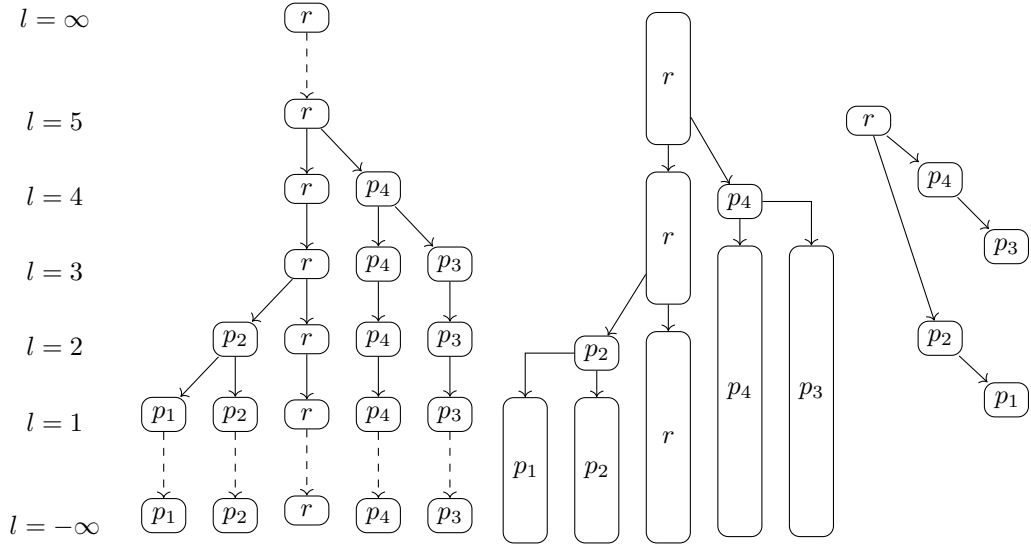


\section{\capitalisewords{A new compressed cover tree}}
\label{sec:text:cover_tree}

This section introduces in Definition \ref{dfn:cover_tree_compressed} a new compressed cover tree to solve Problem \ref{pro:knn}. 
We also prove relevant properties of the expansion constant $c(R)$ and minimized expansion constant $c_m(R)$ of Definition~\ref{dfn:expansion_constant}.
All extra details and proofs of this section are in Appendices~\ref{sec:cover_tree},\ref{sec:distinctive_descendant_set}.
\smallskip

A compressed cover tree in Definition~\ref{dfn:cover_tree_compressed} will be significantly simpler than an explicit cover tree \citet[Definition~2.2]{elkin2022counterexamples}, where any given point $p$ can appear in many different nodes simultaneously. 
\smallskip

To regain the functionality of the explicit cover tree,
we introduce the new concept of a distinctive descendant set $\Sd_i(p, \T(R))$ in Definition~\ref{dfn:distinctive_descendant_set}.
See Figure~\ref{fig:text:tripleexample} for a comparison between implicit, explicit, and compressed cover trees. 

\begin{restatable}[a compressed cover tree $\T(R)$]{dfn}{dfncovertreecompressed}\label{dfn:cover_tree_compressed}
Let $R$ be a finite set in a metric space $(X,d)$. 
	\emph{A compressed cover tree} $\T(R)$ has the vertex set $R$ with a root $r \in R$ and a \emph{level} function $l : R \rightarrow \Z$ satisfying the conditions below.
	\smallskip
	
	\noindent
	(\ref{dfn:cover_tree_compressed}a)
	\emph{Root condition} :  
 	the level of the root node $r$ satisfies $l(r) \geq 1 +  \max\limits_{p \in R \setminus \{r\}}l(p)$.
	\smallskip
	
	\noindent
	(\ref{dfn:cover_tree_compressed}b)
	\emph{Cover condition} : 
	for every node $q \in R\setminus \{r\}$, we select a unique \emph{parent} $p$ and a level $l(q)$ such that $d(q,p) \leq 2^{l(q)+1}$ and $l(q) < l(p)$; 
	this parent node $p$ has a single link to its  \emph{child} node $q$. 
	\smallskip
	
	\noindent
	(\ref{dfn:cover_tree_compressed}c)
	\emph{Separation condition} : 
	for $i \in \Z$, the \emph{cover set} 
	$C_i = \{p \in R \mid l(p) \geq i\}$ has
	$d_{\min}(C_i) = \min\limits_{p \in C_{i}}\min\limits_{q \in C_{i}\setminus \{p\}} d(p,q) > 2^{i}$.
	\medskip
	
	\noindent
	Since there is a 1-1 map between $R$ and all nodes of $\T(R)$, the same notation $p$ can refer to a point in the set $R$ or to a node of the tree $\T(R)$.  
	Set $l_{\max} = 1 +  \max\limits_{p \in R \setminus \{r\}}l(p) $ and $l_{\min} = \min\limits_{p \in R}l(p)$.
	For any node $p\in\T(R)$, $\Child(p)$ denotes the set of all children of $p$, including $p$ itself.
	For any node $p \in\T(R)$, define the \emph{node-to-root} path as a unique sequence of nodes $w_0,\dots,w_m$ such that $w_0 = p$, $w_m$ is the root and $w_{j+1}$ is the parent of $w_{j}$ for $j=0,...,m-1$. 
\smallskip

	A node $q \in\T(R)$ is a \emph{descendant} of a node $p$ if $p$ is in the node-to-root path of $q$. 
	A node $p$ is an \emph{ancestor} of $q$ if $q$ is in the node-to-root path of $p$. 
	Let $\Desc(p)$ be the set of all descendants of $p$, including itself $p$. 
\end{restatable}

Lemma~\ref{lem:packing} links the minimized expansion constant with the doubling dimension. This result is used in the proofs of the width bound of a compressed cover tree in Lemma~\ref{lem:compressed_cover_tree_width_bound}, also for the time complexity of a compressed cover tree construction in Lemma~\ref{lem:general_construction_time},  and for the $k$-nearest neighbor search in Lemma~\ref{lem:knn:time}.
All hyperlinks are clickable.

\begin{restatable}[packing]{lem}{lempacking}\label{lem:packing}
Let $S$ be a finite $\delta$-sparse set in a metric space $(X,d)$, so $d(a,b) > \delta$ for all $a,b \in S$. 
	Then, for any point $ p \in X$ and any radius $t > \delta$, we have
	$|\bar{B}(p, t)  \cap S | \leq (c_m(S))^{\mu}$, where $\mu = \lceil \log_2(\frac{4t}{\delta} + 1) \rceil $.
\end{restatable}

Proof of Lemma~\ref{lem:packing} is in \hyperlink{proof:lem:packing}{Appendix~\ref*{sec:cover_tree}}. 
\smallskip

\noindent
Lemma~\ref{lem:compressed_cover_tree_width_bound} shows that the number of children of any node of a compressed cover tree on any specific level can be bounded by using minimized expansion constant $c_m(R)$. 
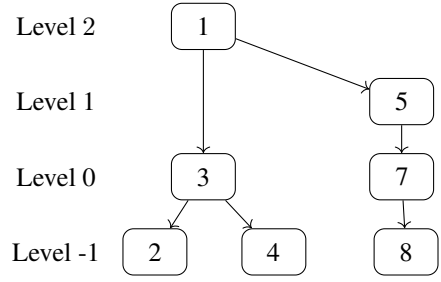
\begin{figure*}
	\centering
	\begin{tikzpicture}[align=center, node distance = 1.0cm, scale = 0.45]

	\node (scalem2text) {Level $2$};
	\node[below of =scalem2text] (scalem1text) {Level 1};
	\node[below of =scalem1text] (scalem0text) {Level 0};
	\node[below of =scalem0text] (scalemm1text) {Level -1};

	\node [blockz,  right=25pt of scalem2text ] (node1) {1};
	\node [blockz,  right=100pt of scalem1text] (node5) {5};
	\node [blockz,  right=25pt of scalem0text] (node3) {3};
	\node [blockz,  right=5pt of scalemm1text] (node2) {2};
	\node [blockz,  right=50pt of scalemm1text] (node4) {4};
	
	\node [blockz,  right=100pt of scalem0text] (node7) {7};
	\node [blockz,  right=100pt of scalemm1text] (node8) {8};

	  \draw[->] (node1) -> (node5);
	  \draw[->] (node1) -> (node3);
	  \draw[->] (node3) -> (node2);
	  \draw[->] (node3) -> (node4);
	  \draw[->] (node5) -> (node7);
	  \draw[->] (node7) -> (node8);

\end{tikzpicture} 
	\caption{Consider a compressed cover tree $\T(R)$ that was built on set $R = \{1,2,3,4,5,7,8\}$. Let $\Sd_i(p, \T(R))$ be a distinctive descendant set of Definition \ref{dfn:distinctive_descendant_set}. Then $V_2(1) = \emptyset, V_{1}(1) = \{5\}$ and $V_{0}(1) = \{3,5,7\}$.
		And also $\Sd_2(1, \T(R)) = \{1, 2,3,4,5,7,8\}$, $\Sd_1(1, \T(R)) = \{1,2,3,4\} $ and $\Sd_{0}(1, \T(R)) = \{1\} $.}
	\label{fig:uniqueDescendant_paper}
\end{figure*}
\begin{restatable}[width bound]{lem}{lemwidthbound}\label{lem:compressed_cover_tree_width_bound}
Let $R$ be a finite subset of a metric space $(X,d)$. 
	For any compressed cover tree $\T(R)$, any node $p$ and any level $i \leq l(p)$ we have  
	$$\{q \in \Child(p) \mid l(q) = i\} \cup \{p\} \leq (c_m(R))^4, $$ where $c_m(R)$ is the minimized expansion constant of $R$.
\end{restatable}


\smallskip
\noindent 
Lemma~\ref{lem:growth_bound} is an important property of the expansion constant, which allows us to calculate the low-bound of the number of points in the larger ball $\bar{B}(q, 4r)$ for any node $q \in R$ and radius $r \in \R_{+}$ using the smaller ball $\bar{B}(q, r)$ and the expansion constant $c(R)$, if there exists a point $p \in R$ which is located in annulus $2r < d(p,q) \leq 3r$. 


\begin{restatable}[growth bound]{lem}{lemgrowthbound}\label{lem:growth_bound}
	Let $(A,d)$ be a finite metric space, let $q \in A$ be an arbitrary point and let $r \in \R$ be a real number.
	Let $c(A)$ be the expansion constant from Definition \ref{dfn:expansion_constant}.
	If there exists a point $p \in A$ such that $2r < d(p,q) \leq 3r$, then $|\bar{B}(q, 4r)|  \geq  (1+\frac{1}{c(A)^2}) \cdot |\bar{B}(q, r)|$. 
\end{restatable}



Lemma~\ref{lem:growth_bound_extension} is a generalization of Lemma~\ref{lem:growth_bound} and will be used to estimate the number of iterations in compressed cover tree construction algorithm, Lemma~\ref{lem:construction_depth_bound} and in the $k$-nearest neighbors algorithm, Lemma~\ref{lem:knn_depth_bound}.

\begin{restatable}[extended growth bound]{lem}{lemgrowthboundextension}\label{lem:growth_bound_extension}
	Let $(A,d)$ be a finite metric space, let $q \in A$ be an arbitrary point. Let $p_1, ..., p_n$ be a sequence of distinct points in $R$, in such a way that  for all $i \in \{2,...,n\}$ we have $4 \cdot d(p_i, q) \leq d(p_{i+1},q)$.
	Then $$|\bar{B}(q,\frac{4}{3} \cdot d(q,p_n))| \geq (1+\frac{1}{c(A)^2})^{n} \cdot |\bar{B}(q,\frac{1}{3} \cdot d(q,p_1))|.$$
\end{restatable}


\begin{restatable}[the height of a compressed cover tree]{dfn}{dfndepth}\label{dfn:depth}
	For a compressed cover tree $\T(R)$ on a finite set $R$,
	the \emph{height set} is $H(\T(R))=\{ i \mid C_{i-1}\neq C_{i}\}\cup \{l_{\max},l_{\min}\}$.
	The size $|H(\T(R))|$ of this set is called the \emph{height} of $\T(R)$.
\end{restatable}
\begin{lem}
	\label{lem:text:depth_bound}
	Any finite set $R$ has the upper bound $|H(\T(R))|\leq 1+\log_2(\Delta(R))$.
\end{lem}


Intuitively $\Sd_i(p, \T(R))$ denotes all the descendants of pair $(p,i)$ in the explicit or implicit cover tree.

\begin{restatable}[Distinctive descendant sets]{dfn}{dfndistinctivedescendantset}\label{dfn:distinctive_descendant_set}
Let $R\subseteq X$ be a finite reference set with a compressed cover tree $\T(R)$. 
	For any node $p \in \T(R)$ and level $i \leq l(p) - 1$, set
	$V_{i}(p) =  \{u \in \Desc(p) \mid i \leq l(u)\leq l(p) - 1\}.$
	If $i \geq l(p)$, then set $V_i(p) = \emptyset$. 
	For any level $i \leq l(p) $, the \emph{distinctive descendant set} is
	$\Sd_i(p, \T(R)) =   \Desc(p) \setminus \bigcup\limits_{u \in V_{i}(p)} \Desc(u)$ and has the size $|\Sd_i(p, \T(R)) |$. 
\end{restatable}
 Lemma~\ref{lem:distinctive_descendant_child_level} shows that if $q \in \Sd_i(p, \T(R))$ then there is a node-to-node path $q = a_0 , ..., a_{m} = p$,
 so that $l(a_{m-1}) \leq i - 1$.

\begin{restatable}{lem}{lemdistinctivedescendantchildlevel}\label{lem:distinctive_descendant_child_level}
Let $R\subseteq X$ be a finite reference set with a cover tree $\T(R)$. 
	In the notations of Definition~\ref{dfn:distinctive_descendant_set}, let $p \in \T(R)$ be any node. 
	If $w \in \Sd_i(p,\T(R))$ then either $w = p$ or there exists $a \in \Child(p) \setminus \{p\}$ such that $l(a) < i$ and $w \in \Desc(a)$.
\end{restatable}

Definition~\ref{dfn:implementation_compressed_cover_tree} explains the concrete implementation of a compressed cover tree.

\begin{restatable}[$\Child(p,i)$ and $\nxt(p,i,\T(R))$]{dfn} {dfnimplementation}\label{dfn:implementation_compressed_cover_tree}
	In a compressed cover tree $\T(R)$ on a set $R$, for any level $i$ and a node $p \in R$, set $\Child(p,i) = \{ a \in \Child(p) \mid l(a) = i \}$. 
	Let $\nxt(p,i,\T(R))$ be the maximal level $j$ satisfying $j < i$ and $\Child(p,i) \neq \emptyset$. 
	If such level does not exist, we set $j = l_{\min}(\T(R)) - 1$.
	For every node $p$, we store its set of children in a linked hash map so that
	 \begin{enumerate}[label=(\alph*)]	
	\item any key $i$ gives access to $\Child(p,i)$, 
	\item  $\Child(p,i)$  $\rightarrow$ $\Child(p,\nxt(p,i, \T(R)))$,	
	\item we can directly access $\max \{j \mid \Child(p,j) \neq \emptyset\}$.
	\end{enumerate}
\end{restatable}

\section{\capitalisewords{Construction of a compressed cover tree}}
\label{sec:text:ConstructionCovertree}

This section discusses a construction of a compressed cover tree. 
New Algorithm~\ref{alg:text:cover_tree_k-nearest_construction_whole} builds a compressed cover tree by using the  Insert() method from \citet[Algorithm~2]{beygelzimer2006cover}, which was specifically adapted for a compressed cover tree, see details in Appendix~\ref{sec:ConstructionCovertree}. 
\smallskip

The proof of  \citet[Theorem~6]{beygelzimer2006cover}, which estimated the time complexity of \citet[Algorithm~2]{beygelzimer2006cover}, was shown to be incorrect by \citet[Counterexample~4.2]{elkin2022counterexamples}. 
The main contribution of this section estimate  the time complexity of Algorithm \ref{alg:text:cover_tree_k-nearest_construction_whole}:
\begin{itemize}
    \item Theorem \ref{thm:construction_time} bounds the time complexity as 
    $O(c_m(R)^{10} \cdot \log_2(\Delta(R)) \cdot |R|)$ via the minimized expansion constant $c_m(R)$ and the aspect ratio $\Delta(R)$.
    \item  Theorem \ref{thm:construction_time_KR} bounds the time complexity as $O(c(R)^{12} \cdot \log_2|R| \cdot |R|)$ via the expansion constant $c(R)$.
\end{itemize}

\begin{dfn}[construction iteration set $L(\T(W),p)$]
	\label{dfn:text:cover_tree_construction_iteration_set}
	Let $W$ be a finite subset of a metric space $(X,d)$. 
	Let $\T(W)$ be a cover tree of Definition \ref{dfn:cover_tree_compressed} built on $W$ and let $p \in X \setminus W$ be an arbitrary point.
	Let $L(\T(W),p)$ be the set of all levels $i$ during iterations \ref{line:text:cof:loop_start}-\ref{line:text:cof:loop_end} of Algorithm \ref{alg:text:cover_tree_k-nearest_construction} launched with the inputs 
	$\T(W),p$. 
	   We set  $$\eta(i) = \min \{ t \in L(\T(W),p) \mid t > i\}.$$
\end{dfn}

\begin{restatable}[correctness of Algorithm \ref{alg:text:cover_tree_k-nearest_construction_whole}]{thmm} {thmconstructioncorrectness}\label{thm:construction_correctness}
	Algorithm~\ref{alg:text:cover_tree_k-nearest_construction_whole} builds a compressed cover tree in Definition~\ref{dfn:cover_tree_compressed}.
\end{restatable}

\begin{restatable}[time complexity of a key step for $\T(R)$]{lem}{lemgeneralconstructiontime}\label{lem:general_construction_time}
	Arbitrarily order all points of a finite reference set $R$ in a metric space $(X,d)$ starting from the root: $r=p_1$, $p_2,\dots,p_{|R|}$.
	Set $W_1=\{r\}$ and $W_{y+1}=W_y\cup\{p_y\}$ for  
	$y=1,...,|R|-1$. 
	Then Algorithm~\ref{alg:text:cover_tree_k-nearest_construction_whole} builds a compressed cover tree $\T(R)$ in time $$O\Big((c_m(R))^8 \cdot \max\limits_{y=1,...,|R|-1}L(\T(W_{y}),p_{y}) \cdot |R|\Big),$$
	where $c_m(R)$ is the minimized expansion constant from Definition \ref{dfn:expansion_constant}.
\end{restatable}

\begin{algorithm}
	\caption{Building a compressed cover tree $\T(R)$ from Definition \ref{dfn:cover_tree_compressed}.
	}
	\label{alg:text:cover_tree_k-nearest_construction_whole}
	\begin{algorithmic}[1]
		\STATE \textbf{Input} : a finite subset $R$ of $(X,d)$, root $r \in R$
		\STATE \textbf{Output} : a compressed cover tree $\T(R)$. 
		\STATE Build the initial compressed cover tree $\T = \T(\{r\})$ consisting of the root node $r$ by setting $l(r) = +\infty$. 
		\FOR{$p \in R \setminus \{r\}$} \alglinelabel{line:text:con:for:begin}
		\STATE $\T \leftarrow $ run AddPoint$(\T , p )$, Algorithm \ref{alg:text:cover_tree_k-nearest_construction}.
		\ENDFOR  \alglinelabel{line:text:con:for:end}
		\STATE For the root $r$ of $\T$ set $l(r) = 1 +  \max_{p \in R \setminus \{r\}}l(p)$
	\end{algorithmic}
\end{algorithm}

\begin{algorithm}
	\caption{Building $\T(W \cup \{p\})$ in 
		lines \ref{line:text:con:for:begin}-\ref{line:text:con:for:end} of Algorithm \ref{alg:text:cover_tree_k-nearest_construction_whole}.}
	\label{alg:text:cover_tree_k-nearest_construction}
	\begin{algorithmic}[1]
		\STATE \textbf{Function} AddPoint(a compressed cover tree $\T(W)$ with a root $r$, a point $p\in X$)
		\STATE \textbf{Output} : compressed cover tree $\T(W \cup \{p\})$. 
		\STATE Set $i \leftarrow l_{\max}(\T(W)) - 1$ and $\eta(l_{\max} - 1) = l_{\max}$ \\
		 \COMMENT{If the root $r$ has no children then $ i \leftarrow -\infty$}
		\STATE Set $R_{l_{\max}} \leftarrow \{r\}$ 
        \STATE initialize sorted dictionary $M$ with $M[l_{\max}] = \{r\}$
		\WHILE{$i \geq l_{\min}$} \alglinelabel{line:text:cof:loop_start}
		 \STATE $V = \cup_{q \in R_{\eta(i)}} \{a \in \Child(q) \mid l(a) = i \} $.
		\STATE Assign $\mathcal{C}_i(R_{\eta(i)}) \leftarrow  R_{\eta(i)} \cup V$.
		\alglinelabel{line:text:cof:dfn_C}
		\STATE Set $R_{i} = \{a \in \C_i(R_{\eta(i)}) \mid d(p,a) \leq 2^{i+1} \}$ 
        \STATE Assign $M[i] = R_{i}$.
		\alglinelabel{line:text:cof:defRim1}
		\IF {$R_i$ is empty}\alglinelabel{line:text:cof:inner_loop:begin}
		\STATE  Launch Algorithm \ref{alg:text:construction_parent_assign} with parameters $(p, M)$ and \textbf{exit this algorithm}.\alglinelabel{line:text:cof:inner_loop:mid}
		\ENDIF \alglinelabel{line:text:cof:inner_loop:end}
		\STATE $t = \max_{ a \in R_{i}} \nxt(a,i,\T(W)) $  \alglinelabel{line:text:cof:dfn_t} \\
		\COMMENT{If $R_{i}$ has no children we set $t = l_{\min} - 1$}
		\STATE $\eta(i) \leftarrow i$ and $i \leftarrow t$ 
		\ENDWHILE \alglinelabel{line:text:cof:loop_end}
        \STATE Launch Algorithm \ref{alg:text:construction_parent_assign} with parameters $(p, M)$. \alglinelabel{line:text:cof:selectParentEnd}
        
	\end{algorithmic}
\end{algorithm}

\begin{algorithm}
\caption{Assign node subprocedure}
\label{alg:text:construction_parent_assign}
\begin{algorithmic}[1]
\STATE \textbf{Function} AssignParent(Point $p$,  dictionary $M$)
\STATE \textbf{Output:} Compressed cover tree $\T(W \cup \{p\})$
\STATE Set $i$ to be the lowest key of $M$.
\WHILE{$i \leq l_{max}$}
\IF {$d(p,R_i) \leq 2^{i}$}\alglinelabel{line:text:assignparent:assgiment:ifstatement}
\STATE Let $q \in R_{i}$ such that $d(q,p) = d(R_{i},p)$, let $x$ maximal integer for which $d(p,q) > 2^{x}$. \alglinelabel{line:text:assignparent:assgiment:first}
 \STATE Set $l(p) = x$ and $q$ to be the parent of $p$.  \alglinelabel{line:text:assignparent:assgiment}
\ENDIF
\STATE Find next key $j > i$ of $M$ and set $i = j$
\ENDWHILE
\end{algorithmic}
\end{algorithm}

\begin{restatable}[time complexity of $\T(R)$ via aspect ratio]{thmm}{thmconstructiontime}\label{thm:construction_time}
	Let $R$ be a finite subset of a metric space $(X,d)$ having the aspect ratio $\Delta(R)$.
	Algorithm~\ref{alg:text:cover_tree_k-nearest_construction_whole} builds
	a compressed cover tree $\T(R)$ in time $O((c_m(R))^8 \cdot \log_2(\Delta(R)) \cdot |R|),$
	where $c_m(R)$ is the minimized expansion constant from Definition~\ref{dfn:expansion_constant}.
\end{restatable}
\begin{proof}
In Lemma \ref{lem:general_construction_time}, use the bounds from Lemma \ref{lem:text:depth_bound}:
$$\max\limits_{y=2,...,|R|}|L(\T(W_{y-1}),p_{y})| \leq H(\T(R))\leq 1 + \log_2\Delta(R).$$
\end{proof}

\begin{restatable}{lem}{lemknnnextlevelfinder}\label{lem:knn_next_level_finder_for_log_depth}
	Let $(X,d)$ be a metric space and let $W \subseteq X$ be its finite subset. Let $q \in X \setminus W$ be an arbitrary point.
	Let $i \in L(\T(W),q)$ be arbitrarily iteration of Definition \ref{dfn:text:cover_tree_construction_iteration_set}. Assume that $t = \eta(\eta(i+1))$ is defined. Then there exists $p \in W$ satisfying $2^{i+1} < d(p,q) \leq 2^{t+1}$.
\end{restatable}

\begin{restatable}[Construction iteration bound]{lem}{lemconstructiondepthbound}\label{lem:construction_depth_bound}
	Let $A, W$ be finite subsets of a metric space $X$ satisfying $W \subseteq A \subseteq X$.
Take a point $q \in A \setminus W$. 
Given a compressed cover tree $\T(W)$ on $W$, Algorithm~\ref{alg:text:cover_tree_k-nearest_construction} runs lines \ref{line:text:cof:loop_start}-\ref{line:text:cof:loop_end} this number of times: $|L(\T(W),q)| = O\big(c(A)^2 \cdot \log_2(|A|)\big)$.
\end{restatable}
\begin{hproof}
Assume that Algorithm~\ref{alg:text:cover_tree_k-nearest_construction} was launched with parameters $(q, \T(W))$
Lemma~\ref{lem:knn_next_level_finder_for_log_depth} showed that for any iterations $i \in L(\T(W),q)$, if $t = \eta(\eta(i+1))$ exists, then there exists $p \in W$ which belongs to annulus $\bar{B}(q,2^{t+1}) \setminus \bar{B}(q, 2^{i+1})$.
We can select a subsequence $S$ of iterations $L(\T(W),q)$, in such a way that for every $i \in S$ there exists point $p_i \in \bar{B}(q,2^{t+1}) \setminus \bar{B}(q, 2^{i+1})$. It can be shown that the size of $S$ selected this way is $12 \cdot |S|\geq |L(\T(W),q) |$
\smallskip 

Denote by $P = (p_1, ..., p_n)$ the sequence of points $p_i$ obtained from $S$. Using Lemma~\ref{lem:growth_bound_extension} we obtain $$|\bar{B}(q,\frac{4}{3}  d(q,p_n))| \geq (1+\frac{1}{c(R)^2})^{n} \cdot |\bar{B}(q,\frac{1}{3} d(q,p_1))|$$
which can be written as $$|A| \geq \frac{|\bar{B}(q,\frac{4}{3} \cdot d(q,p_n))|}{|\bar{B}(q,\frac{1}{3} \cdot d(q,p_1))|} \geq (1+\frac{1}{c(A)^2})^{|S|}$$
Lemma~\ref{lem:hard_function_bound} gives $c(A)^2\log(A) \geq |S|$. Combining this with the fact that $12 \cdot |S|\geq |L(\T(W),q) |$ we finally conclude that $|L(\T(W),q) | \leq 12 \cdot c(A)^2 \cdot \log_2(|A|)$.
\end{hproof}

\begin{restatable}[time for $\T(R)$ via expansion constants]{thmm}{thmconstructiontimeKR}\label{thm:construction_time_KR}
	Let $R$ be a finite subset of a metric space $(X,d)$. Let $A$ be a finite subset of $X$ satisfying $R \subseteq A \subseteq X$. 
	Then Algorithm~\ref{alg:text:cover_tree_k-nearest_construction_whole} builds a compressed cover tree $\T(R)$ in time $O((c_m(R))^8 \cdot c(A)^2 \cdot \log_2(|A|) \cdot |R|)$, see the expansion constants $c(A),c_m(R)$ in Definition \ref{dfn:expansion_constant}.
\end{restatable}
\begin{proof}
	It follows from Lemmas~\ref{lem:construction_depth_bound} and~\ref{lem:general_construction_time}.
\end{proof}


\begin{restatable}{cor}{corconstructiontimeKR}
\label{cor:construction_time_KR}
    Let $R$ be a finite subset of a metric space $(X,d)$. 
	Then Algorithm~\ref{alg:text:cover_tree_k-nearest_construction_whole} builds
	a compressed cover tree $\T(R)$ in time $O((c_m(R))^8 \cdot c(R)^2 \cdot \log_2(|R|)) \cdot |R|),$
	where the constants $c(R),c_m(R)$ appeared in Definition \ref{dfn:expansion_constant}.
\end{restatable}
\begin{proof}
	In Theorem~\ref{thm:construction_time_KR} set $A = R$.
\end{proof}

\section{\capitalisewords{New $k$-nearest neighbor search algorithm}}
\label{sec:text:better_approach_knn_problem}

This section is motivated by
\citet[Counterexample~5.2]{elkin2022counterexamples}, which showed that the proof of past time complexity claim in  \citet[Theorem~5]{beygelzimer2006cover} for the nearest neighbor search algorithm contained gaps.  For extra details and all proofs, see 
Appendix~\ref{sec:better_approach_knn_problem}.
 \smallskip

 \noindent
The gaps are filled by new Algorithm \ref{alg:text:cover_tree_k-nearest} for all $k$-nearest neighbors, which generalizes and improves the original method 
in \citet[Algorithm~1]{beygelzimer2006cover}. 
\smallskip

The first improvement is the $\lambda$-point in line~\ref{line:text:knnu:dfnLambda}, which helps find all $k$-nearest neighbors of a given query point for any $k \geq 1$. 
The second improvement is a new break condition for the loop in line \ref{line:text:knnu:qtoofar:condition}. 
This condition is used in the proof of Lemma \ref{lem:knn_depth_bound} to conclude that the total number of performed iterations is bounded by  $O(c(R)^2\log(|R|))$  during the whole run-time of the algorithm. 
\smallskip

The latter improvement corrects the past gap in proof of \citet[Theorem~5]{beygelzimer2006cover} by bounding the number of iterations independently from the explicit depth \citet[Definition~3.2]{elkin2022counterexamples}.  

 \smallskip

\noindent 
Assuming that we have already constructed a compressed cover tree on a reference set $R$, the two main results estimate the time complexity of a new $k$-nearest neighbor method in Algorithm~\ref{alg:text:cover_tree_k-nearest}m which finds all $k$-nearest neighbors of any query point $q \in X$ in a reference set $R \subseteq X$ as follows:

\begin{itemize}
\item Corollary~\ref{cor:cover_tree_knn_miniziminzed_constant_time} bounds the time complexity as $O\Big ( \log_2(k) \cdot (\log_2(\Delta(R)) + |\bar{B}(q, 5d_k(q,R))|) \Big ),$ where $\Delta(R)$ is the aspect ratio and $c_m(R)$ is considered fixed (hence hidden).
\item Theorem \ref{thm:knn_KR_time} bounds the time complexity as $O\Big (\log_2(k) \cdot \big(\log_2(|R|) + k\big) \Big)$, where the expansion constant $c(R \cup \{q\})$ is considered  fixed (hence hidden). 
\end{itemize}

 \begin{dfn}[$\la$-point]
	\label{dfn:text:lambda-point}
	Fix a query point $q$ in a metric space $(X,d)$ and fix any level $i \in \Z$. 
	Let $\T(R)$ be its compressed cover tree on a finite reference set $R \subseteq X$. 
	Let $C$ be a subset of a cover set $C_i$ from Definition~\ref{dfn:cover_tree_compressed} satisfying $\sum\limits_{p \in C}|\Sd_i(p, \T(R))| \geq k$, where $\Sd_i(p, \T(R))$ is the distinctive descendant set from Definition \ref{dfn:distinctive_descendant_set}.
	For any $k\geq 1$, define $\la_k(q,C)$ as a point $\la\in C$ that minimizes $d(q,\la)$ subject to $\sum\limits_{p \in N(q;\la)}|\Sd_i(p, \T(R)) |\geq k$. 
\end{dfn}

\begin{dfn}
	\label{dfn:text:knn_iteration_set}
	Let $R$ be a finite subset of a metric space $(X,d)$. 
	Let $\T(R)$ be a cover tree of Definition \ref{dfn:cover_tree_compressed} built on $R$ and let $q \in X$ be arbitrary point.
	Let $L(\T(R),q)$ be the set of all levels $i$ during iterations of lines~\ref{line:text:knnu:loop_begin}-\ref{line:text:knnu:loop_end} of Algorithm~\ref{alg:text:cover_tree_k-nearest} launched with inputs 
	$\T(R),q$. 
	If Algorithm~\ref{alg:text:cover_tree_k-nearest} reaches line \ref{line:text:knnu:qtoofar} at a level 
	$\varrho \in L(\T(R),q)$, then we say that $\varrho$ is \emph{special}. 
	Set $\eta(i) = \min \{ t \in L(\T(R),q) \mid t > i\}$. 
\end{dfn}

\begin{algorithm}
	\caption{$k$-nearest neighbor search by a compressed cover tree}
	\label{alg:text:cover_tree_k-nearest}
	\begin{algorithmic}[1]
		\STATE \textbf{Input} : compressed cover tree $\T(R)$, a query point $q\in X$, an integer $ k \in \Z_{+} $
		\STATE Set $i \leftarrow l_{\max}(\T(R)) - 1$ and $\eta(l_{\max}-1) = l_{\max}$
		\STATE  Let $r$ be the root node of $\T(R)$. Set $R_{l_{\max}}=\{r\}$.
		\WHILE{$i \geq l_{\min}$} \alglinelabel{line:text:knnu:loop_begin}
		\STATE $V = \cup_{q \in R_{\eta(i)}} \{a \in \Child(q) \mid l(a) = i \} $.
		\STATE Assign $\mathcal{C}_i(R_{\eta(i)}) \leftarrow  R_{\eta(i)} \cup V$. \alglinelabel{line:text:knnu:dfn_C}
		\STATE Compute $\lambda = \lambda_k(q,\C_{i}(R_{\eta(i)}))$ \alglinelabel{line:text:knnu:dfnLambda} by 
		Algorithm \ref{alg:lambda}.
		\STATE $R_{i} = \{p \in \C_i(R_{\eta(i)}) \mid d(q,p) \leq d(q,\lambda) + 2^{i+2}\}$ \alglinelabel{line:text:knnu:dfnRi}
		\IF {$d(q,\lambda) > 2^{i+2}$} \alglinelabel{line:text:knnu:qtoofar:condition}
		\STATE Collect the distinctive descendants $\Sd_i(p,\T(R))$ of all points $p \in R$ in set S, see Algorithm~\ref{alg:cover_tree_k-nearest_final_collection}.
		\STATE Compute and \textbf{output} $k$-nearest neighbors of the query point $q$ from set $S$.
		\alglinelabel{line:text:knnu:qtoofar}
		\ENDIF \alglinelabel{line:text:knnu:qtoofar:condition:endif}
		\STATE Set $j \leftarrow \max_{ a \in R_{i}} \nxt(a,i,\T(R))$ \\
		\COMMENT{If such $j$ is undefined, we set $j = l_{\min}-1$} \alglinelabel{line:text:knnu:dfnindexj}
		\STATE Set $\eta(j) \leftarrow i$ and $i \leftarrow j$.
		\ENDWHILE \alglinelabel{line:text:knnu:loop_end}
		\STATE Compute and \textbf{output} $k$-nearest neighbors of query point $q$ from the set $R_{l_{\min}}$.
		\alglinelabel{line:text:knnu:final_line}
	\end{algorithmic}
\end{algorithm}

\begin{restatable}[correctness of Algorithm~\ref{alg:text:cover_tree_k-nearest}]{thmm} {thmknncorrectness}\label{thm:cover_tree_knn_correct}
		Algorithm~\ref{alg:text:cover_tree_k-nearest} correctly finds all $k$-nearest neighbors of query point $q$ within the reference set $R$. 
\end{restatable}

\begin{restatable}{lem}{lemknntime}\label{lem:knn:time}
		Algorithm~\ref{alg:text:cover_tree_k-nearest} has
	the following time complexities of its lines
	\smallskip
	
	\noindent
	(a) 
	$\max\{\li{\ref{line:text:knnu:loop_begin}-\ref{line:text:knnu:qtoofar:condition}}, \li{\ref{line:text:knnu:qtoofar:condition:endif}-\ref{line:text:knnu:loop_end}} , \li{\ref{line:text:knnu:final_line}}\} = O\big(c_m(R)^{10} \cdot \log_2(k)\big)$;
	\smallskip
	
	\noindent
	(b) 
	$\li{\ref{line:knnu:qtoofar:condition}-\ref{line:knnu:qtoofar:condition:endif} } = O\big(|\bar{B}(q, 5 d_k(q,R))| \cdot \log_2(k)\big).$
\end{restatable}

	
	


\begin{restatable}{thmm}{thmgeneraltime}\label{thm:cover_tree_knn_general_time}
		Let $R$ be a finite set in a metric space $(X,d)$, $c_m(R)$ be the minimized constant from Definition \ref{dfn:expansion_constant}.
	Given a compressed cover tree $\T(R)$, Algorithm~\ref{alg:text:cover_tree_k-nearest} finds all $k$-nearest neighbors of a query point $q\in X$ in time 
	$$O\Big ( \log_2(k) \cdot ((c_m(R))^{10}  \cdot |L(q,\T(R))| + |\bar{B}(q, 5 d_k(q,R)) |)\Big ),$$
	where $L(\T(R),q)$ is the set of all performer iterations (lines~\ref{line:text:knnu:loop_begin}-\ref{line:text:knnu:loop_end} ) of Algorithm~\ref{alg:text:cover_tree_k-nearest}. 
\end{restatable}
\begin{proof}
	Apply Lemma~\ref{lem:knn:time} to estimate the time complexity of Algorithm~\ref{alg:text:cover_tree_k-nearest}: \\ 
	$O\big( |L(\T(R),q)| \cdot 
	(\li{\ref{line:text:knnu:loop_begin}-\ref{line:text:knnu:qtoofar:condition}}
	+ \li{\ref{line:text:knnu:qtoofar:condition:endif}-\ref{line:text:knnu:loop_end}}  
	+ \li{\ref{line:text:knnu:final_line}}) 
	+\li{\ref{line:text:knnu:qtoofar:condition}-\ref{line:text:knnu:qtoofar:condition:endif} }\big)$.
\end{proof}
Corollary \ref{cor:cover_tree_knn_miniziminzed_constant_time} gives a run-time bound using only minimized expansion constant $c_m(R)$,
where if $R \subset \R^{m}$, then $c_m(R) \leq 2^{m}$. Recall that $\Delta(R)$ is the aspect ratio of $R$ introduced in 
Definition \ref{dfn:radius+d_min}.


\begin{restatable}{cor}{corminimizedexpansionconstanttime}\label{cor:cover_tree_knn_miniziminzed_constant_time}
		Let $R$ be a finite set in a metric space $(X,d)$. 
	Given a compressed cover tree $\T(R)$, Algorithm~\ref{alg:text:cover_tree_k-nearest} finds all $k$-nearest neighbors of $q$ in time  $O\Big ((c_m(R))^{10} \cdot \log_2(k) \cdot \log_2(\Delta(R)) + |\bar{B}(q, 5d_k(q,R))| \cdot \log_2(k) \Big ).$
\end{restatable}

\begin{proof}
	Replace $|L(q,\T(R))|$ in the time complexity of Theorem \ref{thm:cover_tree_knn_general_time} by its upper bound in Lemma \ref{lem:text:depth_bound}: $|L(q,\T(R))| \leq |H(\T(R))| \leq \log_2(\Delta(R)).$
\end{proof}

Lemma~\ref{lem:knn_depth_bound} is proved similarly to Lemma~\ref{lem:construction_depth_bound}.
For full details see \hyperlink{proof:lem:knn_depth_bound}{Appendix~\ref*{sec:approxknearestneighbor}}.

\begin{restatable}{lem}{lemknndepthbound}\label{lem:knn_depth_bound}
		Algorithm \ref{alg:text:cover_tree_k-nearest} executes lines \ref{line:text:knnu:loop_begin}-\ref{line:text:knnu:loop_end} the following number of times: $|L(\T(R),q)| = O(c(R \cup \{q\})^2 \cdot \log_2(|R|))$.
\end{restatable}

\begin{restatable}{thmm}{thmknnkrtime}\label{thm:knn_KR_time}
	Let $R$ be a finite reference set in a metric space $(X,d)$. Let $q\in X$ be a query point, $c(R \cup \{q\})$ be the expansion constant of $R \cup \{q\}$ and $c_m(R)$ be the minimized expansion constant from Definition \ref{dfn:expansion_constant}. Given a compressed cover tree $\T(R)$, Algorithm~\ref{alg:text:cover_tree_k-nearest} finds all $k$-nearest neighbors of $q$ in time 
	$O\Big ( c(R \cup \{q\})^2 \cdot \log_2(k) \cdot \big((c_m(R))^{10}  \cdot \log_2(|R|) + c(R \cup \{q\}) \cdot k\big) \Big).$
\end{restatable}

\begin{proof}
	By Theorem~\ref{thm:cover_tree_knn_general_time} the required time complexity is
	$O\Big ((c_m(R))^{10} \cdot \log_2 (k) \cdot |L(q,\T(R))| + |\bar{B}(q, 5d(q,\beta)) | \cdot \log_2(k) \Big )$, 
	for some point $\beta$ among the first $k$-nearest neighbors of $q$.
	Apply Definition \ref{dfn:expansion_constant} to get the upper bound
	\begin{ceqn}
		\begin{align}
			|B(q,5d(q,\beta))| \leq (c(R \cup \{q\}))^3 \cdot |B(q,\frac{5}{8}d(q,\beta))|
		\end{align}
	\end{ceqn}
	Since $|B(q,\frac{5}{8}d(q,\beta))| \leq k$, we have $|B(q,5d(q,\beta))| \leq (c(R \cup \{q\}))^3  \cdot k$.
	It remains to apply Lemma \ref{lem:knn_depth_bound}: $|L(q,\T(R))| = O(c(R \cup \{q\})^2 \cdot \log_2|R|)$.
\end{proof}

\section{Discussion Of Contributions and Next Steps}
\label{sec:discussion}

This paper rigorously proved the time complexity of the exact $k$-nearest neighbor search.
The submission to ICML is strongly motivated by the past gaps in the proofs of time complexities in the highly cited \citet[Theorem~5]{beygelzimer2006cover} at ICML, \citet[Theorem~3.1]{ram2009linear} at NIPS, and \citet[Theorem~5.1]{march2010fast} at KDD.
\smallskip

Though \citet{elkin2022counterexamples} provided concrete counterexamples, no corrections were published.
Main Theorem~\ref{thm:knn_KR_time} and Corollary~\ref{cor:construction_time_KR} finally filled all the gaps.
\smallskip

Since the past obstacles were caused by unclear descriptions and missed proofs, often without pseudo-codes, this paper necessarily fills in all technical details.
Otherwise, future generations would continue citing unreliable results.
\smallskip

\noindent
To overcome the discovered challenges, first Definition~\ref{dfn:kNearestNeighbor} and Problem~\ref{pro:knn} rigorously dealt with a potential ambiguity of $k$-nearest neighbors at equal distances.
This singular case was unfortunately not discussed in the past work at all.
\smallskip

A new compressed cover tree in Definition~\ref{dfn:cover_tree_compressed} substantially simplified the navigating net \citet{krauthgamer2004navigating} and original cover tree \citet{beygelzimer2006cover} by avoiding repetitions of given data points.
This compression clarified the construction and search in Algorithms~\ref{alg:text:cover_tree_k-nearest_construction_whole} and~\ref{alg:text:cover_tree_k-nearest}. 
\smallskip

Sections~\ref{sec:text:ConstructionCovertree} and~\ref{sec:text:better_approach_knn_problem} corrected the approach of \citet{beygelzimer2006cover} as follows.
Assuming that the expansion constants and aspect ratio of a reference set $R$ are fixed, Corollaries~\ref{cor:construction_time_KR} and~\ref{thm:knn_KR_time} rigorously showed that the time complexities are linear in the maximum size of $R,Q$ and near-linear $O(k\log k)$ in the number $k$ of neighbors. 
\smallskip

The library MLpack \cite{Curtin2013a} implemented a version of an explicit cover tree, which was later defined in \citet[Counterexample~4.2]{elkin2022counterexamples}. 
The implementation of a compressed cover tree is similar but conceptually simpler due to its easier structure in Fig.~\ref{fig:text:tripleexample}. 
\smallskip

The new results justify that the MLpack implementations of the $k$-nearest neighbors search now have proved theoretical guarantees for a near-linear time complexity, which was practically important for the recent advances below.
\smallskip

Main Theorem~\ref{thm:knn_KR_time} helped justify a near-linear time complexity for several invariants based on computing $k$-nearest neighbors in a new area of \emph{Geometric Data Science}, whose aim is to build continuous geographic-style maps for moduli space of real data objects parametrized by complete invariants under practically important equivalence relations.
 
The key example is a finite cloud of unlabeled points up to isometry maintaining all inter-point distances.
The most general isometry invariant SDD (\emph{Simplexwise Distance Distribution} \cite{kurlin2023simplexwise}) is conjectured to be complete for any fintie point clouds in any metric space.

In a Euclidean space $\R^n$, the SDD was adapted to the stronger invariant SCD (\emph{Simplexwise Centered Distribution} \cite{widdowson2023recognizing}), whose completeness and polynomial complexity (in the number $m$ of points for a fixed dimension $n$) was proved in \cite{kurlin2023strength}.
\smallskip

The related and much harder problem is for periodic sets of unlabeled points, which model all solid crystalline materials (periodic crystals). 
The first generically complete invariant using $k$-nearest neighbors was the sequence of \emph{density functions} $\psi_k(t)$ measuring a fractional volume of $k$-fold intersections of balls with a variable radius $t$ and centers at all atoms of a crystal \cite{edelsbrunner2021density}.
\smallskip

These density functions have efficient algorithms in the low dimensions $n=2,3$ through higher-degree Voronoi domains \cite{smith2022practical} of periodic point sets. 

The first continuous and complete invariant for periodic point sets in $\R^n$ is the \emph{isoset} of local atomic environments up to a justified stable radius \cite{anosova2021isometry}.
The first continuous metric on isosets was introduced in \cite{anosova2022algorithms} with an approximate algorithm that has a polynomial time complexity (for a fixed dimension $n$) and a small approximation factor (about $4$ in $\R^3$).

The much faster generically complete isometry invariant for both finite and periodic sets of points is the PDD (Pointwise Distance Distribution \cite{widdowson2021pointwise}) consisting of distances to $k$ nearest neighbors per point.

The implemented search for atomic neighbors was so fast that all (more than 660 thousand) periodic crystals in the world's largest database of real materials were hierarchically compared by the PDD and its simplified version AMD (Average Minimum Distance \cite{widdowson2022average}).
\smallskip

Due to the ultra-fast running time, more than 200 billion pairwise comparisons were completed over two days on a modest desktop while past tools were estimated to require over 34 thousand years \cite{widdowson2022resolving}. 
\smallskip

The most important conclusion from the search results is the \emph{Crystal Isometry Principle} saying that any real periodic crystal has a uniquely defined location in a single continuous space of all isometry classes of periodic point sets \cite{widdowson2022resolving}.
\smallskip

This \emph{Crystal Isometry Space} contains all known and not yet discovered crystals similar to the much simpler and discrete Mendeleev's table of chemical elements. 
\smallskip

The next step is to improve the complexity of the $k$-nearest neighbor search to a purely linear time $O(c(R)^{O(1)}|R|)$ with no other extra hidden parameters by using a new compressed cover tree on both sets $Q,R$. 
\smallskip

Since a similar approach \citet{ram2009linear} was shown to have incorrect proof in \citet[Counterexample~6.5]{elkin2022counterexamples} and \citet{curtin2015plug} used some additional parameters $I, \theta$, this goal will require significantly more effort to understand if $O(c(R)^{O(1)}|R|)$ is achievable by using a compressed cover tree.




We thank all reviewers for their time and helpful suggestions.
This work was supported by the EPSRC grants EP/R018472/1, EP/X018474/1, and the Royal Academy Engineering fellowship IF2122/186 of the second author.

\bibliography{ICML2023Knn}
\bibliographystyle{icml2023}

\newpage
\appendix
\onecolumn

The appendices below contain the full version of the paper
with detailed proofs and pseudo codes

\section{The $k$-nearest neighbor search and overview of results}
\label{sec:intro_knn}

In the modern formulation, $k$-nearest neighbors problem intends to discover all $k\geq 1$ nearest neighbors in a given reference set $R$ for all points from another given query set $Q$.
Both sets belong to a common ambient space $X$ with a distance $d$ satisfying all metric axioms.
The simplest example of $X$ is $\R^n$ with the Euclidean metric. A query set $Q$ can be a single point or a subset of a larger reference set $R$. 
\medskip

\noindent
The \emph{exact} $k$-nearest neighbor problem asks for all true (non-approximate) $k$-nearest neighbors in $R$ for every query point $q\in Q$. 
Another probabilistic version of the $k$-nearest neighbor search \citet{har2006fast,manocha2007empirical} aims to find exact $k$-nearest neighbors with a given probability. 
The probabilistic $k$-nearest neighbor problem can be simplified to $k$ instances of 1-nearest-neighbors problem by splitting $R$ into $k$ subsets $R_{1}, ..., R_{k}$ and searching for nearest neighbors in each subset.
The approximate version \citet{arya1993approximate,krauthgamer2004navigating,andoni2018approximate, wang2021comprehensive} of the nearest neighbor search looks for an $\epsilon$-approximate neighbor $r\in R$ of every query point $q \in Q$ such that $d(q,r) \leq (1+\epsilon)d(q,\NN(q))$ , where $\epsilon>0$ is fixed and $\NN(q)$ is the exact first nearest neighbor of $q$. 

\medskip
\noindent
\textbf{Spacial data structures.}
It is well known that the time complexity of a brute-force approach of finding all 1st nearest neighbors of points from $Q$ within $R$ is proportional to the product $|Q|\cdot|R|$ of the sizes of $Q, R$.
Already in the 1970s real data was big enough to motivate faster algorithms and sophisticated data structures.
One of the first spacial data structures, a \emph{quadtree} \citet{finkel1974quad}, hierarchically splits a reference set $R\subset\R^2$ by subdividing its bounding box (a root) into four smaller boxes (children), which are recursively subdivided until final boxes (leaf nodes) contain only a small number of reference points.
A generalization of the quadtree to $\R^n$ exposes an exponential dependence of its computational complexity on the dimension $n$, because the $n$-dimensional box is subdivided into $2^n$ smaller boxes.
\medskip

\noindent
The first attempt to overcome this curse of dimensionality was the $kd$-tree \citet{bentley1975multidimensional} that subdivides a subset of $R$ at every step into two subsets instead of $2^n$ subsets.
Many more advanced algorithms utilizing spatial data structures have positively impacted various related research areas such as a minimum spanning tree \citet{bentley1978fast}, range search \citet{pelleg1999accelerating},  $k$-means clustering  \citet{pelleg1999accelerating}, and ray tracing \citet{fussell1988fast}.
The spacial data structures for finding nearest neighbors in the chronological order are $k$-means tree \citet{fukunaga1975branch}, $R$ tree \citet{beckmann1990r}, ball tree \citet{omohundro1989five}, $R^*$ tree \citet{beckmann1990r}, vantage-point tree \citet{yianilos1993data}, TV trees \citet{lin1994tv}, X trees \citet{berchtold1996x}, principal axis tree \citet{mcnames2001fast}, spill tree \citet{liu2004investigation}, cover tree \citet{beygelzimer2006cover}, cosine tree \citet{holmes2008quic}, max-margin tree \citet{ram2012nearest}, cone tree \citet{ram2012maximum} and others.


\dfnaspectratio*

\dfnkNearestNeighbor*

\noindent
For $Q=R=\{0,1,2,3\}$, the point $q=1$ has ordered distances $d_1=0<d_2=1=d_3<d_4=2$. 
The nearest neighbor sets are $\NN_1(1;R)=\{1\}$,
$\NN_2(1;R)=\{0,1,2\}=\NN_3(1;R)$, $\NN_4(1;R)=R$. 
So 0 can be a 2nd neighbor of 1, then 2 becomes a 3rd neighbor of 1, or these neighbors of $0$ can be found in a different order.


\proknn*


\noindent 
In a metric space, let $\bar B(p,t)$ be the closed ball with a center $p$ and a radius $t\geq 0$.
The notation $|\bar B(p,t)|$ denotes the number (if finite) of points in the closed ball. 
Definition~\ref{dfn:expansion_constant} recalls the expansion constant $c$ from \citet{beygelzimer2006cover} and introduces the new minimized expansion constant $c_m$, which is a discrete analog of the doubling dimension \citet{cole2006searching}. 

\dfnexpansionconstant*

\lemexpansionconstantproperty*
\begin{proof}
Let us first prove that $c_m(R) \leq c_m(U)$. Let $\epsilon > 0$ be arbitrary real number. 
By definition of $c_m(U)$ there exists set $\xi > 0$ and set $A$ satisfying $U \subseteq A$ for which
\begin{equation}
\label{eqa:expconstantproperty}
\sup\limits_{p \in A,t > \xi} |\dfrac{|\bar{B}(p,2t) \cap A|}{|\bar{B}(p,t) \cap A|} - c_m(U)| \leq \epsilon 
\end{equation}
Since $R \subseteq U$ we have $R \subseteq A$ therefore we can choose the same $\xi$ and set $U$ which satisfy inequality (\ref{eqa:expconstantproperty}). Therefore it follows $c_m(R) \leq c_m(U) + \epsilon$. Since $\epsilon$ was chosen arbitrarily it follows that $c_m(R) \leq c_m(U)$. 

\medskip
\noindent 

To prove that $c_m(R) \leq c(R)$, note that 
$ \sup\limits_{p \in A,t > \xi} |\dfrac{|\bar{B}(p,2t) \cap A|}{|\bar{B}(p,t) \cap A|} \leq
 \sup\limits_{p \in A,t > 0} |\dfrac{|\bar{B}(p,2t) \cap A|}{|\bar{B}(p,t) \cap A|}$. Then by choosing $\xi = \frac{d_{\min}(R)}{4}$ and $A = R$ we have:
$$c_m(R) \leq  \sup\limits_{p \in R,t > 0}|\dfrac{|\bar{B}(p,2t) \cap R|}{|\bar{B}(p,t) \cap R|} - c_m(U)| = c(R)$$

\end{proof}


\noindent
Note that both $c(R), c_m(R)$ are always defined when $R$ is finite. 
We will show that a single outlier can make the expansion constant $c(R)$ as large as $O(|R|)$.
The set $R=\{1,2,\dots,n,2n+1\}$ of $|R|=n+1$ points has $c(R)=n+1$ because $\bar B(2n+1;n)=\{2n+1\}$ is a single point, while $\bar B(2n+1;2n)=R$ is the full set of $n+1$ points. 
On the other hand the same set $R$ can be extended to a larger uniform set $A=\{1,2,\dots,2n-1,2n\}$ whose expansion constant $c(A)=2$, therefore the minimized constant of the original set $R$ becomes much smaller: $c_m(R) \leq c(A)=2<c(R)=n+1$.
\medskip

\noindent
The constant $c$ from \citet{beygelzimer2006cover} equals to $2^{\text{dim}_{KR}}$ from \citet[Section~2.1]{krauthgamer2004navigating}. 
In \citet[Section~1.1]{krauthgamer2004navigating} the doubling dimension $2^{\text{dim}}$ is defined as a minimum value $\rho$ such that any set $X$ can be covered by $2^{\rho}$ sets whose diameters are half of the diameter of $X$.
The past work \citet{krauthgamer2004navigating} proves that  $2^{\text{dim}} \leq 2^{n}$ for any subset of $\R^n$.
Theorem \ref{thm:normed_space_exp_constant} will prove that $c_m(R) \leq 2^{n}$ for any a finite subset $R\subset\R^{n}$, so $c_m(R)$ mimics $2^{\text{dim}}$.
\medskip

\noindent
\textbf{Navigating nets}.
In 2004, \citet[Theorem~2.7]{krauthgamer2004navigating} claimed that a navigating net can be constructed in time 
$O\big(2^{O(\text{dim}_{KR}(R)} |R| (\log|R|) \log(\log|R|)\big)$ and
all $k$-nearest neighbors of a query point $q$ can be found in time $O(2^{O(\text{dim}_{KR}(R \cup \{q\})}(k + \log|R|)$, where $\text{dim}_{KR}(R \cup \{q\})$ is the expansion constant defined above. 
All proofs and pseudo-codes were omitted.
The authors didn't reply to our request for details. 
\medskip

\noindent
\textbf{Modified navigating nets} \citet{cole2006searching} were used in 2006 to claim the time $O(\log(n) + (1/\epsilon)^{O(1)})$ for the $(1+\epsilon)$-approximate neighbors.
All proofs and pseudo-codes were left out, also for the construction of the modified navigating net for the claimed time $O(|R| \cdot \log(|R|))$. 
\medskip

\noindent
\textbf{Cover trees}. 
In 2006, \citet{beygelzimer2006cover} introduced a cover tree inspired by the navigating nets \citet{krauthgamer2004navigating}. 
This cover tree was designed to prove a worst-case bound for the nearest neighbor search in terms of the size $|R|$ of a reference set $R$ and the expansion constant $c(R)$ of Definition \ref{dfn:expansion_constant}. 
Assume that a cover tree is already constructed on set $R$. Then \citet[Theorem~5]{beygelzimer2006cover} claims that nearest neighbor of any query point $q \in Q$ could be found in time $O(c(R)^{12} \cdot \log|R|)$. In 2015, \citet[Section~5.3]{curtin2015improving} pointed out that the proof of \citet[Theorem~5]{beygelzimer2006cover} contains a crucial gap, now have been confirmed by a specific dataset in \citet[Counterexample~5.2]{elkin2022counterexamples}. The time complexity result of the cover tree construction algorithm \citet[Theorem~6]{beygelzimer2006cover} had a similar issue, the gap of which is exposed rigorously in \citet[Counterexample~4.2]{elkin2022counterexamples}.
\medskip

\begin{figure}
	\centering
	\begin{tikzpicture}[align=center, node distance = 1.0cm, scale = 0.45]
    
	\node (scaleminftext) {$C_{-\infty}$};
	\node [blockz, right of =scaleminftext] (scaleminf1) {1};
	\node [blockz, right of =scaleminf1] (scaleminf2) {2};
	\node [blockz, right of =scaleminf2] (scaleminf3) {3};
	\node [blockz, right of =scaleminf3] (scaleminf4) {4};
	\node [blockz, right of =scaleminf4] (scaleminf5) {5};

	\node [above of =scaleminftext](scalem1text) {$C_{-1}$};
	\node [blockz, above of =scaleminf1] (scalem11) {1};
	\node [blockz, above of =scaleminf2] (scalem12) {2};
	\node [blockz, above of =scaleminf3] (scalem13) {3};
	\node [blockz, above of =scaleminf4] (scalem14) {4};
	\node [blockz, above of =scaleminf5] (scalem15) {5};
	
	\node [above of =scalem1text](scale0text) {$C_{0}$};
	\node [blockz, above of =scalem11] (scale01) {1};
	\node [blockz, above of =scalem13] (scale03) {3};
	\node [blockz, above of =scalem15] (scale05) {5};	
	
	\node [above of =scale0text](scale1text) {$C_{1}$};
	\node [blockz, above of =scale01] (scale11) {1};
	\node [blockz, above of =scale05] (scale15) {5};
	
	\node [above of =scale1text](scale2text) {$C_{2}$};
	\node [blockz, above of =scale11] (scale21) {1};

	\node [above of =scale2text](scaleinftext) {$C_{\infty}$};
	\node [blockz, above of =scale21] (scaleinf1) {1};
    
      \draw[dashed,->] (scalem11) -> (scaleminf1);
      \draw[dashed,->] (scalem12) -> (scaleminf2);
      \draw[dashed,->] (scalem13) -> (scaleminf3);
      \draw[dashed,->] (scalem14) -> (scaleminf4);
      \draw[dashed,->] (scalem15) -> (scaleminf5);

	  \draw[->] (scale01) -> (scalem11);

	  \draw[->] (scale03) -> (scalem12); 
     \draw[dashed, ->] (scale03) -> (scalem13);

	  \draw[->] (scale05) -> (scalem15);

	  \draw[->] (scale03) -> (scalem14);
	  
	  \draw[dashed,->] (scaleinf1) -> (scale21);
	  \draw[->] (scale21) -> (scale11);
	  \draw[->] (scale15) -> (scale05);
	  
	  \draw[->] (scale11) -> (scale01);
	  \draw[->] (scale11) -> (scale03);
	  \draw[->] (scale21) -> (scale15);
      
\end{tikzpicture} 
	\hspace{0.5cm}
	\begin{tikzpicture}[align=center, node distance = 1.0cm, scale = 0.45]
    

	\node (scalem2text) {Level 2};
	\node[below of =scalem2text] (scalem1text) {Level 1};
	\node[below of =scalem1text] (scalem0text) {Level 0};
	\node[below of =scalem0text] (scalemm1text) {Level -1};
	\node [blockz,  right=25pt of scalem2text ] (node1) {1};
	\node [blockz,  right=100pt of scalem1text] (node5) {5};
	\node [blockz,  right=25pt of scalem0text] (node3) {3};
	\node [blockz,  right=5pt of scalemm1text] (node2) {2};
	\node [blockz,  right=50pt of scalemm1text] (node4) {4};

	  \draw[->] (node1) -> (node5);
	  \draw[->] (node1) -> (node3);
	  \draw[->] (node3) -> (node2);
	  \draw[->] (node3) -> (node4);

	
	
	

    

	  
	  
	  
      
\end{tikzpicture} 
	\caption{\textbf{Left:} an implicit cover tree from \citet[Section~2]{beygelzimer2006cover} at ICML 2006 for a finite set of reference points $R = \{1,2,3,4,5\}$ with the Euclidean distance $d(x,y) = |x-y|$. 
		\textbf{Right:} a new compressed cover tree in Definition~\ref{dfn:cover_tree_compressed} corrects the past worst-case complexity for $k$-nearest neighbors search in $R$.}
	\label{fig:implicitcompressed}
\end{figure}
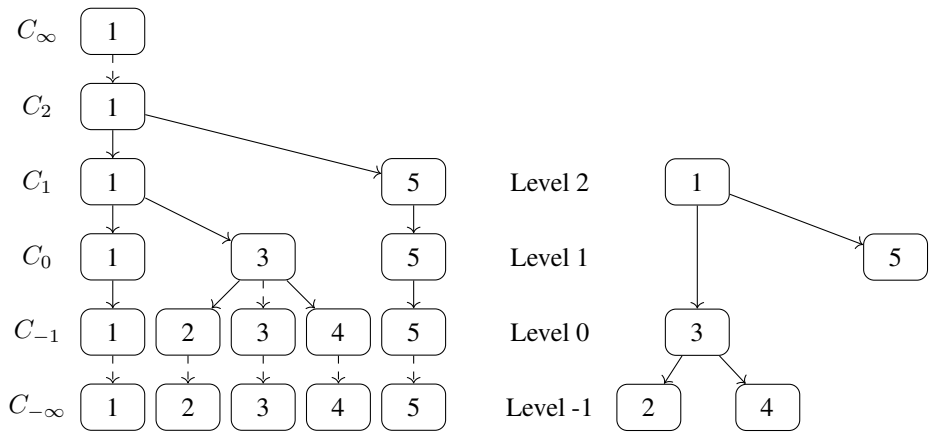

\noindent 
\textbf{Further studies in cover trees.} A noteworthy paper on cover trees \citet{kollar2006fast} introduced a new probabilistic algorithm for the nearest neighbor search, as well as corrected the pseudo-code of the cover tree construction algorithm of \citet[Algorithm~2]{beygelzimer2006cover}.
Later in 2015, a new, more efficient implementation of cover tree was introduced in \citet{izbicki2015faster}. However, no new time-complexity results were proven. A study \citet{jahanseir2016transforming} explored connections between modified navigating nets \citet{cole2006searching} and cover trees \citet{beygelzimer2006cover}.
 Multiple papers \citet{beygelzimer2006coverExtend, ram2009linear, curtin2015plug} studied possibility of solving $k$-nearest neighbor problem ( Problem \ref{pro:knn} ) by using cover tree on both, the query set and the reference set, for further details see \citet[Section~6]{elkin2022counterexamples}.
\medskip 

\noindent
\textbf{The main contributions} are the following.
\begin{itemize}
    \item Definition~\ref{dfn:cover_tree_compressed} introduces a compressed cover tree.
    \item Theorem \ref{thm:construction_time}  and Corollary \ref{cor:construction_time_KR} estimate the time to build a compressed cover tree.
    \item Theorem \ref{thm:knn_KR_time} and Corollary~\ref{cor:cover_tree_knn_miniziminzed_constant_time} estimate the time to find all $k$-nearest neighbors as in Problem~\ref{pro:knn}.
    \item Theorem~\ref{thm:approximate_k_nearestneighbors} estimates the time complexity of approximate $k$-nearest neighbor search.
\end{itemize}

\medskip

\noindent 
This work corrects the past gaps of the single-tree approach via an original  cover tree \citet{beygelzimer2006cover} by using a new compressed cover tree $\T(R)$ from Definition \ref{dfn:cover_tree_compressed}, which can be constructed on any finite reference set $R$ with a metric $d$. 
Theorem~\ref{thm:construction_time_KR} will prove that a compressed cover tree $\T(R)$ can be built in time $O(c_m(R)^8 \cdot c(R)^2 \cdot \log_2(|R|) \cdot |R|)$.
\medskip

 \noindent
The past gap in the proof of the time complexity \citet[Theorem~1]{beygelzimer2006cover} for nearest neighbor search is tackled by new Algorithm \ref{alg:cover_tree_k-nearest}, which add an essential block to the original code in \citet[Algorithm~1]{beygelzimer2006cover}.
The extra block eliminates the issue of having too many successive iterations when a query point $q$ is disproportionately far away from the remaining candidate set $R_i$ on some level $i$.
Then Lemma \ref{lem:knn_depth_bound} shows that the number of iterations of Algorithm \ref{alg:cover_tree_k-nearest} is bounded by $O(c(R)^2\log_2(|R|))$. 
This new lemma replaces the old result \citet[Lemma~4.3]{beygelzimer2006cover}, which had a similar bound for the number of explicit levels of a cover tree, for further information see \citet[Definition~3.2]{elkin2022counterexamples}
The old result cannot be used to estimate the number of iterations of \citet[Algorithm~1]{beygelzimer2006cover} due to \citet[Counterexample~5.2]{elkin2022counterexamples}. 
\medskip

\noindent
Assume that a compressed cover tree $\T(R)$ is already constructed on a reference set $R$. 
Our first main Theorem \ref{thm:knn_KR_time} shows that $k$-nearest neighbors of a query node $q$ can be found in time of
	$$O\Big ( c(R \cup \{q\})^2 \cdot \log_2(k) \cdot \big((c_m(R))^{10}  \cdot \log_2(|R|) + c(R \cup \{q\}) \cdot k\big) \Big).$$
Recall that $c(R)$ can potentially become as large as $O(|R|)$ when $R$ is not uniformly distributed.
Our second main Corollary~\ref{cor:cover_tree_knn_miniziminzed_constant_time} estimates the time complexity of the new $k$-nearest neighbor search by using only the minimized expansion constant $c_m(R)$ of Definition \ref{dfn:expansion_constant} and the aspect ratio $\Delta(R)$ of Definition \ref{dfn:radius+d_min} as parameters. These parameters are less dependent on the point distribution (or noise) in the sets $R,Q$.
In many cases, $\Delta(R)$ is relatively small and $c_m(R)$ depends mostly on the dimension of the ambient space $X$. It is shown that $k$-nearest neighbors of $q$ in a reference set $R$ can be found in time 
$$O\Big ((c_m(R))^{10} \cdot \log_2(k) \cdot \log_2(\Delta(R)) + |\bar{B}(q, 5d_k(q,R))| \cdot \log_2(k) \Big ), \text{ where }$$
 $d_k(q,R)$ is the distance from $q$ to its $k$th nearest neighbor.
Tables ~\ref{table:dim:construction}-\ref{table:dim:knearest} summarize past and new results.

\begin{table}[H]
	\label{table:KR:construction}
	\centering
	\caption{Results for building data structures with hidden classic expansion constant $c(R)$ of Definition \ref{dfn:expansion_constant} or KR-type constant $2^{\text{dim}_{KR}}$ \citet[Section~2.1]{krauthgamer2004navigating}}
             \vskip 0.15in
	\begin{tabular}{|V{3.0cm}|V{6cm}|V{25mm}|V{35mm}|}
		\hline
		Data structure    & claimed time complexity   & space & proofs \\
		\hline
		Navigating nets \citet{krauthgamer2004navigating} & $O\big(2^{O(\text{dim}_{KR})} \cdot |R| \log(|R|) \log(\log|R| )\big)$, \citet[Theorem~2.6]{krauthgamer2004navigating} & $O(2^{O(\text{dim})}|R|)$ & Not available \\
		\hline
		Cover tree \citet{beygelzimer2006cover}   &  $O(c(R)^{O(1)} \cdot |R| \cdot \log|R|)$, \citet[Theorem~6]{beygelzimer2006cover} & $O(|R|)$ & \citet[Counterexample~4.2]{elkin2022counterexamples} shows that the past proof is incorrect \\ 
		\hline
		Compressed cover tree [dfn \ref{dfn:cover_tree_compressed}] &     $O\big(c(R)^{O(1)} \cdot |R| \cdot \log(R) \big)$  &  $O(|R|)$ Lemma~\ref{lem:linear_space_cover_tree}     & Corollary \ref{cor:construction_time_KR} \\
		\hline       
	\end{tabular}
\end{table}
\begin{table}[H]
	\label{table:KR:knearest}
	\centering
             \vskip 0.15in
	\caption{Results for exact $k$-nearest neighbors of one query point $q \in X$ using hidden classic expansion constant $c(R)$ of Definition \ref{dfn:expansion_constant} or KR-type constant $2^{\text{dim}_{KR}}$  \citet[Section~2.1]{krauthgamer2004navigating} and assuming that all data structures are already built. Note that the dimensionality factor $2^{\text{dim}_{KR}}$ is equivalent  to $c(R)^{O(1)}$. }
	\begin{tabular}{|V{3.0cm}|V{5cm}|V{22mm}|V{35mm}|}
		\hline
		Data structure     & claimed time complexity   & space  & proofs \\
		\hline
		Navigating nets \citet{krauthgamer2004navigating} & $O\big(2^{O(\text{dim}_{KR})}(\log(|R|) + k)\big)$ for $k\geq 1$ \citet[Theorem~2.7]{krauthgamer2004navigating} & $O(2^{O(\text{dim})} |R|)$ & Not available \\
		\hline
		Cover tree \citet{beygelzimer2006cover}   &  $O\big(c(R)^{O(1)}\log|R|\big)$ for $k=1$ \citet[Theorem~5]{beygelzimer2006cover} & $O(|R|)$ & \citet[Counterexample~5.2]{elkin2022counterexamples} shows that the past proof is incorrect \\ 
		\hline
		Compressed cover tree, Definition \ref{dfn:cover_tree_compressed} &     $O\big(c(R)^{O(1)} \cdot \log(k) \cdot (\log(|R|) + k)\big)$  & $O(|R|)$, Lemma~\ref{lem:linear_space_cover_tree}      & Theorem \ref{thm:knn_KR_time} \\
		\hline            
	\end{tabular}
\end{table}
\begin{table}[H]
	\centering
	\caption{Building data structures with hidden $c_m(R)$ or dimensionality constant $2^{\text{dim}}$ \citet[Section~1.1]{krauthgamer2004navigating}}
             \vskip 0.15in
	\begin{tabular}{|V{3.0cm}|V{5.5cm}|V{25mm}|V{35mm}|}
		\hline
		Data structure    & claimed time complexity   & space & proofs \\
		\hline
		Navigating nets \citet{krauthgamer2004navigating} & $O\big(2^{O(\text{dim})} \cdot |R| \cdot \log(\Delta) \cdot \log(\log((\Delta ))\big)$
		& $O(2^{O(\text{dim})}|R|)$  & \citet[Theorem~2.5]{krauthgamer2004navigating} \\
		\hline
		Compressed cover tree [dfn \ref{dfn:cover_tree_compressed}] & $O\big(c_m(R)^{O(1)}\cdot|R|\log(\Delta(|R|))\big)$    & $O(|R|)$ Lemma~\ref{lem:linear_space_cover_tree}     & Theorem \ref{thm:construction_time}  \\
		\hline           
	\end{tabular}
	\label{table:dim:construction}  
\end{table}
\begin{table}[H]
	\centering
      \vskip 0.15in
	\caption{Results for exact $k$-nearest neighbors of one point $q$ using hidden $c_m(R)$ or dimensionality constant $2^{\text{dim}}$ \citet[Section~1.1]{krauthgamer2004navigating}  assuming that all structures are built.}
	\begin{tabular}{|V{3.0cm}|V{55mm}|V{22mm}|V{3.0cm}|}
		\hline
		Data structure     & claimed time complexity & space & proofs \\
		\hline
		Navigating nets \citet{krauthgamer2004navigating} & $O\big(2^{O(\text{dim})} \cdot \log(\Delta) + |\bar{B}(q,O(d(q,R))|\big)$ for $k = 1$  & $O(2^{O(\text{dim})}|R|)$  & a proof outline in \citet[Theorem~2.3]{krauthgamer2004navigating} \\
		\hline
		Compressed cover tree, Definition~\ref{dfn:cover_tree_compressed} &     $O\big(c_m(R)^{O(1)} \cdot \log(k) \cdot (\log(|\Delta|) + |\bar{B}(q,O(d_k(q,R))| ) \big )$  & $O(|R|)$, Lemma~\ref{lem:linear_space_cover_tree}     & Corollary~\ref{cor:cover_tree_knn_miniziminzed_constant_time} \\
		\hline          
	\end{tabular}
	\label{table:dim:knearest}
\end{table}




\section{Compressed cover tree}
\label{sec:cover_tree}
This section introduces in Definition \ref{dfn:cover_tree_compressed} a new compressed cover tree, which will be used to solve Problem \ref{pro:knn}. 
Other important results are Lemmas~\ref{lem:packing} and~\ref{lem:growth_bound}.
Given a $\delta$-sparse finite metric space $R$, Lemma~\ref{lem:packing} shows that the number of points of $R$ in the closed ball $\bar{B}(p,t)$ has the upper bound $c_m(S)^{\mu}$, where $\mu$ depends on $\frac{t}{\delta}$.
Lemma~\ref{lem:growth_bound} will imply that if there are points $p,q$ in a finite metric space $R$ satisfying $2r < d(p,q) \leq 3r$ for some $r \in \R$, then $|\bar{B}(q,4r)| \geq (1 + \frac{1}{c(R)^2})|\bar{B}(q,r)|$.

 \dfncovertreecompressed*

\begin{lem}[Linear space of $\T(R)$]
\label{lem:linear_space_cover_tree}
Let $(R,d)$ be a finite metric space. Then any cover tree $\T(R)$ from Definition \ref{dfn:cover_tree_compressed} takes $O(|R|)$ space. 
\end{lem}
\begin{proof}
Since $\T(R)$ is a tree , both its vertex set and its edge set contain at most $|R|$ nodes. Therefore $\T(R)$ takes at most $O(|R|)$ space. 
\end{proof}

\begin{figure}[h]
	\centering
	\begin{subfigure}{.30\textwidth}
		\centering
		\begin{tikzpicture}[align=center, node distance = 1.0cm, scale = 0.45]

	\node (scalem2text) {Level $i$};
	\node[below of =scalem2text] (scalem1text) {Level $i-1$};
	\node[below of =scalem1text] (scalem0text) {Level -1};
	\node [blockz,  right=30pt of scalem2text ] (node1) {$2^{i}$};
    \node [blockz,  right=25pt of scalem1text ] (node2) {1};
	\node [blockz,  right=32pt of scalem0text ] (node3) {0};
	

	  \draw[->] (node1) -> (node2);
	   \draw[->] (node2) -> (node3);

\end{tikzpicture}
		\label{fig:cover_tree_variant_one}
		
	\end{subfigure}
	\begin{subfigure}{.30\textwidth}
		\centering
		\begin{tikzpicture}[align=center, node distance = 1.0cm, scale = 0.45]

	\node (scalem2text) {Level $i$};
	\node[below of =scalem2text] (scalem1text) {Level $i-1$};
	\node[below of =scalem1text] (scalem0text) {Level -1};
	\node [blockz,  right=30pt of scalem2text ] (node1) {$2^{i}$};
    \node [blockz,  right=25pt of scalem1text ] (node2) {0};
	\node [blockz,  right=32pt of scalem0text ] (node3) {1};
	

	  \draw[->] (node1) -> (node2);
	   \draw[->] (node2) -> (node3);

\end{tikzpicture}
		\label{fig:cover_tree_variant_two}
	\end{subfigure}
	\begin{subfigure}{.30\textwidth}
		\centering
		\begin{tikzpicture}[align=center, node distance = 1.0cm, scale = 0.45]

	\node (scalem2text) {Level $i$};
	\node[below of =scalem2text] (scalem1text) {Level $i-1$};
	\node[below of =scalem1text] (scalem0text) {Level -1};
	\node [blockz,  right=30pt of scalem2text ] (node1) {0};
    \node [blockz,  right=50pt of scalem1text ] (node2) {$2^{i}$};
	\node [blockz,  right=32pt of scalem0text ] (node3) {1};
	

	  \draw[->] (node1) -> (node2);
	   \draw[->] (node1) -> (node3);

\end{tikzpicture}
		\label{fig:cover_tree_variant_three}
	\end{subfigure}
	\caption{
		Compressed cover trees $\T(R)$ from Definition~\ref{dfn:cover_tree_compressed} for $R = \{0,1,2^{i}\}$. 
	}
	\label{fig:cover_tree_easy_example}
\end{figure}
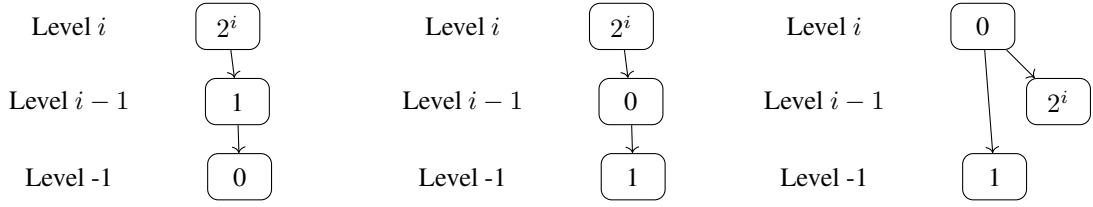

\begin{figure}
	\centering
	\begin{tikzpicture}[align=center, node distance = 1.0cm, scale = 0.45]

	\node (scale3) {Level $2$};
	\node[below of =scale3] (scale2) {Level 1};
	\node[below of =scale2] (scale1) {Level 0};
	\node[below of =scale1] (scale0) {Level $-1$};
	\node [blockzm1,  right=1pt of scale0 ] (node1) {1};
	\node [blockzm1,  right=26pt of scale0 ] (node3) {3};
	\node [blockzm1,  right=51pt of scale0 ] (node5) {5};
	\node [blockzm1,  right=76pt of scale0 ] (node7) {7};
	\node [blockzm1,  right=101pt of scale0 ] (node9) {9};
	\node [blockzm1,  right=126pt of scale0 ] (node11) {11};
	\node [blockzm1,  right=151pt of scale0 ] (node13) {13};
	\node [blockzm1,  right=176pt of scale0 ] (node15) {15};
		\node [blockzm1,  right=13pt of scale1 ] (node2) {2};
		\node [blockzm1,  right=63pt of scale1 ] (node6) {6};
		\node [blockzm1,  right=113pt of scale1 ] (node10) {10};
		\node [blockzm1,  right=163pt of scale1 ] (node14) {14};

		\node [blockzm1,  right=38pt of scale2 ] (node4) {4};
		\node [blockzm1,  right=138pt of scale2 ] (node12) {12};
		\node [blockzm1,  right=88pt of scale3 ] (node8) {8};
	

	
    \draw[->] (node2) -> (node1);
    \draw[->] (node2) -> (node3);
     \draw[->] (node6) -> (node5);
    \draw[->] (node6) -> (node7);
     \draw[->] (node10) -> (node9);
    \draw[->] (node10) -> (node11);
     \draw[->] (node14) -> (node13);
    \draw[->] (node14) -> (node15);

    \draw[->] (node8) -> (node4);
    \draw[->] (node8) -> (node12);

      \draw[->] (node4) -> (node2);
     \draw[->] (node4) -> (node6);
         
         \draw[->] (node12) -> (node10);
     \draw[->] (node12) -> (node14);

\end{tikzpicture}
	\caption{Compressed cover tree $\T(R)$ on the set $R$ in Example \ref{exa:cover_tree_big} with root $16$. }
	\label{fig:cover_tree_big}
\end{figure}
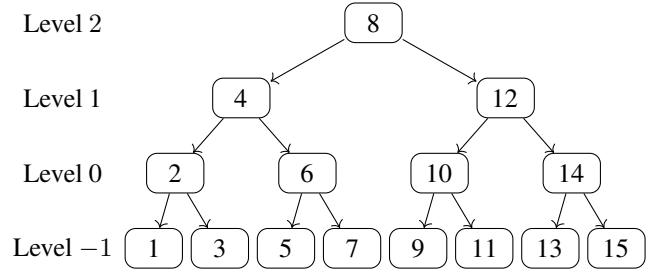

\begin{exa}[$\T(R)$ in Fig.~\ref{fig:cover_tree_big}]
	\label{exa:cover_tree_big}
	Let $(\R, d = |x-y|)$ be the real line with euclidean metric.
	Let $R = \{1,2,3,...,15\} $ be its finite subset. 
	Fig.~\ref{fig:cover_tree_big} shows a compressed cover tree on the set $R$ with the root $r=8$. 
	The cover sets of $\T(R)$ are $ C_{-1} = \{1,2,3,...,15\}$, $C_0 = \{2,4,6,8,10,12,14\}$, $C_{1} = \{4,8,12\}$ and $C_{2} = \{8\}$.
	We check the conditions of Definition \ref{dfn:cover_tree_compressed}.
	\begin{itemize}
		\item Root condition $(\ref{dfn:cover_tree_compressed}a)$: 
		since $\max_{p \in R \setminus \{8\}}d(p, 8) = 7$ and $\ceil{\log_2(7)} - 1= 2$, the root can have the level $l(8) = 2$.
		\item Covering condition (\ref{dfn:cover_tree_compressed}b) : for any $i \in {-1,0,1,2}$, let $p_i$ be arbitrary point having $l(p_i) = i$. Then we have 
		$d(p_{-1}, p_{0}) = 1 \leq 2^{0}$, 
		$d(p_0, p_1) = 2 \leq 2^{1}$ and  
		$d(p_1, p_2) = 4 \leq 2^{2}$.
		\item  Separation condition (\ref{dfn:cover_tree_compressed}c) : $d_{\min}(C_{-1}) = 1 > \frac{1}{2} = 2^{-1}$, $d_{\min}(C_{0}) = 2 > 1 = 2^{0}, d_{\min}(C_{1}) = 4 > 2 = 2^{1}$. 
		\bs
	\end{itemize}
\end{exa}

\noindent
A cover tree was defined in \citet[Section~2]{beygelzimer2006cover} as a tree version of a navigating net from \citet[Section ~ 2.1]{krauthgamer2004navigating}. 
For any index $i \in \Z\cup \{\pm\infty\}$, the level $i$ set of this cover tree coincides with the cover set $C_i$ above, which can have nodes at different levels in Definition~\ref{dfn:cover_tree_compressed}. 
Any point $p \in C_i$ has a single parent in the set $C_{i+1}$, which satisfied conditions (\ref{dfn:cover_tree_compressed}b,c). 
\citet[Section~2]{beygelzimer2006cover} referred to this original tree as an implicit representation of a cover tree.
Such a tree in Figure \ref{fig:tripleexample} (left) contains infinitely many repetitions of every point $p\in R$ in long branches and will be called an \emph{implicit cover tree}.
\medskip

\noindent
Since an implicit cover tree is formally infinite, for practical implementations, the authors of \citet{beygelzimer2006cover} had to use another version that they named an explicit representation of a cover tree. 
We call this version an \emph{explicit cover tree}. 
Here is the full defining quote at the end of \citet[Section~2]{beygelzimer2006cover}: "The explicit representation of the tree coalesces all nodes in which the only child is a self-child". 
In an explicit cover tree, if a subpath of every node-to-root path consists of all identical nodes without other children, all these identical nodes collapse to a single node, see Figure \ref{fig:tripleexample} (middle). 
\medskip

\noindent
Since an explicit cover tree still contains repeated points, Definition~\ref{dfn:cover_tree_compressed} is well-motivated by the aim to include every point only once, which saves memory and simplifies all subsequent algorithms, see Fig.~\ref{fig:tripleexample} (right).

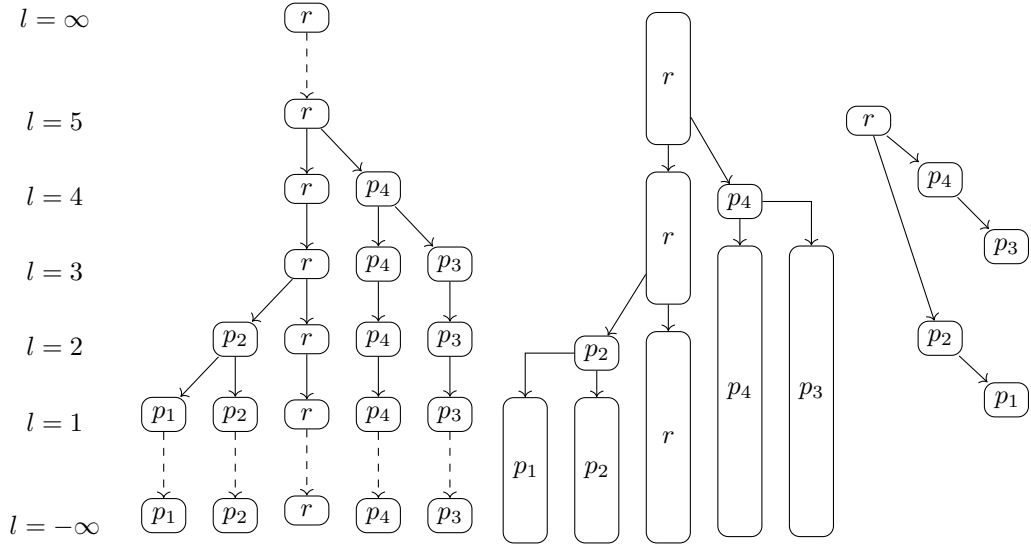
\begin{figure}
	\centering
	\begin{tikzpicture}[align=center, node distance = 1.0cm, scale = 0.45]

    \node (scaleinf) {$l = \infty$};
    \node [below = 12.5pt of scaleinf](scalemidinf) {};
     \node[below=25pt of scaleinf] (scale6) {$l = 5$};
    \node[below of =scale6] (scale5) {$l = 4$};
    \node[below of =scale5] (scale4) {$l = 3$};
	\node[below of =scale4] (scale3) {$l = 2$};
	\node[below of =scale3] (scale2) {$l = 1$};
	\node[below=25pt of scale2] (scale1) {$l = -\infty$};

 	\node [blockzm2,  right=70pt of scaleinf ] (node1x) {$r$};
 	\node [blockzm2, below=25pt of node1x ] (node1a) {$r$};
 	\node [blockzm2, below of = node1a ] (node1b) {$r$};
 	\node [blockzm2,  below of = node1b ] (node1c) {$r$};
 	\node [blockzm2,   below of = node1c ] (node1d) {$r$};
 	\node [blockzm2,  below of = node1d ] (node1e) {$r$};
 	\node [blockzm2,   below =25pt of node1e ] (node1f) {$r$};
 	
 	\draw[dashed, ->] (node1x) -> (node1a);
 	    \draw[->] (node1a) -> (node1b);
 	     \draw[->] (node1b) -> (node1c);
 	      \draw[->] (node1c) -> (node1d);
 	       \draw[->] (node1d) -> (node1e);
 	      \draw[dashed, ->] (node1e) -> (node1f);

 		\node [blockzm2, right=10pt of node1b ] (node2b) {$p_{4}$};
 	\node [blockzm2,  below of = node2b ] (node2c) {$p_{4}$};
 	\node [blockzm2,   below of = node2c ] (node2d) {$p_{4}$};
 	\node [blockzm2,  below of = node2d ] (node2e) {$p_{4}$};
 	\node [blockzm2,   below =25pt of node2e ] (node2f) {$p_{4}$};
 	
 	    \draw[->] (node1a) -> (node2b);
 	    
 	    \draw[->] (node2b) -> (node2c);
 	      \draw[->] (node2c) -> (node2d);
 	       \draw[->] (node2d) -> (node2e);
 	      \draw[dashed, ->] (node2e) -> (node2f);

 		\node [blockzm2, right = 10pt of node2c ] (node3c) {$p_{3}$};
 	\node [blockzm2,   below of = node3c ] (node3d) {$p_{3}$};
 	\node [blockzm2,  below of = node3d ] (node3e) {$p_{3}$};
 	\node [blockzm2,   below =25pt of node3e ] (node3f) {$p_{3}$};
 	
 	        \draw[->] (node2b) -> (node3c);
 	      \draw[->] (node3c) -> (node3d);
 	       \draw[->] (node3d) -> (node3e);
 	      \draw[dashed, ->] (node3e) -> (node3f);

 		\node [blockzm2,   left = 10 pt of  node1d ] (node4d) {$p_{2}$};
 	\node [blockzm2,  below of = node4d ] (node4e) {$p_{2}$};
 	\node [blockzm2,   below =25pt of node4e ] (node4f) {$p_{2}$};

            \draw[->] (node1c)   -> (node4d);
 	       \draw[->] (node4d) -> (node4e);
 	      \draw[dashed, ->] (node4e) -> (node4f);

 	\node [blockzm2,  left = 10 pt of node4e ] (node5e) {$p_{1}$};
 	\node [blockzm2,   below =25pt of node5e ] (node5f) {$p_{1}$};
 	
 	\draw[->] (node4d) -> (node5e);
 	 \draw[dashed, ->] (node5e) -> (node5f);
 	
 	\node[rectangle, draw, fill=white, 
    text width=1.0em, text centered, rounded corners, minimum height=50pt, minimum width = 1.0em, right = 220pt of scalemidinf] (ynode1a) {$r$};
    
    	\node[rectangle, draw, fill=white, 
    text width=1.0em, text centered, rounded corners, minimum height=50pt, minimum width = 1.0em, below = 10pt of ynode1a] (ynode1b) {$r$};
    
    	\node[rectangle, draw, fill=white, 
    text width=1.0em, text centered, rounded corners, minimum height=80pt, minimum width = 1.0em, below = 10pt of ynode1b] (ynode1c) {$r$};

    \node[blockzm2, above right = -18 pt and 10 pt of ynode1b] (ynode2a) {$p_4$};
    
    	\node[rectangle, draw, fill=white, 
    text width=1.0em, text centered, rounded corners, minimum height=110pt, minimum width = 1.0em, below = 10pt of ynode2a] (ynode2b) {$p_4$};
    
    \node[rectangle, draw, fill=white, 
    text width=1.0em, text centered, rounded corners, minimum height=110pt, minimum width = 1.0em, right = 10pt of ynode2b] (ynode3a) {$p_3$};

        \node[blockzm2, above left = -15 pt and 10 pt of ynode1c] (ynode4a) {$p_2$};
    
    	\node[rectangle, draw, fill=white, 
    text width=1.0em, text centered, rounded corners, minimum height=55pt, minimum width = 1.0em, below = 10pt of ynode4a] (ynode4b) {$p_2$};
    
    \node[rectangle, draw, fill=white, 
    text width=1.0em, text centered, rounded corners, minimum height=55pt, minimum width = 1.0em, left = 10pt of ynode4b] (ynode5a) {$p_1$};

    \draw[->] (ynode1a) -> (ynode1b);
    \draw[->] (ynode1b) -> (ynode1c);
    \draw[->] (ynode1a) -> (ynode2a);
    \draw[->] (ynode2a) -| (ynode3a);
    \draw[->] (ynode2a) -> (ynode2b);
    
      \draw[->] (ynode1b) -> (ynode4a);
     \draw[->] (ynode4a) -| (ynode5a);
     \draw[->] (ynode4a) -> (ynode4b);
    

 	\node[blockzm2, right = 285pt of scale6] (xnode1) {$r$};
	\node[blockzm2, below right = 10 pt and 10 pt of xnode1] (xnode2) {$p_4$};
 	\node[blockzm2, below right = 70 pt and 10 pt of xnode1] (xnode4) {$p_2$};
 	\node[blockzm2, below right = 10 pt and 8 pt of xnode4] (xnode5) {$p_1$};
 	\node[blockzm2, below right = 12 pt and 8 pt of xnode2] (xnode3) {$p_3$};
 	
 	\draw[->]	(xnode1) -> (xnode2);
 	\draw[->]	(xnode2) -> (xnode3);
 	\draw[->]	(xnode1) -> (xnode4);
 	\draw[->]	(xnode4) -> (xnode5);

\end{tikzpicture}
	\caption{A comparison of past cover trees and a new tree in Example \ref{exa:implicitexplicitexample}. \textbf{Left:} an implicit cover tree contains infinite repetitions. \textbf{Middle:} an explicit cover tree. \textbf{Right:} a compressed cover tree from Definition \ref{dfn:cover_tree_compressed} includes each point once.  }
	\label{fig:tripleexample}
\end{figure}

\begin{exa}[a short train line tree]
	\label{exa:implicitexplicitexample}
	Let $G$ be the unoriented metric graph consisting of two vertices $r,q$ connected by three different edges $e,h,g$ of lengths $|e| = 2^6$ , $|h| = 2^{3}$ , $|g| = 1$. Let $p_{4}$ be the middle point of the edge $e$. 
	Let $p_{3}$ be the middle point of the subedge $(p_4 , q)$. 
	Let $p_{2}$ be the middle point of the edge $h$.
	Let $p_{1}$ be the middle point of the subedge $(p_{2}, q)$. 
	Let $R = \{p_1, p_2,p_3,p_4,r\}$. 
	We construct a compressed cover tree $\T(R)$ by choosing the level $l(p_i) = i$ and by setting the root $r$ to be the parent of both $p_2$ and $p_4$, $p_4$ to be the parent of $p_{3}$, and $p_{2}$ to be the parent of $p_{1}$. 
	Then $\T(R)$ satisfies all the conditions of Definition \ref{dfn:cover_tree_compressed}, see a comparison of the three cover trees in Fig.~\ref{fig:tripleexample}.
	\bs
\end{exa}


\begin{figure}[h]
	\centering
	 \begin{tikzpicture}[scale = 1.0, node distance = 1cm]


\foreach \Point in {(-2,1), (-1,1), (-2,0), (-1,0), (-2,-1) , (-1,-1) , (0,-1), (0,0), (0,1)}{
    \node[circle,fill=black,inner sep=1.5pt] at \Point {};
}


\node [circle,fill=red,inner sep=1.5pt, red, label = {$p$}] at (2.5,0) {};
\draw (2.5,0) circle[radius = 1.9];

\node at (-1,-1.5) {$R \setminus \{p\}$};

\end{tikzpicture}
	\caption{Example~\ref{exa:outlierconstruction} describes a set $R$ with a big expansion constant $c(R)$. 
		Let $R\setminus \{p\}$ be a finite subset of a unit square lattice in $\R^2$, but a point $p$ is located far away from $R\setminus \{p\}$ at a distance larger than $\diam(R \setminus \{p\})$.
		Definition \ref{dfn:expansion_constant} implies that $c(R) = |R|$. }
	\label{fig:outlierconstruction}
\end{figure}
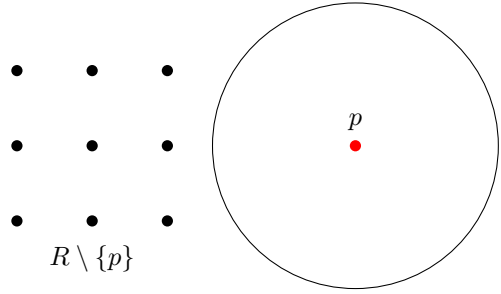

\noindent
Even a single outlier point can make the expansion constant big. 
Consider set $R = \{1,2,...,n-1,2n\}$ for some $n \in \Z_{+}$. 
Since $|\bar{B}(2n,n)| = 1$ and $|\bar{B}(2n,n)| = |R|$, we have $c(R) = |R|$. 
Example~\ref{exa:outlierconstruction} shows that expansion constant of a set $R$ can be as big as $|R|$. 

\begin{exa}[one outlier can make the expansion constant big]
	\label{exa:outlierconstruction}
	Let $R$ be a finite metric space and $p \in R$ satisfy $d(p,R \setminus \{t\}) > \diam(R \setminus \{p\})$. 
	Since $\bar{B}(p, 2d(p,R \setminus \{t\}) = R$ , $\bar{B}(p, d(p,R \setminus \{t\}) = \{p\}$, we get $c(R) = N$, see Fig.~\ref{fig:outlierconstruction}. 
	\bs
\end{exa}

\noindent
Example~\ref{exa:minimized_normal_expansion_constant} shows that the minimized expansion can be significantly smaller than the original expansion constant. 

\begin{exa}[minimized expansion constants]
	\label{exa:minimized_normal_expansion_constant}
	Let $(\R, d)$ be the Euclidean line. 
	For an integer $n>10$, consider the finite sets $R = \{2^{i} \mid i \in [1,n]\}$ and let $Q = \{i \mid i \in [1,2^n]\}$. 
	If $0<\epsilon < 10^{-9}$, then
	$\bar{B}(2^n, 2^{n-1} - \epsilon) = \{2^n\} $ and $\bar{B}(2^n, 2(2^{n-1} - \epsilon)) = R$, so $c(R) = n$.
	For any $q \in Q$ and any $t \in \R$, we have that $\bar{B}(q,t) = \mathbb{Z} \cap [q - t, q + t]$ and 
	$\bar{B}(q,2t) = \mathbb{Z} \cap [q - 2t, q + 2t]$, so $c(Q) \leq 4$.  
	Then $c_m(R) \leq c_m(Q) \leq c(Q) \leq 4$ by Lemma~\ref{lem:expansion_constant_property}.
	\bs 
\end{exa}

\noindent
Lemma~\ref{lem:compressed_cover_tree_descendant_bound} provides an upper bound for a distance between a node and its descendants.  

\begin{lem}[a distance bound on descendants]
	\label{lem:compressed_cover_tree_descendant_bound}
	Let $R$ be a finite subset of an ambient space $X$ with a metric $d$. 
	In a compressed cover tree $\T(R)$, let $q$ be any descendant of a node $p$. Let the node-to-root path $S$ of $q$ contain a node $u$ satisfying $u \in \Child(p) \setminus \{p\}$. Then $d(p,q) \leq 2^{l(u) + 2} \leq 2^{l(p) + 1}$. 
	\bs
\end{lem}
\begin{proof}
	Let $(w_0, ..., w_m)$ be a subpath of the node-to-root path for $w_0 = q$ , $w_{m-1} = u$, $w_m = p$. 
	Then $d(w_{i}, w_{i+1}) \leq 2^{l(w_i) + 1}$ for any $i$. 
	The first required inequality follows from the triangle inequality below: 
	$$    d(p,q) \leq \sum^{m-1}_{j = 0}d(w_j, w_{j+1})  \leq \sum^{m-1}_{j = 0}2^{l(w_j) + 1} \leq   \sum_{t = l_{\min}}^{l(u) + 1}2^{t}\leq  2^{l(u) + 2} $$
	Finally, $l(u) \leq l(p) - 1$ implies that $d(p,q) \leq 2^{l(p)+1}$.
\end{proof}

\noindent
Lemma~\ref{lem:packing} uses the idea of \citet[Lemma~1]{curtin2015plug} to show that if $S$ is a $\delta$-sparse subset of a metric space $X$, then $S$ has at most $(c_m(S))^\mu$ points in the ball $\bar{B}(p,r)$, where $c_m(S)$ is the minimized expansion constant of $S$, while $\mu$ depends on $\delta,r$. 

\begin{figure}
	\centering
	\input   \begin{tikzpicture}[scale = 1.0]




\node [circle,fill=red, red, inner sep=1pt, label = {$p$}] at (0,0) {};
\draw (0,0) circle[radius = 2];

\node [circle,fill=black, black, inner sep=1pt] at (-1,-1) {};
\draw[line width = .5pt, dash pattern=on 1pt off 2pt] (-1,-1) circle[ radius = 0.5];

\node [circle,fill=black, black, inner sep=1pt] at (1,0.75) {};
\draw[line width = .5pt, dash pattern=on 1pt off 2pt] (1,0.75) circle[ radius = 0.5];

 \draw[dashed] (1,0.75) -- (1.5,0.75);
 \node at (1.25,0.6) {\tiny $\delta / 2$}; 

\node [circle,fill=black, black, inner sep=1pt] at (-0.1,-0.4) {};
\draw[line width = .5pt, dash pattern=on 1pt off 2pt] (-0.1,-0.4) circle[ radius = 0.5];

\node [circle,fill=black, black, inner sep=1pt] at (0.75,-1.25) {};
\draw[line width = .5pt, dash pattern=on 1pt off 2pt] (0.75,-1.25) circle[ radius = 0.5];

\node [circle,fill=black, black, inner sep=1pt] at (-0.2,1.2) {};
\draw[line width = .5pt, dash pattern=on 1pt off 2pt] (-0.2,1.2) circle[ radius = 0.5];

\node [circle,fill=black, black, inner sep=1pt] at (-1.1,0.65) {};
\draw[line width = .5pt, dash pattern=on 1pt off 2pt] (-1.1,0.65) circle[ radius = 0.5];

 \draw[dashed] (0,0) -- node[below] {$t$} node[above] {} ++ (2,0);


\node at (-1.5,-0.25) {$S$};

\end{tikzpicture} 
	\caption{This volume argument proves Lemma~\ref{lem:packing}. By using an expansion constant, we can find an upper bound for the number of smaller balls of radius $\frac{\delta}{2}$ that can fit inside a larger $\bar{B}(p, t)$. }
	\label{fig:packingLemma}
\end{figure}
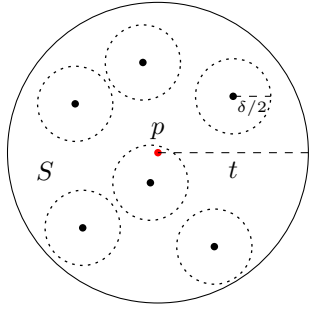

\lempacking*
\begin{proof}
	Assume that $d(p,q) > t$ for any point $q \in S$. 
	Then $\bar{B}(p, t) \cap S = \emptyset$ and the lemma holds trivially. Otherwise $\bar{B}(p, t) \cap S$ is non-empty. 
	By Definition~\ref{dfn:expansion_constant} of a minimized expansion constant, for any small enough $\epsilon > 0$, we can always find $\xi \leq \frac{2t + \frac{\delta}{2}}{2^{\mu}}$ and a set $A$ such that $S \subseteq A \subseteq X$  for which:
	\begin{ceqn}
		\begin{equation}
			\label{eqa:dfn_of_exp_constant}
			|B(q,2s) \cap A| \leq (c_m(S) + \epsilon) \cdot | B(q,s) \cap A|
		\end{equation}
	\end{ceqn}
	for any $q \in A$ and $s > \xi$. Note that for any $u \in \bar{B}(p,t) \cap S$ we have $\bar{B}(u, \frac{\delta}{2}) \subseteq  \bar{B}(p, t + \frac{\delta}{2})$. 
	Therefore, for any point $q \in \bar{B}(p,t) \cap S$, we get
	$$\bigcup_{u \in \bar{B}(p, t) \cap S}\bar{B}(u, \frac{\delta}{2}) \subseteq  \bar{B}(p,t + \frac{\delta}{2}) \subseteq \bar{B}(q, 2t +  \frac{\delta}{2})$$
	Since all the points of $S$ were separated by $\delta$, we have
	\begin{equation*}
		\label{eqa:packing_zero}
		| \bar{B}(p, t) \cap S| \cdot \min_{u \in \bar{B}(p, t) \cap S}| \bar{B}(u, \frac{\delta}{2}) \cap A|  \leq \sum_{u \in \bar{B}(p, t) \cap S} | \bar{B}(u, \frac{\delta}{2}) \cap A |\leq | \bar{B}(q, 2t + \frac{\delta}{2}) \cap A |
	\end{equation*}
	In particular, by setting $q = \mathrm{argmin}_{a \in S \cap \bar{B}(p,t)}| \bar{B}(a, \frac{\delta}{2})|$, we get:
	\begin{ceqn}
		\begin{equation}
			\label{eqa:packing_one}
			| \bar{B}(p, t) \cap S| \cdot |  \bar{B}(q, \frac{\delta}{2}) \cap A| \leq | \bar{B}(q, 2t + \frac{\delta}{2}) \cap A |
		\end{equation}
	\end{ceqn}
	Inequality~(\ref{eqa:dfn_of_exp_constant}) applied $\mu$ times 
	for the radii $s_i = \dfrac{2t + \frac{\delta}{2}}{2^{i}} $, $i = 1,...,\mu$, implies that:
	\begin{ceqn}
		\begin{equation} 
			\label{eqa:packing_two}
			|\bar{B}(q,2t + \frac{\delta}{2}) \cap A| \leq (c_m(S) + \epsilon)^{\mu}|\bar{B}(q, \dfrac{2t + \frac{\delta}{2}}{2^{ \mu}}) \cap A |  \leq  (c_m(S) + \epsilon)^{ \mu}|\bar{B}(q, \frac{\delta}{2}) \cap A|. 
		\end{equation}
	\end{ceqn}
	By combining inequalities (\ref{eqa:packing_one}) and (\ref{eqa:packing_two}), we get
	$$| \bar{B}(p,t) \cap S  |\leq \dfrac{|\bar{B}(q, 2t + \frac{\delta}{2})  \cap A |}{|\bar{B}(q, \frac{\delta}{2})  \cap A|} \leq (c_m(S)+\epsilon)^{\mu}.$$
	The required inequality is obtained by letting $\epsilon \rightarrow 0$.\end{proof}

\noindent
\citet[Section~1.1]{krauthgamer2004navigating} defined dim($X$) of a space $(X,d)$ as the minimum number $m$ such that every set $U \subseteq X$ can be covered by $2^{m}$ sets whose diameter is a half of the diameter of $U$. 
If $U$ is finite, an easy application of Lemma \ref{lem:packing} for $\delta = \frac{r}{2}$ shows that
$\text{dim}(X) \leq \sup_{A \subseteq X}(c_m(A))^4 \leq \sup_{A \subseteq X}\inf_{A \subseteq B \subseteq X}(c(B))^4,$
where $A$ and $B$ are finite subsets of $X$. 
\medskip

\noindent
Let $T(R)$ be an implicit cover tree of \citet{beygelzimer2006cover} on a finite set $R$. 
\citet[Lemma~4.1]{beygelzimer2006cover} showed that the number of children of any node $p \in T(R)$ has the upper bound $(c(R))^4$. 
Lemma~\ref{lem:compressed_cover_tree_width_bound} generalizes \citet[Lemma~4.1]{beygelzimer2006cover} for a compressed cover tree.

\lemwidthbound*
\begin{proof}
    \hypertarget{proof:lem:compressed_cover_tree_width_bound}{\empty}
	By the covering condition of $\T(R)$, any child $q$ of $p$ located on the level $i$ has $d(q,p) \leq 2^{i+1}$. 
	Then the number of children of the node $p$ at level $i$ at most $|\bar{B}(p,2^{i+1})|$. 
	The separation condition in Definition~\ref{dfn:cover_tree_compressed} implies that the set $C_i$ is a $2^{i}$-sparse subset of $X$. 
	We apply Lemma \ref{lem:packing} for $t = 2^{i+1}$ and $\delta = 2^{i}$. 
	Since $4 \cdot \frac{t}{\delta} + 1 \leq 4 \cdot 2 + 1 \leq 2^4$, we get $|\bar{B}(q,2^{i+1}) \cap C_i|  \leq (c_m(C_{i}))^4$. Lemma \ref{lem:expansion_constant_property} implies that $(c_m(C_{i}))^4  \leq (c_m(R))^4 $, so the upper bound is proved.
\end{proof}

%

\lemgrowthbound*
\begin{proof}
\hypertarget{proof:lem:growth_bound}{\empty}
	Since $\bar{B}(q,r) \subset \bar{B}(p,3r + r)$, we have
	$|\bar{B}(q,r)| \leq |\bar{B}(q,4r)| \leq c(A)^2 \cdot |\bar{B}(p,r)|$.
	And since $\bar{B}(p,r)$ and $\bar{B}(q,r)$ are disjoint and are subsets of $\bar{B}(q,4r)$, we have
	$$|\bar{B}(q, 4r)| \geq |\bar{B}(q,r)| + |\bar{B}(p,r)| \geq |\bar{B}(q,r)| + \frac{|\bar{B}(q,r)|}{c(A)^2} \geq (1 + \frac{1}{c(A)^2}) \cdot |\bar{B}(q,r)|,$$
	which proves the claim.
\end{proof}
\lemgrowthboundextension*
\begin{proof}
	\hypertarget{proof:lem:growth_bound_extension}{\empty} Let us prove this by induction. In basecase $n = 1 $ define $r = \frac{d(q,p_{m})}{3}$. 
	Now by Lemma \ref{lem:growth_bound} we have 
	$$|\bar{B}(q, \frac{4}{3}d(q,p_1)) | = |\bar{B}(q, 4r)|  \geq  (1+\frac{1}{c(A)^2}) \cdot |\bar{B}(q, r)| = |\bar{B}(q, \frac{1}{3}d(q,p_1)) |.$$

	\medskip
	
	Assume now that the claim holds for some $i = m$ and let $p_1, ..., p_{m+1}$ be a sequence satisfying the condition of Lemma \ref{lem:growth_bound_extension}. 
	By induction assumption we have $ |\bar{B}(q, \frac{4}{3}d(q,p_m)) | \geq (1+\frac{1}{c(A)^2})^m \cdot |\bar{B}(q, \frac{1}{3}d(q,p_1)) |$. Let us pick $r = \frac{d(q,p_{m+1})}{3}$. Then we have:
	\begin{ceqn}
		\begin{align*}
			|\bar{B}(q,\frac{4}{3} \cdot d(q,p_{m+1}))| & \geq (1+\frac{1}{c(A)^2}) \cdot |\bar{B}(q,\frac{1}{3} \cdot d(q,p_{m+1}))| \\
			&\geq (1+\frac{1}{c(A)^2}) \cdot |\bar{B}(q,\frac{4}{3} \cdot d(q,p_{m}))| \\
			&\geq (1+\frac{1}{c(A)^2}) \cdot (1+\frac{1}{c(A)^2})^m \cdot |\bar{B}(q,\frac{1}{3} \cdot d(q,p_{1}))| \\
			&\geq (1+\frac{1}{c(A)^2})^{m+1} \cdot |\bar{B}(q,\frac{1}{3} \cdot d(q,p_{1}))|
		\end{align*}
	\end{ceqn}
	which proves the claim. 
\end{proof}

\begin{lem}
	\label{lem:hard_function_bound}
	For every $x \in \R$ satisfying $x \geq 2$, the following inequality holds:
	$$x^2 \geq \frac{1}{\log_2(1 + \frac{1}{x^2})} .$$
\end{lem}
\begin{proof}
Let $\ln$ be natural logarithm. Note first that for any $u > 0$ we have:
$$\frac{u}{u+1} = \int^u_0 \frac{dt}{u+1} \leq \int^u_0 \frac{dt}{t+1} = \ln(u+1).$$
By setting $u = \frac{1}{x^2} > 0$ we get:
$$\log_2(1 + \frac{1}{x^2}) = \frac{\ln(\frac{1}{x^2})}{\ln(2)} \geq \frac{1}{\ln(2)} \cdot \frac{\frac{1}{x^2}}{\frac{1}{x^2}+1} 
= \frac{1}{\ln(2)} \cdot \frac{1}{x^2+1}. $$
Let us now show that for $x \geq 2$ we have: $\frac{1}{\ln(2)} \cdot \frac{1}{x^2+1} \geq \frac{1}{x^2} $.
Note first that $4 \geq \frac{\ln(2)}{1 - \ln(2)}$. Since $x \geq 2$ we have $x^2 \geq \frac{\ln(2)}{1 - \ln(2)}$.
Therefore $x^2 - \ln(2) \cdot x^2 \geq \ln(2)$ and $x^2 \geq \ln(2) \cdot (1 + x^2)$. 
It follows that $\frac{1}{\ln(2)}\frac{1}{1+x^2} \geq \frac{1}{x^2}$, which proves the claim.




\end{proof}

\dfndepth*


\noindent
The new concept of the height $|H(\T)|$ will justify a near-linear parameterized worst-case complexity in Theorem~\ref{thm:knn_KR_time}.
By condition~(\ref{dfn:cover_tree_compressed}b), the height $|H(\T(R))|$ counts the number of levels $i$ whose cover sets $C_i$ include new points that were absent on higher levels.
Then $|H(\T)|\leq|R|$ as any point can be alone at its own level.


\begin{lem}
	\label{lem:depth_bound}
	Any finite set $R$ has the bound $|H(\T(R))|\leq 1+\log_2(\Delta(R))$.
	\bs
\end{lem}
\begin{proof}
	\hypertarget{proof:lem:depth_bound}{\empty} We have $|H(\T(R))|\leq l_{\max} - l_{\min}+1$ by Definition~\ref{dfn:depth}. 
	We estimate $l_{\max} - l_{\min}$ as follows.
	\medskip
	
	\noindent
	Let $p \in R$ be a point such that $\rad(R) = \max_{q \in R}d(p,q)$. 
	Then $R$ is covered by the closed ball $\bar B(p; \rad(R))$.
	Hence the cover set $C_i$ at the level $i=\log_2(\rad(R))$ consists of a single point $p$. 
	The separation condition in Definition~\ref{dfn:cover_tree_compressed} implies that 
	$l_{\max}\leq \log_2(d_{\max}(R))$.
	Since any distinct points $p,q \in R$ have $d(p,q)\geq d_{\min}(R)$, the covering condition implies that no new points can enter the cover set $C_i$ at the level $i=[\log_2(d_{\min}(R))]$, so $l_{\min}\geq\log_2(d_{\min}(R))$.
	Then
	$|H(\T(R))| \leq 1+l_{\max} - l_{\min} \leq 
	1+\log_2(\frac{\rad(R)}{d_{\min}(R)})$.
\end{proof}

\noindent
If the aspect ratio $\Delta(R) = O(\text{Poly}(|R|))$ polynomially depends on the size $|R|$, then $|H(\T(R))| \leq O(\log(|R|))$.
Lemma \ref{lem:growth_bound} corresponds \citet[Lemma~4.2]{beygelzimer2006cover} with slightly modified notation.

\section{The minimized expansion constant in a normed vector space on $\R$}
\label{sec:minimized_exp_constant}

In this section, main Theorem \ref{thm:normed_space_exp_constant} will show that, for any finite subset $R$ of a normed vector space $(\R^n, \Vert \cdot \Vert )$, the minimized expansion constant from Definition \ref{dfn:expansion_constant} has the upper bound $2^n$, so 
$$c_m(R)  = \inf\limits_{0 < \xi}\inf\limits_{R\subseteq A\subseteq \R^{n}}\sup\limits_{p \in A,t > \xi}\dfrac{|\bar{B}(p,2t) \cap A|}{|\bar{B}(p,t) \cap A|} \leq 2^{n}.$$ 
The proof of Theorem \ref{thm:normed_space_exp_constant} is based on the volume argument.
We briefly explain the idea before giving the proof later.
For this purpose, we recall the definition of the Lebesgue measure in Definition \ref{dfn:lebesgue_measure}.

\medskip
\noindent
In Definition \ref{dfn:voronoi_region} we define concepts of grid $G(\delta) = \delta \cdot \Z^{n}$ and cubic regions 
$\bar{V}(p,\delta) = p + [-\frac{\delta}{2}, \frac{\delta}{2}]^n$. For every $\delta > 0$ we define grid extension $U(\delta)$ of $R$ 
as set $U(\delta) = (G(\delta) \setminus f(R)) \cup R$, where $f:R \rightarrow G(\delta)$ is used to replace points of $R$ with their nearest neighbors in grid $G(\delta)$. 

\medskip
\noindent
Note that $\xi$ in the definition of $c_m(R)$ acts as a low bound for radius $t > \xi$.
Let $\gamma > 0$ be a constant, that depends on dimension $n$ and norm $\Vert \cdot \Vert$.
In Lemma \ref{lem:u_bounds} it is shown that if $\delta$ satisfies
$0 < \delta <  \frac{\xi}{\gamma}$, then for any $p \in U(\delta)$ and $t > \xi$ we can bound $|\bar{B}(p,t) \cap U(\delta)|$ as follows:
$$\frac{\mu( \bar{B}(p,t - \delta \cdot \gamma))}{\delta^n} \leq |\bar{B}(p,t) \cap U(\delta)| \leq \frac{\mu( \bar{B}(p,t + \delta \cdot \gamma))}{\delta^n}. $$
 Therefore 
$$\frac{|\bar{B}(p,2t) \cap U(\delta)|}{|\bar{B}(p,t) \cap U(\delta)|} \leq \frac{\mu( \bar{B}(p,2t + \delta \cdot \gamma))}{\mu( \bar{B}(p,t - \delta \cdot \gamma))}.$$
Now since this inequality is satisfied for any $\delta > 0$, we can choose arbitrary dense grid extension $U(\delta)$.
It will be shown that when $\delta \rightarrow 0$ , then $$\frac{\mu( \bar{B}(p,2t + \delta \cdot \gamma))}{\mu( \bar{B}(p,t - \delta \cdot \gamma))} \rightarrow 2^{n}.$$
Then we can conclude that $c_m(R) \leq 2^{n} $.

\begin{dfn}[Normed vector space $(\R^n,\Vert\cdot\Vert)$ on real numbers $\R$ {\citet{Rudin1990-tp}}]
	\label{dfn:normed_rn_space}
	Consider $\R^{n}$ as a vector space. A \emph{norm} is a function $\Vert \cdot \Vert:\R^{n} \rightarrow \R$ satisfying the properties below.
	\begin{enumerate}
		\item Non-negativity : $\Vert x \Vert \geq 0$.
		\item The norm is positive on nonzero vectors, so $\Vert x \Vert = 0$ implies that $x = 0$.
		\item For every vector $x \in \R^{n}$, and every scalar $a \in \R$: $\Vert a \cdot x \Vert = |a| \cdot \Vert x \Vert$.
		\item The triangle inequality holds for every $x \in \R^{n}$ and $y \in \R^{n}$, $\Vert x + y \Vert \leq \Vert x \Vert + \Vert y \Vert$.
	\end{enumerate}
	A norm induces a metric by the formula $d(x,y) = \Vert x - y \Vert$. For every $i \in \{1,...,n\}$ let $e_i$ be a unit vector of $\R^n$ i.e. $e_i(i) = 1$ and $e_i(j) = 0$ for all $j \in \{1,...,n\} \setminus \{i\}$. Define $\rho = \max_{i \in \{1,...,n\}}\Vert e_i \Vert$. \bs 
\end{dfn}

\begin{dfn}[Lebesgue outer measure, {\citet[Section~2.A]{Jones2000-db}}]
	\label{dfn:lebesgue_measure}
	Let $\R^{n}$ be an $n$-dimensional space. Define $n$-dimensional interval as
	$$I = \{x \in \R^{n} \mid a_i \leq x_i \leq b_i, i = 1,...,n\} = [a_1,b_1] \times ... \times [a_n,b_n],$$
	with sides parallel to the coordinate axes. Define Lebesgue outer measure $\mu^{*}:\{A \mid A \subseteq \R^{n}\} \rightarrow [0,\infty) \cup 
	\{\infty\}$ 
	of interval $I$ as 
	$\mu^{*}(I) = (b_1-a_1) \cdot ... \cdot (b_n-a_n)$. The Lebesgue $\mu$ measure of a set $A \subseteq \R^n$ is defined as:
	$$\mu^{*}(A) = \inf_{A} \{\sum^{\infty}_{i = 0}\mu^{*}(I_i) \mid A \subseteq \cup^{\infty}_{i = 0}I_i\}, $$
	where the infinium is taken over all covering of $A$ by countably many intervals $I_i, i = 1,2...$.
	If set $E \subseteq \R^{n}$ satisfies $\mu^{*}(A) = \mu^{*}(A \cap E) + \mu^{*}(A \setminus E)$
	for all $A \subseteq \R^{n}$, then $E$ is lebesgue-measurable and we set $\mu(E) = \mu^{*}(E)$. 
	\bs
\end{dfn}
\noindent
It should be noted that all open sets and closed sets , as well as compact sets are Lebesgue-measurable. 

\begin{lem}[Basic properties of Lebesgue measure, {\citet[Section~2.A]{Jones2000-db}}]
	\label{lem:basic_properties_lebesgue_measure}
	A Lebesgue outer measure $\mu^{*}$ of Definition \ref{dfn:lebesgue_measure} satisfies the following conditions:
	\begin{enumerate}
		\item $\mu^{*}(\emptyset) = 0,$
		\item $\mu^{*}(A) \leq \mu^{*}(B)$ whenever $A \subseteq B \subseteq \R^{n}$ and
		\item $\mu^{*}( \cup^{\infty}_{i = 1} \mu^{*}(A_i)) \leq \sum^{\infty}_{i=1}\mu^{*}(A_i) $.
	\end{enumerate}
\end{lem}

\begin{lem}[Lebesgue measure scale property, {\citet[Section~3.B]{Jones2000-db}}]
	\label{lem:scale_property_lebesgue_measure}
	Let $\mu$ be Lebesgue measure on a normed vector space $(\R^{n},\Vert\cdot\Vert)$.
Then, for any $p \in \R^{n}$ and $t \in \R_{+}$, we have:
	$\mu(\bar{B}(p,t)) =  t^{n} \cdot \mu (\bar{B}(p,1) )$.
\end{lem}

\begin{dfn}[Grid and Cubic region]
	\label{dfn:voronoi_region}
	Let $\R^{n}$ be a normed vector space
	and let $\delta \in \R$. 
	Define $\delta$-grid on $\R^{n}$ as the set
	$G(\delta) = \{ t \cdot \delta  \mid t \in \Z^{n} \}$. For any $p \in \R^{n}$ define its open cubic region $V(p,\delta) \subseteq \R^{n}$ as the set $\{ p + u \mid u \in (-\frac{\delta}{2},\frac{\delta}{2})^n \}$ and closed cubic region $\bar{V}(p,\delta) \subseteq \R^{n}$ as $\{ p + u \mid u \in [-\frac{\delta}{2},\frac{\delta}{2}]^n \}$.
\end{dfn}
\noindent
Note that the union $\cup_{p \in G(\delta)}V(p, \delta)$ covers whole set $\R^{n}$. 

\begin{lem}[Cubic regions are separate]
	\label{lem:open_voronoi_region_separation}
	In conditions of Definition \ref{dfn:voronoi_region} let $p, q \in G(\delta)$ be distinct points. 
	Then their cubic regions are separate i.e. $V(p,\delta) \cap V(q, \delta) = \emptyset$.
\end{lem}
\begin{proof}
	Assume contrary that there exists $a \in V(p,\delta) \cap V(q, \delta)$, then $|a(i) - p(i)| < \frac{\delta}{2}$ and $|a(i) - q(i)| < \frac{\delta}{2}$
	for all $i \in \{1,...,n\}$. Since $p \neq q$, there exists index $j$, such that $p(j) \neq q(j)$. By definition of grid $G(\delta)$ it follows that
	$|p(j) - q(j)| \geq \delta$. However, by triangle inequality we have 
	$$|p(j) - q(j)| \leq |p(j) - a(j) | + |q(j) - a(j) |< \frac{\delta}{2} + \frac{\delta}{2} = \delta,$$
	which is a contradiction. Therefore $V(p,\delta) \cap V(q, \delta) = \emptyset$. 
	
\end{proof}

\begin{lem}
	\label{lem:diam_voronoi_region}
	Let $\R^n$ be a normed vector space of Definition \ref{dfn:normed_rn_space}. 
	Let $\delta \in \R$ and let $G(\delta)$ be a grid of Definition \ref{dfn:voronoi_region}. 
	Let $p \in G(\delta)$ and let $q \in V(p,\delta)$, then $d(p,q) \leq \frac{ n \cdot \delta \cdot \rho}{2}$
\end{lem}
\begin{proof}
	Let $\gamma \in \R$ be such that $q = p + \gamma$. By condition (3) of Definition \ref{dfn:normed_rn_space} we have
	$\Vert \gamma(i) \Vert \leq \Vert e_i \Vert \cdot \frac{\delta }{2} \leq \frac{\delta \cdot \rho}{2}$ for all $i \in \{1,...,n\}$.
	By the definition of norm and triangle inequality we have:
	$$d(p,q) = \Vert p -q \Vert = \Vert\gamma\Vert \leq \sum^n_{i = 1} \Vert \gamma(i) \Vert \leq \frac{n \cdot \delta \cdot \rho}{2}.$$
\end{proof}

\noindent
Any normed vector space $(\R^{n},\Vert\cdot\Vert)$ has the metric $d(x,y) = \Vert x - y \Vert$.

\begin{lem}[Existence of covering grid ]
	\label{lem:covering_grid_existence}
	Let $R$ be a finite subset of a normed vector space $(\R^{n}, \Vert\cdot\Vert)$. 
	Then for any $\delta \in \R$ having $\delta < \frac{d_{\min}(R)}{n \cdot \rho}$, then any map
	$f: R \rightarrow G(\delta)$ which maps $p \in R$ to one of its nearest neighbor in $G(\delta)$ is a well-defined injection.
\end{lem}
\begin{proof}
	Let $f$ be an arbitrary map $f: R \rightarrow G(\delta)$ mapping point $p \in R$ to one of its nearest neighbors. 
	This map is clearly well-defined. Let us now show that it is injective. Assume that $x,y \in R$ are such that
	$f(x) = f(y)$. Then by triangle inequality and Lemma \ref{lem:diam_voronoi_region} we have:
	$$d(x,y) \leq d(x,p) + d(p,y) \leq n \cdot \delta \cdot \rho < d_{\min}(R),$$
	it follows that $x = y$. Therefore map $f$ is injective.

	
\end{proof}

\begin{lem}
	\label{lem:covering_grid_inequalities}
	Let $R$ be a finite subset of normed space $(\R^{n},d)$, let $\rho$ be as in Definition \ref{dfn:normed_rn_space} and let $\delta \in \R$ be such that $0 < \delta < \frac{d_{\min}(R)}{n \cdot \rho}$.
	Let $ p \in R$ be arbitrary point and let $t > \frac{n \cdot \delta \cdot \rho}{2}$ be a real number.
	Then there exists a set $U(\delta)$ satisfying $R \subseteq U(\delta)$ and
	$$|G(\delta) \cap \bar{B}(p , t - \frac{n \cdot \delta \cdot \rho}{2}) | \leq |U(\delta) \cap \bar{B}(p, t)| 
	\leq |G(\delta) \cap \bar{B}(p , t + \frac{n \cdot \delta \cdot \rho}{2}) |$$
\end{lem}
\begin{proof}
	Let $f: R \rightarrow G(\delta)$ be an injection from Lemma \ref{lem:covering_grid_existence}, which maps every $q \in R$ to one of its nearest neighbors in $G(\delta)$.
	Define $U(\delta) = (G(\delta) \setminus f(R)) \cup R$.
	Let us first show that $$g:U(\delta) \cap \bar{B}(p, t) \rightarrow G(\delta) \cap \bar{B}(p , t + \frac{n \cdot \delta \cdot \rho}{2}),$$
	defined by $g(q) = f(q)$, if $q \in R$ and $g(q) = q$, if $q \notin R$,  is an injection. Let us show first that the map $g$ is well-defined, if $q \notin R$, the claim is trivial. Let $q \notin R$, then by triangle inequality 
	$d(g(q),p) \leq  d(q,p) + d(g(q), q) \leq t + \frac{n \cdot \delta \cdot \rho}{2}$. Assume now that $g(a) = g(b)$ for some $a, b \in U(\delta) \cap \bar{B}(p,t)$. Let us first show that either $a,b$ both belong to $R$ or neither of $a,b$ belong to $R$. Assume contrary that $a \in R$ and $b \notin R$.
	Since $b \notin R$ we have $b \in G(\delta) \setminus f(R)$. On the other hand since $h(a) = h(b)$ we have $f(a) = b$, therefore $b \in f(R)$, which is a contradiction. If both, $a$ and $b$ belong to $R$ we have $a = b$, similarly if $a ,b \notin R$ we have $a = b$ by injectivity of function $f$.
	Therefore we have now shown that $g$ is well-defined injection. 
	It follows $|U(\delta) \cap \bar{B}(p, t)| \leq  |G(\delta) \cap \bar{B}(p , t + \frac{n \cdot \delta \cdot \rho}{2})|$.
	Let us now show that map 
	$$ h : G(\delta) \cap \bar{B}(p , t - \frac{n \cdot \delta \cdot \rho}{2}) \rightarrow U(\delta) \cap \bar{B}(p,t),$$
	defined by $h(q) = f^{-1}(q)$, if $q \in f(R)$ and $h(q) = q$, if $q \notin f(R)$ is well-defined injection. 
	Let us first show that the map is well-defined. Let $q \in G(\delta) \cap \bar{B}(p , t - \frac{n \cdot \delta \cdot \rho}{2})$, if $q \notin f(R)$ the claim is satisfied trivially. If $q \in f(R)$, then by definition $d(h(q), q) \leq \frac{n \cdot \delta \cdot \rho}{2}$. By using triangle inequality we obtain: $$d(p, h(q)) \leq d(p,q) + d(q,h(q)) \leq t - \frac{n \cdot \delta \cdot \rho}{2} + \frac{n \cdot \delta \cdot \rho}{2} \leq t. $$
	Therefore $h(q) \in U(\delta) \cap \bar{B}(p,t)$. 
	
	\medskip
	
	\noindent
	Let us now show that $h$ is an injection. Let $a,b \in G(\delta) \cap \bar{B}(p , t - \frac{n \cdot \delta}{2}) $ assume that $h(a) = h(b)$, let us show that $a = b$. Let us first show that either $a,b \in f(R)$ or neither of $a,b$ belong to $f(R)$.
	Assume contrary that $a \in f(R)$ and $b \notin f(R)$. Then $h(a) = h(b)$ 
	implies that $f^{-1}(a) = b$. Since $f^{-1}(a) \in R$, we have $b \in R$. Since $b \in G(\delta)$, it follows that $f(b) = b$, which is a contradiction since $b \notin f(R)$. Assume now that $a, b \in f(R)$, then the claim follows by noting that $f^{-1}$ is injection. If $a, b \notin f(R)$ , then claim follows by noting that $h(a) = a$ and $h(b) = b$. Therefore map $h$ is injection. It follows that 
	$|G(\delta) \cap \bar{B}(p , t - \frac{n \cdot \delta \cdot \rho}{2})| \leq |U(\delta) \cap \bar{B}(p,t)|.$

	
\end{proof}

\begin{lem}
	\label{lem:covering_grid_inclusions}
	Let $R$ be a finite subset of normed vector space $\R^{n}$ and let $\delta \in \R$.
	For any $p \in G(\delta)$ recall that $V(p,\delta)$ is Minkowski sum $p + (-\frac{\delta}{2},\frac{\delta}{2})^n$.
	Define $$\bar{W}(p, t, \delta) = \bigcup_{q \in \bar{B}(p,t) \cap G(\delta)} \bar{V}(q,\delta).$$
	Then for any $p \in R$ and $t > \frac{n \cdot \delta \cdot \rho}{2}$ we have:
	$$\bar{B}(p, t - \frac{n \cdot \delta \cdot \rho}{2}) \subseteq \bar{W}(p,t,\delta) \subseteq \bar{B}(p, t + \frac{n \cdot \delta \cdot \rho}{2}).$$
\end{lem}
\begin{proof}
	Let $u \in \bar{B}(p, t-\frac{n \cdot \delta \cdot \rho}{2})$ be an arbitrary point. Since $\{\bar{V}(q,\delta) \mid q \in G(\delta)\}$ covers $R$ it follows that there exists
	$a \in G(\delta)$ such that $u \in \bar{V}(a, \delta)$. By triangle inequality we obtain:
	$$d(a,p) \leq d(a,u) + d(u,p) \leq \frac{n \cdot \delta \cdot \rho}{n} + t - \frac{n \cdot \delta \cdot \rho}{n} \leq t.  $$
	It follows that $\bar{V}(w,\delta) \in \bar{W}(p,t)$, therefore $p \in \bar{W}(p,t,\delta)$. We have $\bar{B}(p, t - \frac{n \cdot \delta \cdot \rho}{2}) \subseteq \bar{W}(p,t,\delta).$
	Let $u \in \bar{W}(p,t,\delta)$, then there exists $a \in G(\delta)$ such that $u \in \bar{V}(a, \delta)$ and $\bar{V}(a, \delta) \in \bar{W}(p,t)$. By triangle inequality we obtain: 
	$$d(u,p) \leq d(u,a) + d(a,p) \leq \frac{n \cdot \delta \cdot \rho}{n} + t . $$
	It follows that $u \in \bar{B}(p, t + \frac{n \cdot \delta \cdot \rho}{2})$. Therefore 
	$\bar{W}(p,t,\delta) \subseteq \bar{B}(p, t + \frac{n \cdot \delta \cdot \rho}{2}) $ which proves the claim.
	
	
\end{proof}

\begin{lem}[Countable additivity, {\citet[Section~2.A]{Jones2000-db}}]
	\label{lem:measure_additivty}
	Assume that $A_i \subseteq \R^{n}$, $i = 1,2,...,$ are pairwise disjoint i.e. $A_i \cap A_j = \emptyset$ for all $i \neq j$ Lebesgue-measurable sets. Then $$\mu(\bigcup^{\infty}_{i = 0}A_i) = \sum^{\infty}_{i = 1}\mu(A_i) .$$
\end{lem}
\begin{lem}[Lebesgue measure of $\bar{W}(p,t,\delta)$]
\label{lem:lebesgue_measure_w}
	In notations of Lemma \ref{lem:covering_grid_inclusions}
	let $\mu$ be a Lebesgue measure on $R$ from Definition \ref{dfn:lebesgue_measure}, then $\mu(\bar{W}(p,t,\delta)) = \delta^n \cdot |\bar{B}(p,t) \cap G(\delta)|$.
\end{lem}
\begin{proof}
	Define $W(p, t, \delta) = \bigcup\limits_{q \in \bar{B}(p,t) \cap G(\delta)} V(q,\delta).$
	Recall that for all $p \in \R^{n}$ and $\delta > 0$ set $\bar{V}(p,t)$ is a closed $n-$dimensional interval and ${V}(p,t)$ is an open $n$-dimensional interval. Therefore we have $\mu(\bar{V}(p,t)) = \mu(V(p,t))$. Since $\bar{V}(p,t)$ is a closed interval, it follows that $\mu(\bar{V}(p,t)) = \delta^{n}$. 
	Since all the sets of $W$ are separate we can use Lemma \ref{lem:measure_additivty} to obtain:
	$$\mu(W(p,t,\delta))  = \sum_{A \in W(p,t)} \mu(A) =  \sum_{A \in \bar{W}(p,t)} \mu(A) = \delta^{n} \cdot |\bar{B}(p,t) \cap G(\delta)|$$
	By Lemma \ref{lem:basic_properties_lebesgue_measure} (2), since $W(p,t,\delta) \subseteq \bar{W}(p,t,\delta)$ we obtain $\mu(\cup\bar{W}(p,t)) \geq \delta^{n} \cdot |\bar{B}(p,t) \cap G(\delta)|$. On the other hand, by Lemma \ref{lem:basic_properties_lebesgue_measure} (3) we obtain 
	$$\mu(\bar{W}(p,t,\delta)) \leq \sum_{A \in \bar{W}(p,t)} \mu(A) = \delta^{n} \cdot |\bar{B}(p,t) \cap G(\delta)| $$
	Therefore we have shown that $\mu(\bar{W}(p,t,\delta)) = \delta^{n} \cdot |\bar{B}(p,t) \cap G(\delta)|$.
\end{proof}

\begin{lem}[Set $U(\delta)$ bounds]
	\label{lem:u_bounds}
	Let $\R^{n}$ be a normed vector space.
	Let $R \subseteq \R^{n}$ be its finite subset.
	Then any set $U(\delta)$ of Lemma \ref{lem:covering_grid_inequalities} satisfies the following inequalities:
	$$\frac{\mu( \bar{B}(p,t - \delta \cdot n \cdot \rho))}{\delta^n} \leq |\bar{B}(p,t) \cap U(\delta)| \leq \frac{\mu( \bar{B}(p,t + \delta \cdot \gamma \cdot n \cdot \rho))}{\delta^n}, $$
	for all $p \in R$ and $t > n \cdot \delta \cdot \rho$.
\end{lem}
\begin{proof}

Let $p \in \R^{n}$ be an arbitrary point and let $t > n \cdot \delta \cdot \rho$ be an arbitrary real number.
By Lemma \ref{lem:covering_grid_inequalities} it follows:
$$|G(\delta) \cap \bar{B}(p , t + \frac{n \cdot \delta \cdot \rho}{2})| \leq |\bar{B}(p,t) \cap U(\delta)| \leq |G(\delta) \cap \bar{B}(p , t + \frac{n \cdot \delta \cdot \rho}{2})|.$$
Let $ \bar{W}(p,t + \frac{n \cdot \delta \cdot \rho}{2},\delta) = \cup_{q} \{ \bar{V}(q,\delta) \mid q \in \bar{B}(p, t + \frac{n \cdot \delta \cdot \rho}{2}) \}.$ By Lemma \ref{lem:covering_grid_inclusions} we have:  
$$
\bar{B}(p, t - n \cdot \delta \cdot \rho) \subseteq \bar{W}(p,t - \frac{n \cdot \delta \cdot \rho}{2},\delta) \text{ and }
\bar{W}(p,t + \frac{n \cdot \delta \cdot \rho}{2},\delta) \subseteq \bar{B}(p, t + n \cdot \delta \cdot \rho)
$$
By Lemma \ref{lem:basic_properties_lebesgue_measure} we have $\mu(\bar{W}(p,t + \frac{n \cdot \delta \cdot \rho}{2},\delta)) \leq \mu(\bar{B}(p, t + n \cdot \delta \cdot \rho))$. By Lemma \ref{lem:lebesgue_measure_w} we have:
$$\mu(\bar{W}(p,t + \frac{n \cdot \delta \cdot \rho}{2},\delta)) = \delta^n \cdot |\bar{B}(p,t + \frac{n \cdot \delta \cdot \rho}{2}) \cap G(\delta)|.$$
By combining the facts we obtain:
$$ |\bar{B}(p,t) \cap U(\delta)| \leq |G(\delta) \cap \bar{B}(p , t + \frac{n \cdot \delta \cdot \rho}{2})| \leq \frac{\mu(\bar{W}(p,t + \frac{n \cdot \delta \cdot \rho}{2},\delta))}{\delta^n} \leq \frac{\mu(\bar{B}(p, t + n \cdot \delta \cdot \rho))}{\delta^n} $$
$$ |\bar{B}(p,t) \cap U(\delta)| \geq |G(\delta) \cap \bar{B}(p , t - \frac{n \cdot \delta \cdot \rho}{2})| \geq \frac{\mu(\bar{W}(p,t - \frac{n \cdot \delta \cdot \rho}{2},\delta))}{\delta^n} \geq \frac{\mu(\bar{B}(p, t - n \cdot \delta \cdot \rho))}{\delta^n} $$
which concludes the proofs.

\end{proof}


\begin{lem}[Set $U(\delta)$ is locally finite]
	\label{lem:grid_locally_finite}
	Let $\R^{n}$ be a normed vector space.
	Let $R \subseteq \R^{n}$ be its finite subset
	Then any set $U(\delta)$ from Lemma \ref{lem:covering_grid_inequalities} is locally finite.
\end{lem}
\begin{proof}
With the exact same proof of Lemma \ref{lem:u_bounds} it can be shown that 
$$|\bar{B}(p,t) \cap U(\delta)| \leq \frac{\mu( \bar{B}(p,t + \delta \cdot n \cdot \rho))}{\delta^n}$$ 
is satisfied for all $p \in R$ and $t > 0$.
Therefore $|\bar{B}(p,t) \cap U(\delta)|$ is finite as well.

\end{proof}
\noindent
Recall that \emph{minimized expansion constant} of Definition \ref{dfn:expansion_constant} of a finite subset $R$ of a metric space $(X,d)$ was defined as  $c_m(R) = \lim\limits_{\xi \rightarrow 0^{+}}\inf\limits_{R\subseteq A\subseteq X}\sup\limits_{p \in A,t > \xi}\dfrac{|\bar{B}(p,2t) \cap A|}{|\bar{B}(p,t) \cap A|}$ where $A$ is a locally finite set which covers $R$.

\begin{thmm}[The minimized expansion constant of a finite subset $R$ of $\R^{n}$ is at most $2^{n}$]
	\label{thm:normed_space_exp_constant}
	Let $R$ be a finite subset of a normed Euclidean space $\R^{n}$.
	Let $c_m(R)$ be the minimized expansion constant of Definition~\ref{dfn:expansion_constant}, 
	then $c_m(R) \leq 2^{n}$.
\end{thmm}
\begin{proof}
	Let $0 < \xi < \frac{d_{\min}(R)}{2}$ be an arbitrary real number.  	
	Let $0< \delta < \frac{\xi}{n \cdot \rho}$ be a real number.
	Since $\delta < \frac{d_{\min}(R)}{2 \cdot n \cdot \rho}$ by Lemma \ref{lem:u_bounds} we have: $$\frac{\mu( \bar{B}(p,t - \delta \cdot n \cdot \rho))}{\delta^n} \leq |\bar{B}(p,t) \cap U(\delta)| \leq \frac{\mu( \bar{B}(p,t + \delta \cdot \gamma \cdot n \cdot \rho))}{\delta^n}$$
Note that by Lemma \ref{lem:scale_property_lebesgue_measure} we have: $\mu(\bar{B}(q,y)) =  y^{n} \cdot \mu (\bar{B}(q,1) )$ for any $q \in \R^n$ and $y \in \R_{+}$. Therefore 
	$$ \frac{|\bar{B}(p, 2t) \cap U(\delta)|}{|\bar{B}(p, t) \cap U(\delta)|} \leq \frac{\mu(\bar{B}(p, 2t + n\delta\rho)) \cdot \delta^2}{\mu(\bar{B}(p, t - n\delta\rho)) \cdot \delta^2} = \frac{ (2t + n\delta\rho)^n \cdot \mu(\bar{B}(p,1)) }{ (t - n\delta\rho)^n \cdot \mu(\bar{B}(p,1)) } 
	= \frac{(2t + n\delta\rho)^n}{(t - n\delta\rho)^n},$$
	is satisfied for for all $t > \xi$.
	Since $0 < \xi < \frac{d_{\min}(R)}{2}$ was chosen arbitrarily, we conclude that:  $$c_m(R) =  \lim\limits_{\xi \rightarrow 0^{+}}\inf\limits_{R\subseteq A\subseteq X}\sup\limits_{p \in A,t > \xi}\dfrac{|\bar{B}(p,2t)| \cap A}{|\bar{B}(p,t)| \cap A} \leq \lim_{\delta \rightarrow 0} \frac{|\bar{B}(p, 2t) \cap U(\delta)|}{|\bar{B}(p, t) \cap U(\delta)|} = \lim_{\delta \rightarrow 0} \frac{(2t + \delta \cdot n\rho)^n}{(t - \delta \cdot n\rho)^n} = \frac{2^n \cdot t^n}{t^n} = 2^{n}.$$

\end{proof}

\section{Distinctive descendant sets}
\label{sec:distinctive_descendant_set}

This section introduces auxiliary concepts for future proofs. 
The main concept is a distinctive descendant set in Definition \ref{dfn:distinctive_descendant_set}. 
The distinctive descendant set at a level $i$ of a node $p \in \T(R)$ in a compressed cover tree corresponds to the set of descendants of a copy of node $p$ at level $i$ in the original implicit cover tree $T(R)$.
Other important concepts are $\lambda$-point of Definition \ref{dfn:lambda-point} that is used in Algorithm \ref{alg:cover_tree_k-nearest} as an approximation for $k$-nearest neighboring point. The $\beta$-point property of $\lambda$-point defined in Lemma \ref{lem:beta_point} plays a major role in the proof of the main worst-case time complexity result Theorem \ref{thm:knn_KR_time}.

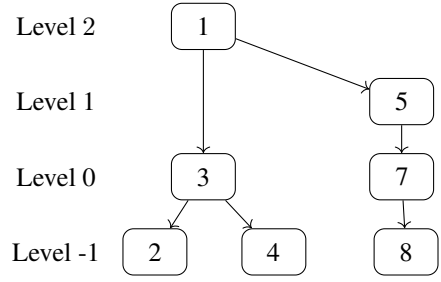
\begin{figure}
	\centering
	\begin{tikzpicture}[align=center, node distance = 1.0cm, scale = 0.45]

	\node (scalem2text) {Level $2$};
	\node[below of =scalem2text] (scalem1text) {Level 1};
	\node[below of =scalem1text] (scalem0text) {Level 0};
	\node[below of =scalem0text] (scalemm1text) {Level -1};

	\node [blockz,  right=25pt of scalem2text ] (node1) {1};
	\node [blockz,  right=100pt of scalem1text] (node5) {5};
	\node [blockz,  right=25pt of scalem0text] (node3) {3};
	\node [blockz,  right=5pt of scalemm1text] (node2) {2};
	\node [blockz,  right=50pt of scalemm1text] (node4) {4};
	
	\node [blockz,  right=100pt of scalem0text] (node7) {7};
	\node [blockz,  right=100pt of scalemm1text] (node8) {8};

	  \draw[->] (node1) -> (node5);
	  \draw[->] (node1) -> (node3);
	  \draw[->] (node3) -> (node2);
	  \draw[->] (node3) -> (node4);
	  \draw[->] (node5) -> (node7);
	  \draw[->] (node7) -> (node8);

\end{tikzpicture} 
	\caption{Consider a compressed cover tree $\T(R)$ that was built on set $R = \{1,2,3,4,5,7,8\}$. Let $\Sd_i(p, \T(R))$ be a distinctive descendant set of Definition \ref{dfn:distinctive_descendant_set}. Then $V_2(1) = \emptyset, V_{1}(1) = \{5\}$ and $V_{0}(1) = \{3,5,7\}$.
		And also $\Sd_2(1, \T(R)) = \{1, 2,3,4,5,7,8\}$, $\Sd_1(1, \T(R)) = \{1,2,3,4\} $ and $\Sd_{0}(1, \T(R)) = \{1\} $.}
	\label{fig:uniqueDescendant}
\end{figure}

\dfndistinctivedescendantset*
\begin{lem}[Distinctive descendant set inclusion property]
	\label{lem:distinctive_descendant_inclusion}
	In conditions of Definition \ref{dfn:distinctive_descendant_set} let $p \in R$ and let $i,j$ be integers satisfying
	$l_{\min}(\T(R)) \leq i \leq j \leq l(p) - 1$.
	Then $\Sd_i(p, \T(R)) \subseteq \Sd_j(p, \T(R))$.
\end{lem}


Essential levels of a node $p \in \T(R)$ have 1-1 correspondence to the set consisting of all nodes containing $p$ in the explicit representation of cover tree in \cite{beygelzimer2006cover}, see Figure \ref{fig:tripleexample} middle. 

\begin{dfn}[Essential levels of a node]
	\label{dfn:essential_levels_node}
	Let $R\subseteq X$ be a finite reference set with a cover tree $\T(R)$. Let $q \in \T(R)$ be a node. Let $(t_i)$ for 
	$i \in \{0,1,...,n\}$ be a sequence of $H(\T(R))$ in such a way that $t_0 = l(q)$, $t_n = l_{\min}(\T(R))$ and for all $i$ we have $t_{i+1} = \nxt(q, t_{i}, \T(R))$. Define the set of essential indices $\Es(q,\T(R)) = \{ t_{i} \mid i \in \{0,...,n\} \}$.
	\bs
\end{dfn}

\begin{lem}[Number of essential levels]
	\label{lem:number_of_explicit_levels}
	Let $R\subseteq X$ be a finite reference set with a cover tree $\T(R)$. Then 
	$\sum_{p \in R}|\mathcal{E}(p,\T(R))| \leq 2 \cdot |R|,$
	where $\mathcal{E}(p,\T(R))$ appears in Definition \ref{dfn:essential_levels_node}. \bs
\end{lem}
\begin{proof}
	Let us prove this claim by induction on size $|R|$. In basecase $R = \{r\}$ and therefore $|\mathcal{E}(r,\T(R))| = 1$.
	Assume now that the claim holds for any tree $\T(R)$, where $|R| = m$ and let us prove that if we add any node $v \in X \setminus R$ to tree $\T(R)$, then $\sum_{p \in R}|\mathcal{E}(p,\T(R \cup \{v\}))| \leq 2 \cdot |R| + 2$. Assume that we have added $u$ to $\T(R)$, in such a way that $v$ is its new parent. Then $|\mathcal{E}(p,\T(R \cup \{v\}))| = |\mathcal{E}(p,\T(R))| + 1 $ and $|\mathcal{E}(v,\T(R \cup \{v\}))| = 1$. We have:
	$$ \sum_{p \in R \cup \{u\}}|\mathcal{E}(p,\T(R))| = \sum_{p \in R}|\mathcal{E}(p,\T(R))| + 1 + |\mathcal{E}(v,\T(R \cup \{v\}))| \leq 2\cdot |R| + 2 \leq 2(|R \cup \{v\}|) $$
	which completes the induction step. 
\end{proof}

\begin{algorithm}
	\caption{This algorithm returns sizes of distinctive descendant set $\Sd_i(p, \T(R))$ for all essential levels $i \in \Es(p,\T(R))$}
	\label{alg:cover_tree_distinctive_descendants}
	\begin{algorithmic}[1]
		\STATE \textbf{Function} : CountDistinctiveDescendants(Node $p$, a level $i$ of $\T(R)$)
		\STATE \textbf{Output} : an integer 
		
		\IF{$i > l_{\min}(\T(Q))$}
		\FOR {$q \in \text{Children}(p)$ having $l(p) = i-1$ or $q = p$}
		\STATE Set $s = 0$
		\STATE $j \leftarrow 1 + \nxt(q, i-1,\T(R))$
		\STATE $s \leftarrow s + $CountDistinctiveDescendants($q$, $j$)
		\ENDFOR
		\ELSE
		\STATE Set $s = 1$
		\ENDIF
		\STATE Set $|\Sd_i(p)| = s$ and \textbf{return} s
	\end{algorithmic}
\end{algorithm}

\begin{lem}
	\label{lem:distinctive_descendants_precompute}
	Let $R$ be a finite subset of a metric space.
	Let $\T(R)$ be a compressed cover tree on $R$. 
	Then, Algorithm~\ref{alg:cover_tree_distinctive_descendants} computes the sizes $|\Sd_{i}(p, \T(R))|$ for all $p \in R$ and essential levels $i \in \Es(p,\T(R))$ in time $O(|R|)$.
	\bs
\end{lem}
\begin{proof}
	By Lemma \ref{lem:number_of_explicit_levels} we have $\sum_{p \in R}|\mathcal{E}(p,\T(R))| \leq 2 \cdot |R|.$ Since CountDistinctiveDescendants is called once for every any combination $p \in R$ and $i \in \mathcal{E}(p,\T(R))$ it follows that the time complexity of Algorithm~\ref{alg:cover_tree_distinctive_descendants} is $O(R)$.
	

\end{proof}



Recall that the neighbor set $N(q;r) = \{p \in C \mid d(q,p) \leq d(q,r)\}$ was introduced in Definition~\ref{dfn:kNearestNeighbor}.
\begin{dfn}[$\la$-point]
	\label{dfn:lambda-point}
	Fix a query point $q$ in a metric space $(X,d)$ and fix any level $i \in \Z$. 
	Let $\T(R)$ be its compressed cover tree on a finite reference set $R \subseteq X$. 
	Let $C$ be a subset of a cover set $C_i$ from Definition~\ref{dfn:cover_tree_compressed} satisfying $\sum_{p \in C}|\Sd_i(p, \T(R))| \geq k$, where $\Sd_i(p, \T(R))$ is the distinctive descendant set from Definition \ref{dfn:distinctive_descendant_set}.
	For any $k\geq 1$, define $\la_k(q,C)$ as a point $\la\in C$ that minimizes $d(q,\la)$ subject to $\sum_{p \in N(q;\la)}|\Sd_i(p, \T(R)) |\geq k$. 
	\bs
\end{dfn}

\begin{algorithm}
	\caption{Finding $k$-lowest element of a finite subset $A \subseteq R$ with priority function $f: A \rightarrow \R$}
	\label{alg:k_smallest_elements}
	\begin{algorithmic}[1]
		\STATE \textbf{Input:} Ordered subset $A \subseteq R$, priority function $f:A \rightarrow \R$, an integer $k \in \Z$
		\STATE Initialize an empty max-binary heap $B$ and an empty array $D$ on points $A$.
		\FOR{$p \in A$}
		\STATE add $p$ to $B$ with priority $f(p)$
		\IF{$|H| \geq k$} 
		\STATE remove the point with a maximal value from $B$
		\ENDIF
		\ENDFOR
		\STATE Transfer points from the binary heap $B$ to the array $D$ in reverse order. 
		\STATE \textbf{return} $D$. 
	\end{algorithmic}
\end{algorithm}

\begin{algorithm}
	\caption{Computation of a $\lambda$-point of Definition \ref{dfn:lambda-point} in line \ref{line:knnu:dfnLambda} of Algorithm \ref{alg:cover_tree_k-nearest} }
	\label{alg:lambda}
	\begin{algorithmic}[1]
		\STATE \textbf{Input:} A point $q \in X$, a subset $C$ of a level set $C_i$ of a compressed cover tree $\T(R)$, an integer $k \in \Z$
		\STATE Define $f: C \rightarrow \R$ by setting $f(p) = d(p,q)$. 
		\STATE Run Algorithm \ref{alg:k_smallest_elements} on inputs $(C, f, k)$ and retrieve array $D$. 
		\STATE Find the smallest index $j$ such that $\sum^{j}_{t = 0}|\Sd_i(D[t] , \T(R))| \geq k$.
		\STATE \textbf{return} $\lambda = D[j]$. 
	\end{algorithmic}
\end{algorithm}

\begin{lem}
	\label{lem:time_k_smallest_elements}
	Let $A \subseteq R$ be a finite subset and let $f: A \rightarrow \R$ be a priority function and let $k \in \Z_{+}$. 
	Then Algorithm \ref{alg:k_smallest_elements} finds $k$-smallest elements of $A$ in time $|A| \cdot \log_2(k)$
\end{lem}
\begin{proof}
	Adding and removing element from binary heap data structure \citet[section~6.5]{Cormen1990} takes at most $O(\log(n))$ time, where $n$ is the size of binary heap.
	Since the size of our binary heap is capped at $k$ and we add/remove at most $|A|$ elements, the total time complexity is 
	$O(|A| \cdot \log_2(k))$.
\end{proof}

\begin{lem}[time complexity of a $\lambda$-point]
	\label{lem:time_lambdapoint}
	In the conditions of Definition~\ref{dfn:lambda-point}, the time complexity of Algorithm \ref{alg:lambda} is $O(|C| \cdot \log_2(k))$.
\end{lem}
\begin{proof}
	Note that in line $4$ we have $|\Sd_i(D[t] , \T(R))| \geq 1$ for all $t = 0 , ..., j$. 
	Therefore the time complexity of line $4$ is $O(k)$. 
	By Lemma \ref{lem:time_k_smallest_elements} The time complexity of line $3$ is $O(|C| \cdot \log_2(k))$, which proves the claim. 
\end{proof}

\begin{lem}[separation]
	\label{lem:separation}
	In the conditions of Definition~\ref{dfn:distinctive_descendant_set}, let $p\neq q$ be nodes of $\T(R)$ with $l(p) \geq i$, $l(q) \geq i$. Then $\Sd_i(p , \T(R)) \cap \Sd_{i}(q, \T(R)) = \emptyset$.  \bs
\end{lem}
\begin{proof}
	Without loss of generality assume that $l(p) \geq l(q)$. If $q$ is not a descendant of $p$, the lemma holds trivially due to $\Desc(q) \cap \Desc(p) = \emptyset$. 
	If $q$ is a descendant of $p$, then $l(q) \leq l(p) - 1$ and therefore $q \in V_i(p)$. 
	It follows that
	$\Sd_i(p , \T(R)) \cap \Desc(q) = \emptyset$ and therefore
	$$\Sd_{i}(p, \T(R)) \cap \Sd_{i}(q, \T(R))   \subseteq \Sd_{i}(p, \T(R)) \cap \Desc(q) = \emptyset.$$ 
\end{proof}

\begin{lem}[Sum lemma]
	\label{lem:sum}
	In the notations of Definition~\ref{dfn:distinctive_descendant_set} assume that $i$ is arbitrarily index and a subset $V \subseteq R$ satisfies $l(p) \geq i$ for all $p \in V$. Then 
	$$|\bigcup\limits_{p \in V}\Sd_i(p,\T(R))| = \sum\limits_{p \in V} |\Sd_i(p, \T(R))|.$$
\end{lem}
\begin{proof}
	Proof follows from Lemma \ref{lem:separation}.
\end{proof}

\noindent 
By Lemma \ref{lem:sum} in Definition \ref{dfn:lambda-point} one can assume that $|\bigcup_{p \in C}\Sd_i(p, \T(R)) | \geq k$.

\lemdistinctivedescendantchildlevel*
\begin{proof}
	\hypertarget{proof:lem:distinctive_descendant_child_level}{\empty} Let $w \in \Sd_i(p)$ be an arbitrary node satisfying $w \neq p$. 
	Let $s$ be the node-to-root path of $w$. 
	The inclusion $\Sd_i(p) \subseteq \Desc(p)$ implies that $w \in \Desc(p)$. 
	Let $a \in \Child(p) \setminus \{p\}$ be a child on the path $s$. 
	If $l(a) \geq i$ then $a \in V_i(p)$. 
	Note that $w \in \Desc(a)$.
	Therefore $w \notin \Sd_i(p)$, which is a contradiction.
	Hence $l(a) < i$. 
\end{proof}

\begin{lem}
	\label{lem:distinctive_descendant_distance}
	In the notations of Definition~\ref{dfn:distinctive_descendant_set}, let $p \in \T(R)$ be any node. 
	If $w \in \Sd_i(p,\T(R))$ then $d(w,p) \leq 2^{i+1}$.\bs
\end{lem}
\begin{proof}
	By Lemma \ref{lem:distinctive_descendant_child_level} either $w = \gamma$ or $w \in \Desc(a)$ for some $a \in \Child(\gamma) \setminus \{\gamma\}$ for which $l(a) < i$.
	If $w = \gamma$, then trivially $d(\gamma, w) \leq 2^{i}$. Else $w$ is a descendant of $a$, which is a child of node $\gamma$ on level $i-1$ or below, therefore by Lemma \ref{lem:compressed_cover_tree_descendant_bound} we have  $d(\gamma, w) \leq 2^{i}$ anyway.
\end{proof}

\begin{lem}
	\label{lem:child_set_equivalence}
	Let $R$ be a finite subset of a metric pace.
	Let $\T(R)$ be a compressed cover tree on $R$. 
	Let $R_j \subseteq C_j$, where $C_j$ is the $i$th cover set of $\T(R)$. Let $i = \max_{p \in R_j}\nxt(p,j,\T(R))$.
	Set 
	$\C_i(R_j) = R_{j} \cup \{a \in \Child(p) \text{ for some }p \in R_i \mid l(a) = i \}$.
	Then 
	$$\bigcup_{p \in \C_i(R_j)}\Sd_{i}(p, \T(R)) = \bigcup_{p \in R_j}\Sd_{j}(p, \T(R)).$$
\end{lem}
\begin{proof}
	Let $a \in \bigcup_{p \in \C_i(R_j)}\Sd_{i}(p, \T(R))$ be an arbitrary node. Then there exits $u \in \C_i(R_j)$ having 
	$ a \in \Sd_{i}(u, \T(R))$. 
	By definition of index $i$, either $u \in R_j$ or $u$ has a parent in $R_j$. If $u \in R_j$ then we note that $V_j(u) \subseteq V_i(u)$.
	Since $a \notin V_i(u)$, we also have $a \notin V_j(u)$. 
	
	\medskip
	
	\noindent 
	Otherwise let $w$ be a parent of $u$. Therefore there are no descendants of $w$ in having level in interval $[l(u) + 1 , l(p) - 1]$.
	Since $l(u) = i$ and $j > i$ it follows that $V_{j}(w) = \emptyset$. 
	Denote $w$ to be the lowest level ancestor of $u$ on level $j$. By cases above we have $a \notin V_{j}(w)$.
	Therefore it follows that
	$$a \in \Sd_j(w, \T(R)) \subseteq \bigcup_{p \in R_j}\Sd_j(p , \T(R)).$$
	To prove the converse inclusion assume now that $a \in \bigcup\limits_{p \in R_j}\Sd_{j}(p, \T(R))$. 
	Then $a \in \Sd_{j}(v, \T(R))$ for some $w \in R_j$.
	Assume that $w$ has no children at the level $i$. 
	Then $V_{j}(w) = V_{i}(w)$ and 
	$$a \in \Sd_{i}(w, \T(R)) \subseteq \bigcup_{p \in \C_i(R_j)}\Sd_{i}(p ,\T(R)).$$ 
	Assume now that $w$ has children at the level $i$.
	If there exists $b \in \Child(w)$ for which $a \in \Desc(b)$. 
	Since $V_{i}(b) = \emptyset$, we conclude that 
	$$a \in \Sd_{i}(b, \T(R)) \subseteq \bigcup_{p \in \C_i(R_j)}\Sd_{i}(p ,\T(R)).$$ 
	Assume that $a \notin \Desc(b)$ for all $b \in \Child(w)$ with $l(b) = i$.
	Then $a \in \Desc(w)$ and $a \notin \Desc(b')$ for any $b' \in V_j(w)$. Then $a \in \Sd_{i}(w, \T(R))$ and the proof finishes:
	$$\bigcup_{p \in R_j}\Sd_{j}(p, \T(R))   \subseteq   \bigcup_{p \in \C_i(R_j)}\Sd_{i}(p, \T(R)).  $$
\end{proof}

\begin{lem}[$\be$-point]
	\label{lem:beta_point}
	In the notations of Definition~\ref{dfn:lambda-point}, let $C\subseteq C_i$ so that $\cup_{p \in C}\Sd_i(p, \T(R))$ contains all $k$-nearest neighbors of $q$. 
	Set $\la = \la_k(q,C)$. 
	Then $R$ has a point $\beta$ among the first $k$ nearest neighbors of $q$ such that $d(q,\lambda) \leq  d(q,\beta) + 2^{i+1}$.\bs
\end{lem}
\begin{proof}
	We show that $R$ has a point $\beta$ among the first $k$ nearest neighbors of $q$ such that
	$$\beta \in \bigcup_{p \in C}\Sd_i(p, \T(R)) \setminus \bigcup_{p \in N(q, \lambda) \setminus \{\lambda\} }\Sd_i(p, \T(R)).$$
	Lemma~\ref{lem:sum} and Definition~\ref{dfn:lambda-point} imply that
	$$ | \bigcup_{p \in N(q, \lambda) \setminus \{\lambda\} }\Sd_i(p, \T(R)) |= \sum_{p \in N(q, \lambda) \setminus \{\lambda\} }| \Sd_i(p, \T(R)) | < k.$$ 
	Since $\cup_{p \in C}\Sd_i(p, \T(R))$ contains all $k$-nearest neighbors of $q$, a required point $\beta\in R$ exists. 
	\medskip
	
	\noindent 
	Let us now show that $\beta$ satisfies $d(q,\lambda) \leq  d(q,\beta) + 2^{i+1}$.
	Let $\gamma \in C \setminus N(q,\lambda) \cup \{\lambda \}$ be such that $\beta \in \Sd_i(\gamma, \T(R))$. 
	Since $\gamma \notin N(q,\lambda) \setminus \{\lambda\}$, we get $d(\gamma, q) \geq d(q, \lambda)$.  
	The triangle inequality says that $ d(q, \gamma) \leq d(q,\beta) + d(\gamma ,\beta) $.
	Finally Lemma~\ref{lem:distinctive_descendant_distance} implies that $d(\gamma, \beta) \leq 2^{i+1}$. 
	Then
	$$d(q,\lambda) \leq d(q, \gamma) \leq d(q,\beta) + d(\gamma ,\beta) \leq d(q,\beta) +2^{i+1}$$
	So $\beta$ is a desired $k$-nearest neighbor satisfying $d(q,\lambda) \leq  d(q,\beta) + 2^{i+1}$.
\end{proof}

\section{Construction of a compressed cover tree}
\label{sec:ConstructionCovertree}

This section introduces a new method Algorithm \ref{alg:cover_tree_k-nearest_construction_whole} for construction of a compressed cover tree, which is based on Insert() method \citet[Algorithm~2]{beygelzimer2006cover} that was specifically adapted for compressed cover tree. The proof of  \citet[Theorem~6]{beygelzimer2006cover}, which estimated the time complexity of \citet[Algorithm~2]{beygelzimer2006cover} was shown to be incorrect  \citet[Counterexample~4.2]{elkin2022counterexamples}. The main contribution of this section are two new time complexity results that bound the time complexity of Algorithm \ref{alg:cover_tree_k-nearest_construction_whole}:
\begin{itemize}
    \item Theorem \ref{thm:construction_time} bounds the time complexity as 
    $O(c_m(R)^{10} \cdot \log_2(\Delta(R)) \cdot |R|)$ by using minimized expansion constant $c_m(R)$ and aspect ratio $\Delta(R)$ as parameters.
    \item  Theorem \ref{thm:construction_time_KR} bounds the time complexity as $O(c(R)^{12} \cdot \log_2|R| \cdot |R|)$ by using expansion constant $c(R)$ as parameter.
\end{itemize}
Definition~\ref{dfn:implementation_compressed_cover_tree} explains the concrete implementation of compressed cover tree.


\dfnimplementation*

\begin{dfn}[construction iteration set $L(\T(W),p)$]
	\label{dfn:cover_tree_construction_iteration_set}
	Let $W$ be a finite subset of a metric space $(X,d)$. 
	Let $\T(W)$ be a cover tree of Definition \ref{dfn:cover_tree_compressed} built on $W$ and let $p \in X \setminus W$ be an arbitrary point.
	Let $L(\T(W),p) \subseteq H(\T(R))$ be the set of all levels $i$ during iterations \ref{line:cof:loop_start}-\ref{line:cof:loop_end} of Algorithm \ref{alg:cover_tree_k-nearest_construction} launched with inputs 
	$\T(W),p$. 
	Set $\eta(i) = \min_{t} \{ t \in L(\T(W),p) \mid t > i\}$. 
\end{dfn}
\begin{algorithm}
	\caption{Building a compressed cover tree $\T(R)$ from Definition \ref{dfn:cover_tree_compressed}.
	}
	\label{alg:cover_tree_k-nearest_construction_whole}
	\begin{algorithmic}[1]
		\STATE \textbf{Input} : a finite subset $R$ of a metric space $(X,d)$
		\STATE \textbf{Output} : a compressed cover tree $\T(R)$. 
		\STATE Choose a random point $r \in R$ to be a root of $\T(R)$ 
		\STATE Build the initial compressed cover tree $\T = \T(\{r\})$ by making $l(r) = +\infty$. 
		\FOR{$p \in R \setminus \{r\}$} \alglinelabel{line:con:for:begin}
		\STATE $\T \leftarrow $ run AddPoint$(\T , p )$ described in Algorithm \ref{alg:cover_tree_k-nearest_construction}.
		\ENDFOR  \alglinelabel{line:con:for:end}
		\STATE For root $r$ of $\T$ set $l(r) = 1 +  \max_{p \in R \setminus \{r\}}l(p)$
	\end{algorithmic}
\end{algorithm}
\begin{algorithm}
	\caption{Building $\T(W \cup \{p\})$ in 
		lines \ref{line:con:for:begin}-\ref{line:con:for:end} of Algorithm \ref{alg:cover_tree_k-nearest_construction_whole}.}
	\label{alg:cover_tree_k-nearest_construction}
	\begin{algorithmic}[1]
		\STATE \textbf{Function} AddPoint(a compressed cover tree $\T(W)$ with a root $r$, a point $p\in X$)
		\STATE \textbf{Output} : compressed cover tree $\T(W \cup \{p\})$. 
		\STATE Set $i \leftarrow l_{\max}(\T(W)) - 1$ and $\eta(l_{\max} - 1) = l_{\max}$
		 \COMMENT{If the root $r$ has no children then $ i \leftarrow -\infty$}
		\STATE Set $R_{l_{\max}} \leftarrow \{r\}$ and initialize sorted dictionary $M$ with $M[l_{\max}] = \{r\}$
		\WHILE{$i \geq l_{\min}$} \alglinelabel{line:cof:loop_start}
		\STATE Assign $\mathcal{C}_i(R_{\eta(i)}) \leftarrow  R_{\eta(i)} \cup \{a \in \Child(q) \text{ for some }q \in R_{\eta(i)} \mid l(a) = i \} $
		\alglinelabel{line:cof:dfn_C}
		\STATE Set $R_{i} = \{a \in \C_i(R_{\eta(i)}) \mid d(p,a) \leq 2^{i+1} \}$ and $M[i] = R_{i}$.
		\alglinelabel{line:cof:defRim1}
		\IF {$R_i$ is empty}\alglinelabel{line:cof:inner_loop:begin}
		\STATE  Launch Algorithm \ref{alg:construction_parent_assign} with parameters $(p, M)$ and \textbf{exit this algorithm}.\alglinelabel{line:cof:inner_loop:mid}
		\ENDIF \alglinelabel{line:cof:inner_loop:end}
		\STATE $t = \max_{ a \in R_{i}} \nxt(a,i,\T(W)) $  \alglinelabel{line:cof:dfn_t}
		\COMMENT{If $R_{i}$ has no children we set $t = l_{\min} - 1$}
		\STATE $\eta(i) \leftarrow i$ and $i \leftarrow t$ 
		\ENDWHILE \alglinelabel{line:cof:loop_end}
        \STATE Launch Algorithm \ref{alg:construction_parent_assign} with parameters $(p, M)$. \alglinelabel{line:cof:selectParentEnd}
        
	\end{algorithmic}
\end{algorithm}

\begin{algorithm}
\caption{Assign node subprocedure}
\label{alg:construction_parent_assign}
\begin{algorithmic}[1]
\STATE \textbf{Function} AssignParent(Point $p$,  dictionary $M$)
\STATE \textbf{Output:} Compressed cover tree $\T(W \cup \{p\})$
\STATE Set $i$ to be the lowest key of $M$.
\WHILE{$i \leq l_{max}$}
\IF {$d(p,R_i) \leq 2^{i}$}\alglinelabel{line:assignparent:assgiment:ifstatement}
\STATE Let $q \in R_{i}$ such that $d(q,p) = d(R_{i},p)$, let $x$ maximal integer for which $d(p,q) > 2^{x}$. \alglinelabel{line:assignparent:assgiment:first}
 \STATE Set $l(p) = x$ and $q$ to be the parent of $p$.  \alglinelabel{line:assignparent:assgiment}
\ENDIF
\STATE Find next key $j > i$ of $M$ and set $i = j$
\ENDWHILE
\end{algorithmic}
\end{algorithm}

\noindent
Let $R$ be a finite subset of a metric space $(X,d)$. 
A compressed cover tree $\T(R)$ will be incrementally constructed by adding points one by one as summarized in Algorithm \ref{alg:cover_tree_k-nearest_construction_whole}. 
First we select a root node $r \in R$ and form a tree $\T(\{r\})$ of a single node $r$ at the level $l_{\max} = l_{\min} = +\infty$. 
Assume that we have a compressed cover tree $\T(W)$ for a subset $W \subset R$. 
For any point $p \in R \setminus W$, Algorithm \ref{alg:cover_tree_k-nearest_construction} builds a larger compressed cover tree $\T(W \cup \{p\})$ from $\T(W)$. 

\medskip 
\noindent
Note that during the construction of the compressed cover tree in Algorithm \ref{alg:cover_tree_k-nearest_construction} we write down additional information for every node $p$ , which includes the number of descendants of node $p$ and the maximal level of nodes in set $\Child(p)$. 

\begin{lem}
\label{lem:separation_of_descendants}
Let $\T(R)$ be a cover tree and let $p \in X$ be a point and let $i \in \Z$. Assume that for some $q \in \T(R)$ we have $d(p,q) > 2^{i+1}$. 
Let $\Sd_i(q,\T(R))$ be as defined in Definition~\ref{dfn:distinctive_descendant_set}.
Then for any $\theta \in \Sd_i(q, \T(R)) \setminus \{q\}$ we have $d(\theta,p) > 2^{l(\theta)}$.
\end{lem}
\begin{proof}
\hypertarget{proof:lem:separation_of_descendants}{\empty} Let $S= (\theta = a_0, ...,a_m = q)$ be a node to node path. Since $\theta \in \Sd_i(q, \T(R)) \setminus \{q\}$ by Lemma~\ref{lem:distinctive_descendant_child_level}
we have $l(a_{m-1}) \leq i - 1$. Therefore $l(\theta) = l(a_0) \leq ... \leq l(a_{m-1}) \leq i-1$. We have the following inequality:
 $$d(q,\theta) \leq  \sum\limits^{h-1}_{z = 0}d(a_z, a_{z+1})  \leq \sum\limits^{j}_{x = l(\theta)+1} 2^{x} = (2^{j+1} - 2^{l(\theta)+1}).$$
  By triangle inequality we have: $d(p,\theta) \geq d(p,\gamma) - d(\gamma, \theta) > 2^{j+1} - (2^{j+1} - 2^{l(\theta)+1}) > 2^{l(\theta)}$.
Therefore $d(p,\theta) > 2^{l(\theta)}$ which proves the claim. 
\end{proof}



\thmconstructioncorrectness*


\begin{proof}
\hypertarget{proof:thm:construction_correctness}{\empty}
	It suffices to prove that Algorithm~\ref{alg:cover_tree_k-nearest_construction} correctly extends a compressed cover tree $\T(W)$ for any finite subset $W\subseteq X$ by adding a point $p$. Let us prove that $\T(W \cup \{p\})$ satisfies Definition~\ref{dfn:cover_tree_compressed}. 

    \medskip
    \noindent 

 We first note that the parent $q$ of $p$ is always assigned in Algorithm~\ref{alg:construction_parent_assign} by choosing $q \in R_i$ which minimizes $d(R_i, p)$ as the parent of $p$ and by setting level of $p$ to be maximal integer which satisfies $d(p,q) > 2^{x}$.
Since $d(R_{i}, p) \leq 2^{i}$ we have $l(p) x < i \leq l(q)$. Therefore $l(p) < l(q)$. We also have
 $d(q,p) \leq 2^{x+1} \leq 2^{l(p) + 1}$. Therefore both conditions of (\ref{dfn:cover_tree_compressed}b) are satisfied. 

 \medskip
 \noindent 

 To check (\ref{dfn:cover_tree_compressed}c) Consider arbitrary cover set $C_{h} = \{q \in \T(W \cup \{p\}) \mid l(q) \geq h\}$. Since we have assumed that $\T(W)$ is a valid cover tree, all the cover sets $C_h$ for $h > l(p)$ satisfy the condition.  Let us consider cover sets having $h \leq l(p)$. Consider a sequence of iterations $l_{\min}(\T(W)) \leq s(0) < s(1) < ... < s(t) = l_{\max}(\T(W))$ that were saved as keys of the dictionary $M$ defined in Line~\ref{line:cof:defRim1} of Algorithm~\ref{alg:cover_tree_k-nearest_construction}. Let $\theta \in C_h$ be an arbitrary node. 

\medskip
\noindent 

Assume first that $\theta \notin R_{s(j)}$ for all $j \in \{0,...,s(t)\}$. Therefore $\theta$ is a non-trivial descendant of some node $\gamma$ that was eliminated in line~\ref{line:cof:defRim1} in some level $s(m)$. We have $\theta \in \Sd_{s(m)}(\gamma, \T(R)) \setminus \{\gamma\}$ and by line~\ref{line:cof:defRim1} also $d(p,\gamma) > 2^{s(m)+1}$. By Lemma~\ref{lem:separation_of_descendants} it follows that $d(p,\theta) > 2^{l(\theta)}$. Since $\theta \in C_h$ it follows that $d(p,\theta) > 2^{h}$.

\medskip
\noindent 

Assume then that $\theta \in R_{s(j)}$ for all $j \in [n,m]$. Assume first that $s(n-1) \geq h$. 
Since $\theta \in \C_{s(n-1)}(R_{s(n)}) \setminus R_{s(n-1)}$ we have $d(p,\theta) > 2^{s(n-1)+1} \geq 2^{h}$. 
We then assume that $s(n-1) < h$.
Since $\theta \in C_h$ it follows that $s(m) \geq h$. Pick minimal $u \geq n$ such that $h \leq s(u)$. Assume first that $h \leq s(u) \leq l(p)$,  then by line~\ref{line:assignparent:assgiment:ifstatement} in Algorithm~\ref{alg:construction_parent_assign} we have $d(\theta, p) \geq d(R_{s(u)}, p) > 2^{s(u)} \geq 2^{h}$. In final case assume that $h \leq l(p) < s(u)$. Since $s(u-1) < h$ it follows that the parent of $p$ was selected from $R_{s(u)}$ in Algorithm~\ref{alg:construction_parent_assign}. 
By line~\ref{line:assignparent:assgiment:first} we have $d(R_{s(u)},p) > 2^{l(p)}$. Therefore $d(\theta, p) > 2^{l(p)} \geq 2^{h}$.

\end{proof}

\lemgeneralconstructiontime*
\begin{proof}
    \hypertarget{proof:lem:general_construction_time}{\empty}
	The worst-case time complexity of Algorithm \ref{alg:cover_tree_k-nearest_construction_whole} is dominated by lines \ref{line:con:for:begin}-\ref{line:con:for:end} which call Algorithm \ref{alg:cover_tree_k-nearest_construction} $O(|R|)$ times in total. 
	\medskip
	
	\noindent
	Assume that we have already constructed a cover tree on set $\T(W_y)$, the goal Algorithm \ref{alg:cover_tree_k-nearest_construction} is to construct tree $\T(W_y \cup \{p_{y+1}\})$.
	By Definition \ref{dfn:cover_tree_construction_iteration_set} loop on lines
	\ref{line:cof:loop_start}-\ref{line:cof:loop_end} is performed $L(\T(W_y),p_{y+1})$ times. 
	Let $R_{*}$ be the maximal size of set $R_i$ during all iterations $i \in L(\T(W_y),p_{y+1})$.
	By Lemma~\ref{lem:compressed_cover_tree_width_bound} since $W_{y+1} \subseteq R \subseteq X$ we have 
	$$|\C_{i}(R_{\eta(i)})| \leq c_m(W_{y+1})^4 \cdot |R_{*}| \leq c_m(R)^4 \cdot |R_{*}|$$ nodes, where $\C_{\eta(i)}(R_{\eta(i)})$ is defined in 
 line~\ref{line:cof:dfn_C}.  Therefore both, lines \ref{line:cof:defRim1}  and \ref{line:cof:dfn_C} take at most $c_m(R)^4|R_{*}|$ time. In line \ref{line:cof:dfn_t} we handle $|R_{*}|$ elements, for each of them we can retrieve index $\nxt(a,i, \T(W_y))$ in $O(1)$ time, since for every $a \in \T(R)$ we can update the last index $j$, when $a$ had children on level $j$ in line \ref{line:cof:dfn_C}. Therefore line  \ref{line:cof:dfn_t} takes at most $O(|R_{*}|)$ time.  
    Algorithm~\ref{alg:construction_parent_assign} is called once during whole run-time of the algorithm takes at most $O(L(\T(W_y),p_{y+1}) * |R_{*}|)$ time. 
    Therefore line~\ref{line:cof:inner_loop:mid} and line~\ref{line:cof:selectParentEnd} take at most $O(L(\T(W_y),p_{y+1}) * |R_{*}|)$ time.
	Let us now bound $|R_{*}|$ during the whole run-time of the algorithm.
	
	\medskip
	
	\noindent
	Let $i$ be an arbitrary level. 
	Note that $R_{i} \subseteq B(p,2^{i+1}) \cap C_{i}$ where $C_{i}$ is a $i$th cover set of $\T(R)$. Since $C_{i}$ is $2^{i}$-spares subset of $R$ we can apply packing Lemma \ref{lem:packing} with $r = 2^{i+1}$ and $\delta = 2^{i}$ to obtain 
	$|B(p,2^{i+1}) \cap C_{i} | \leq (c_m(W))^4 $.  
	Lemma \ref{lem:expansion_constant_property} implies $(c_m(W))^4  \leq (c_m(R))^4 $, therefore $|B(p,2^{i}) \cap C_{i} | \leq (c_m(R))^4$. 
	\smallskip
	
	\noindent
	The time complexity of loop \ref{line:cof:loop_start} - \ref{line:cof:loop_end} in Algorithm \ref{alg:cover_tree_k-nearest_construction} is dominated by line \ref{line:cof:dfn_C} that has time $O(|C(R_i)|) \leq O((c_m(R))^4 \cdot |R_{*}|) \leq O((c_m(R))^8)$. 
	Then the whole  Algorithm \ref{alg:cover_tree_k-nearest_construction_whole} has time
	$$O((c_m(R))^8 \cdot  \max\limits_{y=2,...,|R|}L(\T(W_{y-1}),p_{y}) \cdot |R|)$$ as desired. 
\end{proof}
\thmconstructiontime*
\begin{proof}
	In Lemma \ref{lem:general_construction_time} use the upper bounds  due to Lemma \ref{lem:depth_bound} as follows: \\
	$\max\limits_{y \in 2,...,|R|}|L(\T(W_{y-1}),p_{y})| \leq H(\T(R))\leq 1 + \log_2(\Delta(R))$.
\end{proof}

\lemknnnextlevelfinder*
\begin{proof}
	\hypertarget{proof:lem:knn_next_level_finder_for_log_depth}{\empty}
	Note first that since $\eta(i+3) \in L(\T(R),q)$, there exists distinct 
	$u \in R_{\eta(\eta(i+3))}$ and $v \in \C_{\eta(i+1)}(R_{\eta(\eta(i+1)}))$, in such a way that $u$ is the parent of $v$. 
	Let us show that both of $u,v$ cant belong to set $R_i$. Assume contrary that both $u,v \in R_i$. Then by line  \ref{line:cof:defRim1} of Algorithm \ref{alg:cover_tree_k-nearest_construction}
	we have $d(v,q) \leq 2^{i+1}$ and $d(u,q) \leq 2^{i+1}$. By triangle inequality $d(u,v) \leq d(u,q) + d(q,v) \leq 2^{i+2} \leq 2^{\eta(i+1)}$.
	Recall that we denote a level of a node by $l$.
	On the other hand we have $l(u) \geq \eta(i+1)$ and $l(v) \geq \eta(i+1)$, by separation condition of Definition \ref{dfn:cover_tree_compressed} we have $d(u,v) > 2^{\eta(i+1)}$, which is a contradiction. Therefore only one of $\{u,v\}$ 
	can belong to $R_i$. It sufficies two consider the two cases below: 
	
	\medskip
	
	\noindent
	\textbf{Assume that }$v \notin R_i$. Since $v$ is children of $u$ we have $d(u,v) \leq 2^{\eta(i+1) + 1}$.
	By line \ref{line:cof:defRim1} of Algorithm \ref{alg:cover_tree_k-nearest_construction} we have $d(u,q) \leq 2^{\eta(i+1) + 1}$.
	By triangle inequality 
	$$d(v,q) \leq d(v,u) + d(u,q) \leq 2^{ \eta(i+1) + 1 } + 2^{\eta(i+1) + 1} \leq 2^{\eta(i+1) + 2} \leq 2^{\eta(\eta(i+1)) + 1} $$
	Since $v \notin R_i$ there exists level $t$ having $\eta(i+1) \geq t \geq i$ and $v \in \C_{t}(R_{\eta(t)}) \setminus R_t$.
	Therefore by line \ref{line:cof:defRim1} of Algorithm \ref{alg:cover_tree_k-nearest_construction}  we have $d(q,v) > 2^{t+1} \geq 2^{i+1}$.
	It follows that we have found point $v \in R$ satisfying $2^{i+1} < v \leq 2^{\eta(\eta(i+1)) + 1}$. Therefore $p = v$, is the desired point.
	
	\medskip
	
	\noindent
	\textbf{Assume that }$u \notin R_i$. Since $u \in R_{\eta(\eta(i+1))}$, by line \ref{line:cof:defRim1} of Algorithm \ref{alg:cover_tree_k-nearest_construction} we have $d(u,q) \leq 2^{\eta(\eta(i+1)) + 1}$.
	On the other hand since $u \notin R_i$, there exists level $t$ having $\eta(i+3) \geq t \geq i$ and $u \in \C_{t}(R_{\eta(t)}) \setminus R_t$. Therefore by line \ref{line:cof:defRim1} of Algorithm \ref{alg:cover_tree_k-nearest_construction} we have $d(q,u) > 2^{t+2} \geq 2^{i+2}$.
	It follows that we have found point $u \in R$ satisfying $2^{i+1} < u \leq 2^{\eta(\eta(i+1)) + 1}$. Therefore $p = u$, is the desired point.
	
\end{proof}


\lemconstructiondepthbound*
\begin{proof}
    \hypertarget{proof:lem:construction_depth_bound}{\empty}
	Let $x \in L(\T(R),q)$ be the lowest level of $L(\T(R),q)$.
	Define $s_1 = \eta(\eta(x)+1)$ and let $s_i = \eta(\eta(\eta(s_{i-1}+1))+1)$, if it exists. Assume that $s_{n+1}$ is the last sequence element for which $\eta(\eta(\eta(s_{n-1}+1))+1)$ is defined. Define $S = \{s_1,...,s_{n}\}$. For every $i \in \{1,...,n\}$ let $p_i$ be the point provided by Lemma \ref{lem:knn_next_level_finder_for_log_depth} that satisfies $$ 2^{s_i+1} < d(p_i,q) \leq 2^{\eta(\eta(s_{i}+1)) + 1}.$$
	Let $P$ be the sequence of points $p_i$.  Denote $n = |P| = |S|$. 
	Let us show that $S$ satisfies the conditions of Lemma \ref{lem:growth_bound_extension}. Note that:
	$$4 \cdot d(p_i,q)\leq 4 \cdot 2^{\eta(\eta(s_{i}+1)) + 1} \leq 2^{\eta(\eta(s_{i}+1)) + 3} \leq 2^{\eta(\eta(\eta(s_{i}+1))+1) + 1} \leq 2^{s_{i+1}+1} \leq d(p_{i+1},q)$$
	By Lemma \ref{lem:growth_bound_extension} applied for set $A$ and sequence $P$ we get:
	$$|\bar{B}(q,\frac{4}{3}  d(q,p_n))| \geq (1+\frac{1}{c(R)^2})^{n} \cdot |\bar{B}(q,\frac{1}{3} d(q,p_1))|$$
	Since $\eta(x) \in L(\T(R),q)$ , there exists some point $u \in R_{\eta(x)}$. 
	By definition of $R_i$ we have $d(u,q) \leq 2^{\eta(x) + 1}$. 
	Also $$2^{\eta(\eta(x) + 1)-1} \leq \frac{2^{\eta(\eta(x) + 1)+1}}{3} < \frac{d(q,p_1)}{3}$$
	It follows that:
	$$1 \leq  |\bar{B}(q, 2^{\eta(x) + 1})| \leq |\bar{B}(q, 2^{\eta(\eta(x)+1) - 1}| \leq |\bar{B}(q, \frac{d(q,p_1)}{3})|$$
	Therefore we have
	$$|A| \geq \frac{|\bar{B}(q,\frac{4}{3} \cdot d(q,p_n))|}{|\bar{B}(q,\frac{1}{3} \cdot d(q,p_1))|} \geq (1+\frac{1}{c(A)^2})^{n}$$
	Note that $c(A) \geq 2$ by definition of expansion constant. Then by applying $\log$ and by using Lemma \ref{lem:hard_function_bound} we obtain: $c(A)^2\log(A) \geq n = |S|$. 
	Let $x$ be minimal level of $L(\T(W),q)$ and let $y$ be the maximal level of $L(\T(W),q)$ 
	Note that $S$ is a sub sequence of $L$ in such a way that:
	\begin{itemize}
		\item $[x,s_1] \cap L(\T(R),q) \leq 3$, 
		\item for all $i \in 1,..., n$ we have $[s_i, s_{i+1}] \cap L(\T(R),q) \leq 6 $
		\item $[s_n, y] \cap L(\T(R),q) < 12$
	\end{itemize}
	Since segments $[x,s_1],[s_1,s_2], ...,  [s_2,s_n], [s_n,y]$ cover $|L(\T(R),q)|$,
	it follows that $|S| \geq \frac{|L(\T(R),q)|}{12}$. We obtain that $$|L(\T(R),q)| \leq 12 \cdot c(A)^2 \cdot \log_2(|A|),$$ which proves the claim.
\end{proof}

\thmconstructiontimeKR*
\begin{proof}
	It follows from Lemmas~\ref{lem:construction_depth_bound} and~\ref{lem:general_construction_time}.
\end{proof}


\corconstructiontimeKR*
\begin{proof}
	The proof follows from Theorem~\ref{thm:construction_time_KR} by setting $A = R$.
\end{proof}

\section{$k$-nearest neighbor search algorithm}
\label{sec:better_approach_knn_problem}


This section is motivated by
\citet[Counterexample~5.2]{elkin2022counterexamples}, which showed that the proof of past time complexity claim in  \citet[Theorem~5]{beygelzimer2006cover} for the nearest neighbor search algorithm contained gaps.  The two main results of this sections are Corollary~\ref{cor:cover_tree_knn_miniziminzed_constant_time}
and Theorem \ref{thm:knn_KR_time} which provide new time complexity results for $k$-nearest neighbor problem, assuming that a compressed cover tree was already constructed for the reference set $R$. For the construction algorithm of compressed cover tree and its time complexity, we refer to Section~\ref{sec:ConstructionCovertree}.

 \medskip
 \noindent
 The past mistakes are resolved by introducing a new Algorithm \ref{alg:cover_tree_k-nearest} for finding $k$-nearest neighbors that generalize and improves the original method for finding nearest neighbors using an implicit cover. \citet[Algorithm~1]{beygelzimer2006cover}. The first improvement is $\lambda$-point of line~\ref{line:knnu:dfnLambda} which allows us to search for all $k$-nearest neighbors of a given query point for any $k \geq 1$. The second improvement is a new loop break condition on line \ref{line:knnu:qtoofar:condition}. The new loop break condition is utilized in the proof of Lemma \ref{lem:knn_depth_bound} to conclude that the total number of performed iterations is bounded by  $O(c(R)^2\log(|R|))$  during whole run-time of the algorithm. The latter improvement closes the past gap in proof of \citet[Theorem~5]{beygelzimer2006cover} by bounding the number of iterations independently from the explicit depth \citet[Definition~3.2]{elkin2022counterexamples}, that generated the past confusion. 
 
  \medskip
 \noindent 
 Recall from Definition~\ref{dfn:essential_levels_node} that an essential set $\Es(p,\T(R)) \subseteq H(\T(R)$ consists of all levels $i \in H(\T(R))$ for which $p$ has non-trivial children in $\T(R)$ at level $i$. 
 By Lemma~\ref{lem:distinctive_descendants_precompute} the sizes of distinctive descendants $|\Sd_i(p, \T(R))|$ can be precomputed in a linear time $O(|R|)$ for all $p \in R$ and $i \in \Es(p,\T(R))$.
 Since the size of distinctive descendant set $|\Sd_i(p, \T(R))|$ can only change at indices $i \in \Es(p,\T(R))$, we assume that the sizes of $|\Sd_i(p, \T(R))|$ can be retrieved in a constant time $O(1)$ for any $p \in R$ and $i \in H(\T(R))$ during the run-time of Algorithm~\ref{alg:cover_tree_k-nearest}.



\begin{dfn}
	\label{dfn:knn_iteration_set}
	Let $R$ be a finite subset of a metric space $(X,d)$. 
	Let $\T(R)$ be a cover tree of Definition \ref{dfn:cover_tree_compressed} built on $R$ and let $q \in X$ be arbitrary point.
	Let $L(\T(R),q) \subseteq H(\T(R))$ be the set of all levels $i$ during iterations of lines~\ref{line:knnu:loop_begin}-\ref{line:knnu:loop_end} of Algorithm~\ref{alg:cover_tree_k-nearest} launched with inputs 
	$\T(R),q$. 
	If Algorithm~\ref{alg:cover_tree_k-nearest} reaches line \ref{line:knnu:qtoofar} at level 
	$\varrho \in L(\T(R),q)$, then we say that is \emph{special}. 
	We denote $\eta(i) = \min_{t} \{ t \in L(\T(R),q) \mid t > i\}$. 
	\bs
\end{dfn}

\noindent
Note that $\eta(i)$ of Definition \ref{dfn:knn_iteration_set} may be undefined. If $\eta(i)$ is defined, then by definition we have $\eta(i) \geq i + 1$.
Let $d_k(q,R)$ be the distance of $q$ to its $k$th nearest neighbor in $R$. 

\begin{algorithm}
	\caption{$k$-nearest neighbor search by a compressed cover tree}
	\label{alg:cover_tree_k-nearest}
	\begin{algorithmic}[1]
		\STATE \textbf{Input} : compressed cover tree $\T(R)$, a query point $q\in X$, an integer $ k \in \Z_{+} $
		\STATE Set $i \leftarrow l_{\max}(\T(R)) - 1$ and $\eta(l_{\max}-1) = l_{\max}$
		\STATE  Let $r$ be the root node of $\T(R)$. Set $R_{l_{\max}}=\{r\}$.
		\WHILE{$i \geq l_{\min}$} \alglinelabel{line:knnu:loop_begin}
		\STATE Assign $\mathcal{C}_i(R_{\eta(i)}) \leftarrow  R_{\eta(i)} \cup \{a \in \Child(p) \text{ for some }p \in R_{\eta(i)} \mid l(a) = i \} $ \\ \COMMENT{Recall that $\Child(p)$ contains node $p$ } \alglinelabel{line:knnu:dfn_C}
		\STATE Compute $\lambda = \lambda_k(q,\C_{i}(R_{\eta(i)}))$ from Definition~\ref{dfn:lambda-point} \alglinelabel{line:knnu:dfnLambda} by 
		Algorithm \ref{alg:lambda}.
		\STATE Find $R_{i} = \{p \in \C_i(R_{\eta(i)}) \mid d(q,p) \leq d(q,\lambda) + 2^{i+2}\}$ \alglinelabel{line:knnu:dfnRi}
		\IF {$d(q,\lambda) > 2^{i+2}$} \alglinelabel{line:knnu:qtoofar:condition}
		\STATE Define list $S = \emptyset$
		\FOR{$p \in R_i$} \alglinelabel{line:knnu:qtoofar:loop:start}
		\STATE Update $S$ by running Algorithm \ref{alg:cover_tree_k-nearest_final_collection} on $(p,i)$ \alglinelabel{line:knnu:qtoofar:launch} 
		\ENDFOR  \alglinelabel{line:knnu:qtoofar:loop:end}
		\STATE Compute and \textbf{output} $k$-nearest neighbors of the query point $q$ from set $S$.
		\alglinelabel{line:knnu:qtoofar}
		\ENDIF \alglinelabel{line:knnu:qtoofar:condition:endif}
		\STATE Set $j \leftarrow \max_{ a \in R_{i}} \nxt(a,i,\T(R))$
		\COMMENT{If such $j$ is undefined, we set $j = l_{\min}-1$} \alglinelabel{line:knnu:dfnindexj}
		\STATE Set $\eta(j) \leftarrow i$ and $i \leftarrow j$.
		\ENDWHILE \alglinelabel{line:knnu:loop_end}
		\STATE Compute and \textbf{output} $k$-nearest neighbors of query point $q$ from the set $R_{l_{\min}}$.
		\alglinelabel{line:knnu:final_line}
	\end{algorithmic}
\end{algorithm}

\begin{algorithm}
	\caption{The node collector called in line~\ref{line:knnu:qtoofar:launch} of Algorithm~\ref{alg:cover_tree_k-nearest}.}
	\label{alg:cover_tree_k-nearest_final_collection} 
	\begin{algorithmic}[1]
		\STATE \textbf{Input: }$p \in R$, index $i$.
		\STATE \textbf{Output: } a list $S \subseteq R$ containing all nodes of $\Sd_i(p,\T(R))$.
		\STATE Add $p$ to list $S$. 
		\IF {$i > l_{\min}(\T(R))$}
		\STATE Set $j = \nxt(p,i,\T(R))$ 
		\STATE Set $C =\{a \in \Child(p) \mid l(a) = j \}$
		\FOR{$u \in C$}
		\STATE Call Algorithm \ref{alg:cover_tree_k-nearest_final_collection} with $(u,j)$.
		\ENDFOR
		\ENDIF
	\end{algorithmic}
\end{algorithm}

\begin{exa}[Simulated run of Algorithm \ref{alg:cover_tree_k-nearest}]
	\label{exa:simulatedRun}
	Let $R$ and $\T(R)$ be as in Example \ref{exa:cover_tree_big}. Let $q = 0$ and $k = 5$. Figures \ref{fig:iteration3goodexample}, \ref{fig:iteration2goodexample},  \ref{fig:iteration1goodexample} and \ref{fig:iteration0goodexample} illustrate simulated run of Algorithm \ref{alg:cover_tree_k-nearest} on input $(\T(R), q, k)$. Recall that $l_{\max} = 2$ and $l_{\min} = -1$. During the iteration $i$ of Algorithm \ref{alg:cover_tree_k-nearest} we maintain the following coloring: Points in $R_i$ are colored orange. Points $\C_{\eta(i)}(R_{\eta(i)})$ (of line 5) that are not contained in $R_i$ are colored yellow. The $\lambda$-point of line \ref{line:knnu:dfnLambda} is denoted by using purple color. All the nodes that were present in $R_{\eta(i)}$ , but are no longer included in $R_i$ will be colored red. Finally all the points that are selected as $k$-nearest neighbors of $q$ are colored green in the final iteration. Nodes that haven't been yet visited or that will never be visited are colored white. Let $R_{2} = \{8\}$. Consider the following steps:

	\smallskip
	
	\noindent
	\textbf{Iteration} $i = 1$:  Figure \ref{fig:iteration3goodexample} illustrates iteration $i = 1$ of the Algorithm \ref{alg:cover_tree_k-nearest}. In line \ref{line:knnu:dfn_C} we find
	$\C_1(R_2) = \{4,8,12\}$. Since node $4$ minimizes distance $d(\C_1(R_2),0)$ and distinctive descendant set $\Sd_2(4, \T(R))$ consists of 7 elements we get $\lambda = 4$ and therefore $d(q,\lambda) = 4 \leq 2^{i+2} = 8$.
	In line \ref{line:knnu:dfnRi} we find $R_{1} = \{r \in C \mid d(0,r) \leq d(q,\lambda) + 2^{3} = 12\} = \{4,8,12\}$.
	
	\smallskip
	
	\noindent
	\textbf{Iteration} $i = 0$:   Figure \ref{fig:iteration2goodexample} illustrates iteration $i = 0$ of the Algorithm \ref{alg:cover_tree_k-nearest}. In line \ref{line:knnu:dfn_C}  we find 
	$\C_{0}(R_1) = \{2,4,6,8,10,12,14\}.$ Since $|\Sd_1(2, \T(R))| = 3$, $|\Sd_1(4, \T(R))| = 1$ and $|\T_1(6)| = 3$ and $6$ is the node with smallest to distance $0$ satisfying $\sum_{p \in N(0, 6) = \{2,4,6\}} | \Sd_1(p, \T(R))| \geq 5 = k.$ It follows that $\lambda = 6$.  In line \ref{line:knnu:dfnRi} we find $R_{0} = \{r \in \C(R_1) \mid d(0,r) \leq d(q,\lambda) + 2^{2} = 10\} = \{2,4,6,8,10\}$.
	Since $d(q,\lambda) > 2^{i+2} = 4$. We proceed into lines \ref{line:knnu:qtoofar:condition} - \ref{line:knnu:qtoofar:condition:endif}
	
	\smallskip
	
	\noindent
	\textbf{Final block} lines \ref{line:knnu:qtoofar:condition} - \ref{line:knnu:qtoofar:condition:endif} for  
	$i = 0$:  Figure \ref{fig:iteration1goodexample} marks all the 
	points $S$ discovered by line \ref{line:knnu:qtoofar:launch} as orange. Figure \ref{fig:iteration0goodexample} illustrates the final selection of $k$ points from set $S$ that are selected as the final output $\{1,2,3,4,5\}$.


\end{exa}

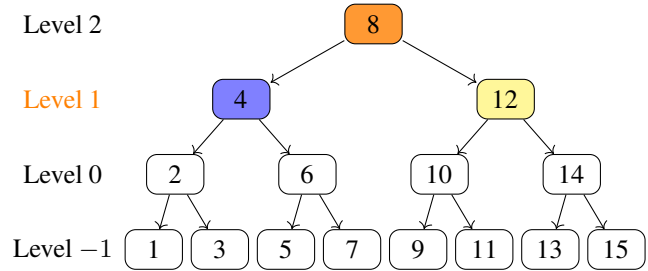
\begin{figure}[H]
	\centering
	\begin{tikzpicture}[align=center, node distance = 1.0cm, scale = 0.45]

	\node (scale3) {Level 2};
	\node[below of =scale3] (scale2) {\color{orange} Level 1};
	\node[below of =scale2] (scale1) {Level 0};
	\node[below of =scale1] (scale0) {Level $-1$};
	\node [blockzm1,  right=1pt of scale0 ] (node1) {1};
	\node [blockzm1,  right=26pt of scale0 ] (node3) {3};
	\node [blockzm1,  right=51pt of scale0 ] (node5) {5};
	\node [blockzm1,  right=76pt of scale0 ] (node7) {7};
	\node [blockzm1,  right=101pt of scale0 ] (node9) {9};
	\node [blockzm1,  right=126pt of scale0 ] (node11) {11};
	\node [blockzm1,  right=151pt of scale0 ] (node13) {13};
	\node [blockzm1,  right=176pt of scale0 ] (node15) {15};

	\node [blockzm1,  right=13pt of scale1 ] (node2) {2};
	\node [blockzm1,  right=63pt of scale1 ] (node6) {6};
	\node [blockzm1,  right=113pt of scale1 ] (node10) {10};
	\node [blockzm1,  right=163pt of scale1 ] (node14) {14};

	\node [blockzm1p,  right=38pt of scale2 ] (node4) {4};
	\node [blockzm1y,  right=138pt of scale2 ] (node12) {12};
	\node [blockzm1o,  right=88pt of scale3 ] (node8) {8};

            \draw[->] (node2) -> (node1);
    \draw[->] (node2) -> (node3);
     \draw[->] (node6) -> (node5);
    \draw[->] (node6) -> (node7);
     \draw[->] (node10) -> (node9);
    \draw[->] (node10) -> (node11);
     \draw[->] (node14) -> (node13);
    \draw[->] (node14) -> (node15);

    \draw[->] (node8) -> (node4);
    \draw[->] (node8) -> (node12);

      \draw[->] (node4) -> (node2);
     \draw[->] (node4) -> (node6);
         
         \draw[->] (node12) -> (node10);
     \draw[->] (node12) -> (node14);


\end{tikzpicture}
	\caption{Iteration $i = 1$ of simulation in Example \ref{exa:simulatedRun} of Algorithm \ref{alg:cover_tree_k-nearest} }
	\label{fig:iteration3goodexample}
\end{figure}
\begin{figure}[H]
	\centering
	\begin{tikzpicture}[align=center, node distance = 1.0cm, scale = 0.45]

	\node (scale3) {Level 2};
	\node[below of =scale3] (scale2) {Level 1};
	\node[below of =scale2] (scale1) {\color{orange} Level 0};
	\node[below of =scale1] (scale0) {Level $-1$};
	\node [blockzm1,  right=1pt of scale0 ] (node1) {1};
	\node [blockzm1,  right=26pt of scale0 ] (node3) {3};
	\node [blockzm1,  right=51pt of scale0 ] (node5) {5};
	\node [blockzm1,  right=76pt of scale0 ] (node7) {7};
	\node [blockzm1,  right=101pt of scale0 ] (node9) {9};
	\node [blockzm1,  right=126pt of scale0 ] (node11) {11};
	\node [blockzm1,  right=151pt of scale0 ] (node13) {13};
	\node [blockzm1,  right=176pt of scale0 ] (node15) {15};

	\node [blockzm1y,  right=13pt of scale1 ] (node2) {2};
	\node [blockzm1p,  right=63pt of scale1 ] (node6) {6};
	\node [blockzm1y,  right=113pt of scale1 ] (node10) {10};
	\node [blockzm1y,  right=163pt of scale1 ] (node14) {14};

	\node [blockzm1o,  right=38pt of scale2 ] (node4) {4};
	\node [blockzm1o,  right=138pt of scale2 ] (node12) {12};
	\node [blockzm1o,  right=88pt of scale3 ] (node8) {8};

    \draw[->] (node2) -> (node1);
    \draw[->] (node2) -> (node3);
     \draw[->] (node6) -> (node5);
    \draw[->] (node6) -> (node7);
     \draw[->] (node10) -> (node9);
    \draw[->] (node10) -> (node11);
     \draw[->] (node14) -> (node13);
    \draw[->] (node14) -> (node15);

    \draw[->] (node8) -> (node4);
    \draw[->] (node8) -> (node12);

      \draw[->] (node4) -> (node2);
     \draw[->] (node4) -> (node6);
         
         \draw[->] (node12) -> (node10);
     \draw[->] (node12) -> (node14);


\end{tikzpicture}
	\caption{Iteration $i = 0$ of simulation in Example \ref{exa:simulatedRun} of Algorithm \ref{alg:cover_tree_k-nearest} }
	\label{fig:iteration2goodexample}
\end{figure}
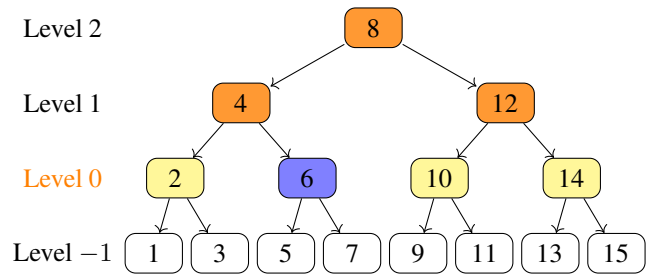

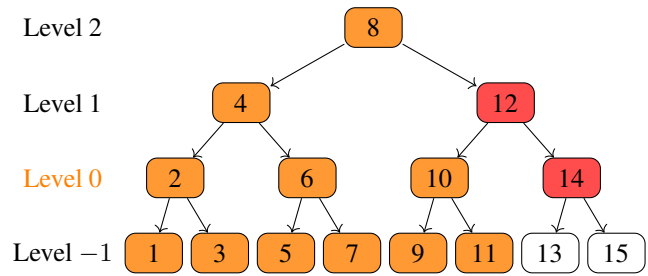
\begin{figure}[H]
	\centering
	\begin{tikzpicture}[align=center, node distance = 1.0cm, scale = 0.45]

	\node (scale3) {Level 2};
	\node[below of =scale3] (scale2) {Level 1};
	\node[below of =scale2] (scale1) {\color{orange} Level 0};
	\node[below of =scale1] (scale0) {Level $-1$};
	\node [blockzm1o,  right=1pt of scale0 ] (node1) {1};
	\node [blockzm1o,  right=26pt of scale0 ] (node3) {3};
	\node [blockzm1o,  right=51pt of scale0 ] (node5) {5};
	\node [blockzm1o,  right=76pt of scale0 ] (node7) {7};
	\node [blockzm1o,  right=101pt of scale0 ] (node9) {9};
	\node [blockzm1o,  right=126pt of scale0 ] (node11) {11};
	\node [blockzm1,  right=151pt of scale0 ] (node13) {13};
	\node [blockzm1,  right=176pt of scale0 ] (node15) {15};

	\node [blockzm1o,  right=13pt of scale1 ] (node2) {2};
	\node [blockzm1o,  right=63pt of scale1 ] (node6) {6};
	\node [blockzm1o,  right=113pt of scale1 ] (node10) {10};
	\node [blockzm1r,  right=163pt of scale1 ] (node14) {14};

	\node [blockzm1o,  right=38pt of scale2 ] (node4) {4};
	\node [blockzm1r,  right=138pt of scale2 ] (node12) {12};
	\node [blockzm1o,  right=88pt of scale3 ] (node8) {8};

    \draw[->] (node2) -> (node1);
    \draw[->] (node2) -> (node3);
     \draw[->] (node6) -> (node5);
    \draw[->] (node6) -> (node7);
     \draw[->] (node10) -> (node9);
    \draw[->] (node10) -> (node11);
     \draw[->] (node14) -> (node13);
    \draw[->] (node14) -> (node15);

    \draw[->] (node8) -> (node4);
    \draw[->] (node8) -> (node12);

      \draw[->] (node4) -> (node2);
     \draw[->] (node4) -> (node6);
         
         \draw[->] (node12) -> (node10);
     \draw[->] (node12) -> (node14);

\end{tikzpicture}
	\caption{Line \ref{line:knnu:qtoofar:launch} of Iteration $i = 0$ of simulation in Example \ref{exa:simulatedRun} of Algorithm \ref{alg:cover_tree_k-nearest} }
	\label{fig:iteration1goodexample}
\end{figure}

\begin{figure}[H]
	\centering
	\begin{tikzpicture}[align=center, node distance = 1.0cm, scale = 0.45]

	\node (scale3) {Level 2};
	\node[below of =scale3] (scale2) {Level 1};
	\node[below of =scale2] (scale1) {\color{orange} Level 0};
	\node[below of =scale1] (scale0) {Level $-1$};
	\node [blockzm1g,  right=1pt of scale0 ] (node1) {1};
	\node [blockzm1g,  right=26pt of scale0 ] (node3) {3};
	\node [blockzm1g,  right=51pt of scale0 ] (node5) {5};
	\node [blockzm1r,  right=76pt of scale0 ] (node7) {7};
	\node [blockzm1r,  right=101pt of scale0 ] (node9) {9};
	\node [blockzm1r,  right=126pt of scale0 ] (node11) {11};
	\node [blockzm1,  right=151pt of scale0 ] (node13) {13};
	\node [blockzm1,  right=176pt of scale0 ] (node15) {15};

	\node [blockzm1g,  right=13pt of scale1 ] (node2) {2};
	\node [blockzm1r,  right=63pt of scale1 ] (node6) {6};
	\node [blockzm1r,  right=113pt of scale1 ] (node10) {10};
	\node [blockzm1r,  right=163pt of scale1 ] (node14) {14};

	\node [blockzm1g,  right=38pt of scale2 ] (node4) {4};
	\node [blockzm1r,  right=138pt of scale2 ] (node12) {12};
	\node [blockzm1r,  right=88pt of scale3 ] (node8) {8};

    \draw[->] (node2) -> (node1);
    \draw[->] (node2) -> (node3);
     \draw[->] (node6) -> (node5);
    \draw[->] (node6) -> (node7);
     \draw[->] (node10) -> (node9);
    \draw[->] (node10) -> (node11);
     \draw[->] (node14) -> (node13);
    \draw[->] (node14) -> (node15);

    \draw[->] (node8) -> (node4);
    \draw[->] (node8) -> (node12);

      \draw[->] (node4) -> (node2);
     \draw[->] (node4) -> (node6);
         
         \draw[->] (node12) -> (node10);
     \draw[->] (node12) -> (node14);


\end{tikzpicture}
	\caption{Line \ref{line:knnu:qtoofar} of iteration $i = 0$ of simulation in Example \ref{exa:simulatedRun} of Algorithm \ref{alg:cover_tree_k-nearest} }
	\label{fig:iteration0goodexample}
\end{figure}
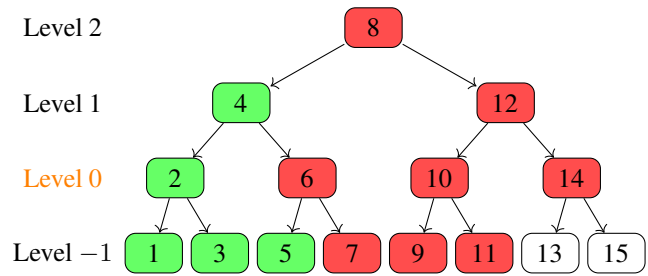

Note that $\bigcup_{p \in R_i}\Sd_i(p, \T(R))$ is decreasing set for which $\bigcup_{p \in R_{l_{\max}}}\Sd_{l_{\max}}(p, \T(R)) = R$ and $$\bigcup_{p \in R_{l_{\min}}}\Sd_{l_{\min}}(p, \T(R)) = R_{l_{\min}}.$$

\begin{lem}[$k$-nearest neighbors in the candidate set for all $i$]
	\label{lem:cover_tree_knn_correct_lem}
	Let $R$ be a finite subset of an ambient metric space $(X,d)$, let $q \in X$ be a query point and let $k \in \Z \cap [1,\infty)$ be a parameter. Let $\T(R)$ be a compressed cover tree of $R$. Assume that $|R| \geq k$. Then for any iteration $i \in L(\T(R),q)$ of Definition \ref{dfn:knn_iteration_set} the candidate set $\bigcup_{p \in R_i}\Sd_i(p, \T(R))$ contains all $k$-nearest neighbors of $q$. \bs
	
\end{lem}
\begin{proof}
	
	Since $R_{l_{\max}} = \{r\}$, where $r$ is the root $\T(R)$ we have $S_{l_{\max}}(r,\T(R)) = R$ and therefore any point among $k$-nearest neighbor of $q$ is contained in $R_{l_{\max}}$. Let $i$ be the largest index for which there exists a point among $k$-nearest neighbor of $q$ that doesn't belong to $\bigcup_{p \in R_{i}}\Sd_i(p, \T(R))$. Let us denote such point by $\beta$, then:
	$$\beta \in \bigcup_{p \in R_{\eta(i)}}\Sd_{\eta(i)}(p, \T(R)) \setminus \bigcup_{p \in R_{i}}\Sd_{i}(p, \T(R)).$$ 
	By Lemma \ref{lem:child_set_equivalence} we have 
	\begin{ceqn}
		
		\begin{equation}
			\label{eqa:neighborsContained}
			\bigcup_{p \in \C_{\eta(i)}(R_{\eta(i)})}\Sd_{i}(p, \T(R)) = \bigcup_{p \in R_{\eta(i)}}\Sd_{\eta(i)}(p, \T(R))
		\end{equation}
		
	\end{ceqn}
	\noindent 
	Let $\lambda$ be as in line \ref{line:knnu:dfnLambda} of Algorithm \ref{alg:cover_tree_k-nearest}. By Equation (\ref{eqa:neighborsContained}) we have $$|\bigcup_{p \in \C_{\eta(i)}(R_{\eta(i)})}\Sd_{i}(p, \T(R))| \geq k,$$ therefore by Definition \ref{dfn:lambda-point} such $\lambda$ exists. Since $\beta \in \bigcup_{p \in \C_{\eta(i)}(R_{\eta(i)})}\Sd_{i}(p, \T(R))$, there exists $\alpha \in \C_{\eta(i)}(R_{\eta(i)})$ satisfying $\beta \in \Sd_{i}(\alpha, \T(R))$. By assumption it follows $\alpha \notin R_{i}$. By line \ref{line:knnu:dfnRi} of the algorithm we have
	\begin{ceqn}
		\begin{equation}
			\label{eqa:neighborsContained2}
			d(\alpha, q) > d(q, \lambda) + 2^{i+2}.
		\end{equation}
	\end{ceqn}
	Let $w$ be arbitrary point in set $\bigcup_{p \in N(q;\la)}\Sd_{i}(p, \T(R))$. Therefore $w \in \Sd_{i}(\gamma, \T(R))$ for some $\gamma \in  N(q;\la)$. By Lemma \ref{lem:distinctive_descendant_distance} applied on $i$ we have  $d(\gamma, w) \leq 2^{i+1}$. By Definition \ref{dfn:lambda-point}  since $\gamma \in  N(q;\la)$ we have $d(q,\gamma) \leq d(q,\lambda)$. By (\ref{eqa:neighborsContained2}) and the triangle inequality we obtain:
	\begin{ceqn}
		\begin{equation}
			\label{eqa:neighborsContained3}
			d(q,w) \leq d(q, \gamma) + d(\gamma,w) \leq d(q,\lambda) + 2^{i+1} < d(\alpha,q) - 2^{i+1} 
		\end{equation}
	\end{ceqn}
	On the other hand $\beta$ is a descendant of $\alpha$ thus we can estimate:
	\begin{ceqn}
		\begin{equation}
			\label{eqa:neighborsContained4}
			d(q,\beta) \geq d(q,\alpha) - d(\alpha,\beta) \geq d(\alpha,q) - 2^{i+1} 
		\end{equation}
	\end{ceqn}
	By combining Inequality (\ref{eqa:neighborsContained3}) with Inequality (\ref{eqa:neighborsContained4}) we obtain $d(q,w) < d(q,\beta)$. Since $w$ was arbitrary point from $\bigcup_{p \in N(q;\la)}\Sd_{i}(p, \T(R))$, that contains at least  $k$ points, $\beta$ cannot be any $k$-nearest neighbor of $q$, which is a contradiction. 
\end{proof}

\thmknncorrectness*
\begin{proof}
    \hypertarget{proof:thm:cover_tree_knn_correct}{\empty}	
	Note that Algorithm~\ref{alg:cover_tree_k-nearest} is terminated by either reaching line \ref{line:knnu:final_line} or by going inside 
	block \ref{line:knnu:qtoofar:loop:start} - \ref{line:knnu:qtoofar:loop:end}.
	
	\medskip
	\noindent
	Assume first that Algorithm~\ref{alg:cover_tree_k-nearest} is terminated by reaching line \ref{line:knnu:final_line}.
	Claim follows directly from Lemma \ref{lem:cover_tree_knn_correct_lem} by noting that since 
	$i = l_{\min}$ all the nodes $p \in R_{l_{\min}}$ do not have any children. Therefore it follows $\bigcup_{p \in R_{l_{\min}}}\Sd_i(p, \T(R)) = R_{l_{\min}}$. Thus all the $k$-nearest neighbors of $q$ are contained in the set $R_{l_{\min}}$.
	
	\medskip
	\noindent
	Assume then that block \ref{line:knnu:qtoofar:loop:start} - \ref{line:knnu:qtoofar:loop:end} is reached during some iteration 
	$i \in L(\T(R),q)$. By Lemma \ref{lem:cover_tree_knn_correct_lem} set $\bigcup_{p \in R_i}\Sd_i(p, \T(R))$ contains all $k$-nearest neighbors of $q$. Note that in line \ref{line:knnu:qtoofar:launch} we collect all nodes of $\bigcup_{p \in R_i}\Sd_i(p, \T(R))$ into single array $S$. Therefore in line \ref{line:knnu:qtoofar} we correctly select $k$ nearest neighbors of $q$ from array $S$, which proves the claim. 
\end{proof}

	
	
\lemknntime*
\begin{proof}
    \hypertarget{proof:lem:knn:time}{\empty}
	\textbf{(a)}
	Let $\varrho \in L(\T(R),q)$ be as in Definition \ref{dfn:knn_iteration_set}. 
	Note that if iteration $\varrho$ is encountered, it becomes the last iteration of $L(\T(R),q)$.
	The total number of children encountered in line \ref{line:knnu:dfn_C} during single iteration (\ref{line:knnu:loop_begin}-\ref{line:knnu:loop_end}) is at most is at most 
	$(c_m(R))^4 \cdot \max\limits_{i \in L(\T(R),q) \setminus \varrho}|R_i|$
	by Lemma \ref{lem:compressed_cover_tree_width_bound}. From Lemma \ref{lem:time_lambdapoint} we obtain that line \ref{line:knnu:dfnLambda}, which launches Algorithm \ref{alg:lambda} takes at most
	$$|\C(R_i)| \cdot \log_2(k) = (c_m(R))^4 \cdot \max\limits_{L(q,\T(R)) \setminus \varrho} |R_i| \cdot \log_2(k) $$ time. 
	Line \ref{line:knnu:dfnRi} never does more work than line \ref{line:knnu:dfn_C}, since in the worst case scenario $R_{\eta(i)}$ is copied to $R_{i}$ in its current form. Line \ref{line:knnu:dfnindexj} handles $|R_{i}|$ nodes, since we can keep track of value of $\nxt(a,i,\T(R))$ of Definition \ref{dfn:implementation_compressed_cover_tree} by updating it when necessary in line \ref{line:knnu:dfn_C} we can retrieve its value in $O(1)$ time. Therefore maximal run-time of line \ref{line:knnu:dfnindexj} is $\max\limits_{i \in L(q,\T(R)) \setminus \varrho}|R_i|$.  Final line \ref{line:knnu:final_line}  picks lowest $k$-elements $R_{\eta(i)}$ ranked by function 
	$f(p) = d(p,q)$. By Lemma \ref{lem:time_k_smallest_elements} it can be computed in time 
	$O(\log_2(k) \cdot \max\limits_{L(q,\T(R)) \setminus \varrho}|R_i|)$.  
	It follows that 
	\begin{ceqn}
		\begin{equation} \label{eqa:thmimpeqadx}
			\centering
			\max(\li{\ref{line:knnu:loop_begin},\ref{line:knnu:qtoofar:condition}}, \li{\ref{line:knnu:qtoofar:condition:endif},\ref{line:knnu:loop_end}} , \li{\ref{line:knnu:final_line}}) = O\Big(c_m(R)^4 \cdot \max_{i \in L(q,\T(R)) \setminus \varrho}|R_i| \cdot \log_2(k) \Big )
		\end{equation}
	\end{ceqn}
	Let us now bound $\max_{i \in L(q,\T(R)) \setminus \varrho}|R_i|$, by showing $|R_i| \leq c_m(R)^6$. 
	Let $C_i$ be the $i$th level of $\T(R)$ as in Definition \ref{dfn:cover_tree_compressed}.
	For all $i \in L(\T(R),q) \setminus \varrho$ we have:
	\begin{ceqn}
		\begin{align}
			\label{eqa:ModifiedQBoundOne}
			R_{i} &= \{r \in \C_i(R_{\eta(i)}) \mid d(p,q) \leq d(q,\lambda) + 2^{i+2}\} \\
			&= B(q,d(q,\lambda)+2^{i+2}) \cap \C_i(R_i) \\
			&\subseteq B(q,2^{i+3}) \cap C_{i} 
			\label{eqa:ModifiedQboundTwo}
		\end{align}
	\end{ceqn}
	From cover-tree condition we know that all the points in $C_{i}$ are separated by $2^{i}$.
	We will now apply Lemma \ref{lem:packing} with $t = 2^{i+3}$ and $\delta = 2^{i}$.
	Since $4\frac{t}{\delta} + 1 = 2^5 + 1 \leq 2^6$ we obtain 
	$\max\limits_{i \in L(q,\T(R)) \setminus \varrho}|R_{i}| \leq |B(q,2^{i+2} ) \cap C_{i}| \leq c_m(R)^6$. The claim follows by replacing $\max\limits_{i \in L(q,\T(R)) \setminus \varrho}|R_{i}|$ with $c_m(R)^6$ in (\ref{eqa:thmimpeqadx}).

	
	
	\medskip
	
	\noindent
	\textbf{(b)}
	Let us now bound the run-time of $\li{\ref{line:knnu:qtoofar:condition}, \ref{line:knnu:loop_end}}$.
	which runs Algorithm \ref{alg:cover_tree_k-nearest_final_collection} for all $(p,i)$, where $p \in R_i$.
	Let $\Sd$ be a distinctive descendant set from Definition \ref{dfn:distinctive_descendant_set}. 
	Algorithm \ref{alg:cover_tree_k-nearest_final_collection} visits every node $u \in \cup_{p \in R_{i}}\Sd_{i}(p, \T(R))$ once, therefore its running time is $O(\cup_{p \in R_{i}}|\Sd_{i}(p, \T(R))|)$. Let us now show that 
	$$\cup_{p \in R_{i}}\Sd_{i}(p, \T(R)) \subseteq \bar{B}(q, 5d_k(q,R))$$
	Note first that by Lemma \ref{lem:cover_tree_knn_correct_lem} set $\cup_{p \in R_{i}}\Sd_{i}(p, \T(R))$ contains all $k$-nearest neighbors of $q$. Using Lemma \ref{lem:beta_point} we find $\beta$ among $k$-nearest neighbors of $q$ satisfying $d(q,\lambda) \leq d(q,\beta) + 2^{i+1}$. From assumption It follows $2^{i+1} \leq  d(q,\beta)$ .
	\medskip
	
	\noindent
	By line \ref{line:knnu:qtoofar:condition}
	we have $d(q, \lambda) \leq 2^{i+1}$.
	By line \ref{line:knnu:qtoofar} we perform depth-first traversal on $$A = \cup_{p \in R_i}\Sd_i(p, \T(R)).$$ 
	Let $u \in \cup_{p \in R_i}\Sd_i(p, \T(R))$ be arbitrary node and let $v \in R_i$ be such that $u \in \Sd_i(v,\T(R))$.
	By Lemma \ref{lem:distinctive_descendant_distance} we have $d(u,v) \leq 2^{i+1}$. Since $v \in R_i$ we have 
	$d(q,v) \leq d(\lambda,q) + 2^{i+2}$. By triangle inequality
	$$d(u,q) \leq d(u,v) + d(v,q) \leq  2^{i+1} + d(\lambda,v) + 2^{i+2}  \leq 2^{i+1} + 2^{i+1} + d(q,\beta) + 2^{i+2} \leq 5 \cdot d(q,\beta)$$
	It follows that $\cup_{p \in R_i}\Sd_i(p, \T(R)) \subseteq \bar{B}(q,5 \cdot d(q,\beta))$.
	Let us now bound the time complexity of line \ref{line:knnu:qtoofar}. 
	By Lemma \ref{lem:time_k_smallest_elements} for any set $A$ is takes $\log(k) \cdot |A|$ time to select $k$-lowest elements. We have:
	$$\li{\ref{line:knnu:qtoofar:condition}, \ref{line:knnu:loop_end}} = O(|\bar{B}(q,5 \cdot d_k(q,R))| \cdot \log(k)). $$
\end{proof}

\thmgeneraltime*
\begin{proof}
	Apply Lemma~\ref{lem:knn:time} to estimate the time complexity of Algorithm~\ref{alg:cover_tree_k-nearest}: \\ 
	$O\big( |L(\T(R),q)| \cdot 
	(\li{\ref{line:knnu:loop_begin}-\ref{line:knnu:qtoofar:condition}}
	+ \li{\ref{line:knnu:qtoofar:condition:endif}-\ref{line:knnu:loop_end}}  
	+ \li{\ref{line:knnu:final_line}}) 
	+\li{\ref{line:knnu:qtoofar:condition}-\ref{line:knnu:qtoofar:condition:endif} }\big)$.
\end{proof}

\noindent 
Corollary \ref{cor:cover_tree_knn_miniziminzed_constant_time} gives a run-time bound using only minimized expansion constant $c_m(R)$,
where if $R \subset \R^{m}$, then $c_m(R) \leq 2^{m}$. Recall that $\Delta(R)$ is aspect ratio of $R$ introduced in 
Definition \ref{dfn:radius+d_min}.

\corminimizedexpansionconstanttime*
\begin{proof}
	Replace $|L(q,\T(R))|$ in the time complexity of Theorem \ref{thm:cover_tree_knn_general_time} by its upper bound from
	Lemma~\ref{lem:depth_bound}: $|L(q,\T(R))| \leq |H(\T(R))| \leq \log_2(\Delta(R))$.
\end{proof}

\noindent
If we are allowed to use the standard expansion constant, that corresponds to KR-dimension of \citet{krauthgamer2004navigating}, then we obtain a stronger result, Theorem \ref{thm:knn_KR_time}.




\begin{lem}
	\label{lem:upper_bound_to_points_beloning_to_reference_set}
	Let $R$ be a finite reference set in a metric space $(X,d)$ and let $q \in X$ be a query point.
	Let $\varrho$ be the \emph{special} level of $L(\T(R),q)$. Let $i \in L(\T(R),q) \setminus \varrho$ be any level. 
	Then if $p \in R_i$ we have $d(p,q) \leq 2^{i+3}$.
\end{lem}
\begin{proof}
	By assumption in this part of the algorithm we have $d(q, \lambda) \leq 2^{i+2}$. By line \ref{line:knnu:dfnRi} of Algorithm \ref{alg:cover_tree_k-nearest}, since $p \in R_i$ we have $d(p,q) \leq d(q,\lambda) + 2^{i+2} \leq 2^{i+2} + 2^{i+2} \leq 2^{i+3}$, which proves the claim.
\end{proof}

\begin{lem}
	\label{lem:lower_bound_to_points_not_belonging_to_reference_set}
	Let $R$ be a finite reference set in a metric space $(X,d)$ and let $q \in X$ be a query point.
	Let $\varrho$ be the \emph{special} level of $L(\T(R),q)$. Let $i \in L(\T(R),q) \setminus \varrho$ be any level. 
	Then if $p \in \C_{i}(R_{\eta(i)}) \setminus R_i$, we have $d(p,q) > 2^{i+2}$.
\end{lem}
\begin{proof}
	By assumption $p \in \C_{i}(R_{\eta(i)}) \setminus R_i$.
	By line \ref{line:knnu:dfnRi} of Algorithm \ref{alg:cover_tree_k-nearest}
	it follows that $d(q,p) > 2^{i+2} + d(q,\lambda) \geq 2^{i+2}$.
	Therefore $d(q,p) > 2^{i+2}$, which proves the claim. 
\end{proof}

\begin{lem}
	\label{lem:knn_next_level_finder_for_log_depth_two}
	Let $i$ be a non-minimal level of $L(\T(R),q)$ of Definition \ref{dfn:knn_iteration_set}. Assume that $t = \eta(\eta(i+3))$ is defined.
	Then there exists $p \in R$ satisfying $2^{i+2} < d(p,q) \leq 2^{t+4}$.
\end{lem}
\begin{proof}
	
	
	
	
	Note first that since $\eta(i+3) \in L(\T(R),q)$, there exists distinct 
	$u \in R_{\eta(\eta(i+3))}$ and $v \in \C_{\eta(i+3)}(R_{\eta(\eta(i+3)}))$, in such a way that $u$ is the parent of $v$. 
	Let us show that both of $u,v$ cant belong to set $R_i$. Assume contrary that both $u,v \in R_i$. Then by Lemma \ref{lem:upper_bound_to_points_beloning_to_reference_set} we have
	$d(v,q) \leq 2^{i+3}$ and $d(u,q) \leq 2^{i+3}$. By triangle inequality $d(u,v) \leq d(u,q) + d(q,v) \leq 2^{i+4} \leq 2^{\eta(i+3)}$.
	Recall that we denote a level of a node by $l$.
	On the other hand we have $l(u) \geq \eta(i+3)$ and $l(v) \geq \eta(i+3)$, by separation condition of Definition \ref{dfn:cover_tree_compressed} we have $d(u,v) > 2^{\eta(i+3)}$, which is a contradiction. Therefore only one of $\{u,v\}$ 
	can belong to $R_i$. It sufficies two consider the two cases below: 
	
	\medskip
	
	\noindent
	\textbf{Assume that }$v \notin R_i$. Since $v$ is children of $u$ we have $d(u,v) \leq 2^{\eta(i+3) + 1}$.
	By Lemma \ref{lem:upper_bound_to_points_beloning_to_reference_set} we have $d(u,q) \leq 2^{\eta(\eta(i+3)) + 3}$.
	By triangle inequality 
	$$d(v,q) \leq d(v,u) + d(u,q) \leq 2^{\eta(\eta(i+3)) + 3} + 2^{\eta(i+3) + 1} \leq 2^{\eta(\eta(i+3)) + 4}$$
	Since $v \notin R_i$ there exists level $t$ having $\eta(i+3) \geq t \geq i$ and $v \in \C_{t}(R_{\eta(t)}) \setminus R_t$.
	Therefore by Lemma \ref{lem:lower_bound_to_points_not_belonging_to_reference_set} we have $d(q,v) > 2^{t+2} \geq 2^{i+2}$.
	It follows that we have found point $v \in R$ satisfying $2^{i+2} < v \leq 2^{\eta(\eta(i+3)) + 4}$. Therefore $p = v$, is the desired point.
	
	\medskip
	
	\noindent
	\textbf{Assume that }$u \notin R_i$. Since $u \in R_{\eta(\eta(i+3))}$, by Lemma \ref{lem:upper_bound_to_points_beloning_to_reference_set} we have $d(u,q) \leq 2^{\eta(\eta(i+3)) + 3}$.
	On the other hand since $u \notin R_i$, there exists level $t$ having $\eta(i+3) \geq t \geq i$ and $u \in \C_{t}(R_{\eta(t)}) \setminus R_t$. Therefore by Lemma \ref{lem:lower_bound_to_points_not_belonging_to_reference_set} we have $d(q,u) > 2^{t+2} \geq 2^{i+2}$.
	It follows that we have found point $u \in R$ satisfying $2^{i+2} < u \leq 2^{\eta(\eta(i+3)) + 4}$. Therefore $p = u$, is the desired point.

\end{proof}

\lemknndepthbound*
\begin{proof}
    \hypertarget{proof:lem:knn_depth_bound}{\empty}
	Let $x \in L(\T(R),q)$ be the lowest level of $L(\T(R),q)$.
	Define $s_1 = \eta(\eta(x)+1)$ and let $s_i = \eta(\eta(\eta(s_{i-1}+3))+3)$, if it exists. Assume that $s_{n+1}$ is the last sequence element for which $\eta(\eta(\eta(s_{n-1}+3))+3)$ is defined. Define $S = \{s_1,...,s_{n}\}$. For every $i \in \{1,...,n\}$ let $p_i$ be the point provided by Lemma \ref{lem:knn_next_level_finder_for_log_depth_two} that satisfies $$ 2^{s_i+2} < d(p_i,q) \leq 2^{\eta(\eta(s_{i}+3)) + 4}.$$
	Let $P$ be the sequence of points $p_i$.  Denote $n = |P| = |S|$. 
	Let us show that $S$ satisfies the conditions of Lemma \ref{lem:growth_bound_extension}. Note that:
	$$4 \cdot d(p_i,q)\leq 4 \cdot 2^{\eta(\eta(s_{i}+3)) + 4} \leq 2^{\eta(\eta(s_{i}+3)) + 6} \leq 2^{\eta(\eta(\eta(s_{i}+3))+3) + 2} \leq 2^{s_{i+1}+2} \leq d(p_{i+1},q)$$
	By Lemma \ref{lem:growth_bound_extension} applied for $A = R \cup q$ and sequence $P$ we get:
	$$|\bar{B}(q,\frac{4}{3}  d(q,p_n))| \geq (1+\frac{1}{c(R)^2})^{n} \cdot |\bar{B}(q,\frac{1}{3} d(q,p_1))|$$
	Since $\eta(x) \in L(\T(R),q)$ , there exists some point $u \in R_{\eta(x)}$. 
	By Lemma \ref{lem:upper_bound_to_points_beloning_to_reference_set} we have $d(u,q) \leq 2^{\eta(x) + 3}$. 
	Also $2^{\eta(\eta(x) + 1)+1} \leq \frac{2^{\eta(\eta(x) + 1)+2}}{3} < \frac{d(q,p_1)}{3}$
	It follows that:
	$$1 \leq  |\bar{B}(q, 2^{\eta(x) + 3})| \leq |\bar{B}(q, 2^{\eta(\eta(x) + 1)} + 1)| \leq |\bar{B}(q, \frac{d(q,p_1)}{3})|$$
	Therefore we have
	$$|R| \geq \frac{|\bar{B}(q,\frac{4}{3} \cdot d(q,p_n))|}{|\bar{B}(q,\frac{1}{3} \cdot d(q,p_1))|} \geq (1+\frac{1}{c(R \cup \{q\})^2})^{n}$$
	Note that $c(R \cup \{q\}) \geq 2$ by definition of expansion constant. Then by applying $\log$ and by using Lemma \ref{lem:hard_function_bound} we obtain: $c(R \cup \{q\})^2\log(|R|) \geq n = |S|$. 
	Let $x$ be minimal level of $L(\T(R),q)$ and let $y$ be the maximal level of $L(\T(R),q)$ 
	Note that $S$ is a sub sequence of $L$ in such a way that:
	\begin{itemize}
		\item $[x,s_1] \cap L(\T(R),q) \leq 3$, 
		\item for all $i \in 1,..., n$ we have $[s_i, s_{i+1}] \cap L(\T(R),q) \leq 10 $
		\item $[s_n, y] \cap L(\T(R),q) < 20$
	\end{itemize}
	Since segments $[x,s_1],[s_1,s_2], ...,  [s_2,s_n], [s_n,y]$ cover $|L(\T(R),q)|$,
	it follows that $|S| \geq \frac{|L(\T(R),q)|}{20}$. We obtain that $$|L(\T(R),q)| \leq 20 \cdot c(R \cup \{q\})^2 \cdot \log_2(|R|),$$ which proves the claim.
\end{proof}

\thmknnkrtime*
\begin{proof}
	By Theorem~\ref{thm:cover_tree_knn_general_time} the required time complexity is
	$$O\Big ((c_m(R))^{10} \cdot \log_2 (k) \cdot |L(q,\T(R))| + |\bar{B}(q, 5d(q,\beta)) | \cdot \log_2(k) \Big )$$
	for some point $\beta$ among the first $k$-nearest neighbors of $q$.
	Apply Definition \ref{dfn:expansion_constant}:
	\begin{ceqn}
		\begin{align}
			|B(q,5d(q,\beta))| \leq (c(R \cup \{q\}))^3 \cdot |B(q,\frac{5}{8}d(q,\beta))|
		\end{align}
	\end{ceqn}
	Since $|B(q,\frac{5}{8}d(q,\beta))| \leq k$, we have $|B(q,5d(q,\beta))| \leq (c(R \cup \{q\}))^3  \cdot k$.
	It remains to apply Lemma \ref{lem:knn_depth_bound}: $|L(q,\T(R))| = O(c(R \cup \{q\})^2 \cdot \log_2|R|)$.
\end{proof}

\noindent
Corollary~\ref{cor:cover_tree_knn_time} combines Theorem~\ref{thm:construction_time_KR} with Theorem~\ref{thm:knn_KR_time}, to show that
Problem~\ref{pro:knn} can be solved in $O(c^{O(1)} \cdot \log(k) \cdot \max\{|Q|, |R|\} \cdot (\log|R|) + k)$ time. 

\begin{cor}[solution to Problem \ref{pro:knn}]
	\label{cor:cover_tree_knn_time}
	In the notations of Theorem~\ref{thm:knn_KR_time}, 
	set $c = \max\limits_{q \in Q}c(R \cup \{q\})$.
	Algorithms~\ref{alg:cover_tree_k-nearest_construction_whole} and \ref{alg:cover_tree_k-nearest} solve Problem \ref{pro:knn} in time 
	$$O\Big( \max(|Q|,|R|) \cdot  c^2 \cdot \log_2(k) \cdot \big ((c_m(R))^{10}  \cdot \log_2(|R|) + c \cdot k \big ) \Big).$$
\end{cor}
\begin{proof}
	For any $q \in Q$, since $\log_2|R \cup \{q\}| \leq 2\log_2|R|$,
	a tree $\T(R)$ can be built in time $$O(c^2 \cdot c_m(R)^8 \cdot \log|R|)$$ by Theorem~\ref{thm:construction_time_KR}.
	Therefore the time complexity is dominated by running Algorithm~\ref{alg:cover_tree_k-nearest} on all points $q \in Q$.
	The final complexity is obtained by multiplying the time from Theorem~\ref{thm:knn_KR_time} by $|Q|$.
\end{proof}


\section{Approximate $k$-nearest neighbor search}
\label{sec:approxknearestneighbor}




The original navigating nets and cover trees were used in \citet[Theorem~2.2]{krauthgamer2004navigating} and \citet[Section~3.2]{beygelzimer2006cover} to solve the $(1+\epsilon)$-approximate nearest neighbor problem for $k=1$.  
The main result, Theorem~\ref{thm:approximate_k_nearestneighbors} justifies a near linear parameterized complexity to find approximate a $k$-nearest neighbor set $\mathcal{P}$ formalized in Definition \ref{dfn:ApproxKNearestNeighbor}. 

\begin{dfn}[approximate $k$-nearest neighbor set $\AP$]
	\label{dfn:ApproxKNearestNeighbor}
	Let $R$ be a finite reference set and let $Q$ be a finite query set of a metric space $(X,d)$.
	Let $q \in Q \subseteq X$ be a query point, $k \geq 1$ be an integer and $\epsilon > 0$ be a real number. 
	Let $\mathcal{N}_k = \cup_{i=1}^k \NN_i(q)$ be the union of neighbor sets from Definition~\ref{dfn:kNearestNeighbor}. 
	A set $\mathcal{P} \subseteq R$ is called an \emph{approximate $k$-nearest neighbors set}, if $|\mathcal{P}| = k$ and there is an injection $f: \mathcal{P} \rightarrow \mathcal{N}_k$ satisfying $d(q, p) \leq (1+\epsilon) \cdot d(q,f(p)) $ for all $p \in \mathcal{P}$.  
	\bs
\end{dfn}


\begin{algorithm}[ht]
	\caption{This algorithm finds approximate $k$-nearest neighbor of Definition \ref{dfn:ApproxKNearestNeighbor}.}
	\label{alg:cover_tree_k-nearest_approximate}
	\begin{algorithmic}[1]
		\STATE \textbf{Input} : compressed cover tree $\T(R)$, a query point $q\in X$, an integer $ k \in \Z_{+}$, real $\epsilon \in \R_{+}$.
		\STATE Set $i \leftarrow l_{\max}(\T(R)) - 1$ and $\eta(l_{\max}-1) = l_{\max}$. Set $R_{l_{\max}}=\{\text{root}(\T(R))\}$.
		\WHILE{$i \geq l_{\min}$} \alglinelabel{line:aknn:loop_begin}
		\STATE Assign $\mathcal{C}_i(R_{\eta(i)}) \leftarrow  R_{\eta(i)} \cup \{a \in \Child(p) \text{ for some }p \in R_{\eta(i)} \mid l(a) = i \}$. 
		\STATE Compute $\lambda = \lambda_k(q,\C_{i}(R_{\eta(i)}))$ from Definition~\ref{dfn:lambda-point} \alglinelabel{line:aknn:dfnLambda} by 
		Algorithm \ref{alg:lambda}.
		\STATE Find $R_{i} = \{p \in \C_i(R_{\eta(i)}) \mid d(q,p) \leq d(q,\lambda) + 2^{i+2}\}$ \alglinelabel{line:aknn:dfnRi}.
		\IF {$\frac{2^{i+2}}{\epsilon} + 2^{i+1}  \leq d(q, \lambda)$}\alglinelabel{line:aknn:ifcondition}
		\STATE Let $\mathcal{P} = \emptyset$. 
		\FOR {$p \in \C_i(R_{\eta(i)})$}
		\IF {$d(p,q) < d(q,\lambda)$}
		\STATE $\mathcal{P} = \mathcal{P} \cup \Sd_{i}(p,\T(R))$
		\ENDIF
		\ENDFOR
		\STATE Fill $\mathcal{P}$ until it has $k$ points by adding points from sets $\Sd_{i}(p,\T(R))$, where $d(p,q) = d(q, \lambda)$.
		\STATE \textbf{return} $\mathcal{P}$. 
		\ENDIF \alglinelabel{line:aknn:endifcondition}

		\STATE Set $j \leftarrow \max_{ a \in R_{i}} \nxt(a,i,\T(R))$.
		\COMMENT{If such $j$ is undefined, we set $j = l_{\min}-1$} \alglinelabel{line:aknn:dfnindexj}
		\STATE Set $\eta(j) \leftarrow i$ and $i \leftarrow j$.
		\ENDWHILE \alglinelabel{line:aknn:loop_end}
		\STATE Compute and \textbf{output} $k$-nearest neighbors of query point $q$ from the set $R_{l_{\min}}$.
		\alglinelabel{line:aknn:final_line}
	\end{algorithmic}
\end{algorithm}




Definition \ref{dfn:knn_iteration_set_approx} is analog of Definition \ref{dfn:knn_iteration_set} for $(1+\epsilon)$-approximate $k$-nearest neighbor search. 
\begin{dfn}[Iteration set of approximate $k$-nearest neighbor search]
	\label{dfn:knn_iteration_set_approx}
	Let $R$ be a finite subset of a metric space $(X,d)$. 
	Let $\T(R)$ be a cover tree of Definition \ref{dfn:cover_tree_compressed} built on $R$ and let $q \in X$ be an arbitrary point.
	Let $L(\T(R),q) \subseteq H(\T(R))$ be the set of all levels $i$ during iterations of lines~\ref{line:aknn:loop_begin}-\ref{line:aknn:loop_end} of Algorithm~\ref{alg:cover_tree_k-nearest_approximate} launched with inputs $(\T(R),q)$. 
	We denote $\eta(i) = \min_{t} \{ t \in L(\T(R),q) \mid t > i\}$. 
	\bs
	
\end{dfn}

\begin{lem}[$k$-nearest neighbors in the candidate set for all $i$]
	\label{lem:cover_tree_knn_correct_lem_approx}
	Let $R$ be a finite subset of an ambient metric space $(X,d)$, let $q \in X$ be a query point , let $k \in \Z \cap [1,\infty)$ and $\epsilon \in \R_{+}$ be parameters. Let $\T(R)$ be a compressed cover tree of $R$. Assume that $|R| \geq k$. Then for any iteration $i \in L(\T(R),q)$ of Algorithm \ref{alg:cover_tree_k-nearest_approximate} the candidate set $\bigcup_{p \in R_i}\Sd_i(p, \T(R))$ contains all $k$-nearest neighbors of $q$. \bs
\end{lem}
\begin{proof}
	Proof of this lemma is similar to Lemma \ref{lem:cover_tree_knn_correct_lem_approx} and is therefore omitted.
\end{proof}

\noindent 
Lemma \ref{lem:approximate_knn_correctness} shows that Algorithm \ref{alg:cover_tree_k-nearest_approximate} correctly returns an Approximate $k$-nearest neighbor set of Definition \ref{dfn:ApproxKNearestNeighbor}.

\begin{lem}[Correctness of Algorithm \ref{alg:cover_tree_k-nearest_approximate}]
	\label{lem:approximate_knn_correctness}
	
	Algorithm \ref{alg:cover_tree_k-nearest_approximate}
	finds an approximate $k$-nearest neighbors set of any query point $q \in X$. 
	\bs
\end{lem}
\begin{proof}
	
	Assume first that condition on line \ref{line:aknn:ifcondition} of Algorithm \ref{alg:cover_tree_k-nearest_approximate}
	is satisfied during some iteration $i \in H(\T(R))$ of Algorithm \ref{alg:cover_tree_k-nearest_approximate}. Let us denote
	$$\mathcal{A} = \bigcup_{p \in \C_i(R_{\eta(i)})} \{\Sd_{i}(p,\T(R)) \mid d(p,q) < d(q,\lambda) \}, 
	\mathcal{B} = \bigcup_{p \in \C_i(R_{\eta(i)})} \{\Sd_{i}(p,\T(R)) \mid d(p,q) = d(q,\lambda) \}.$$
	By Algorithm \ref{alg:cover_tree_k-nearest_approximate} set $\mathcal{P}$ contains all points of $\mathcal{A}$ and rest of the points are filled form $\mathcal{B}$.
	We will now form $f: \mathcal{P} \rightarrow \mathcal{N}_k$ by mapping every point $p \in \mathcal{A} \cap \mathcal{P}$ into itself and then by extending $f$ to be injective map on whole set $\mathcal{P}$ . The claim holds trivially for all points $p \in \mathcal{A} \cap \mathcal{P}$. Let us now consider points  $p \in \mathcal{P} \setminus \mathcal{A}$. Let $\gamma \in \C_i(R_{\eta(i)})$ be such that $p \in \Sd_i(\gamma, \T(R))$ and let $\psi \in \C_i(R_{\eta(i)})$ be such that $f(p) \in  \Sd_i(\psi, \T(R)) $. By using triangle inequality, Lemma \ref{lem:compressed_cover_tree_descendant_bound} and the fact that
	$p \in \mathcal{A} \cup \mathcal{B}$ we obtain:
	\begin{ceqn}
		\begin{equation}
			\label{eqa:ANNCorrectness1}
			d(q, p) \leq  d(q, \gamma) + d(\gamma,p) \leq d(q, \lambda) + 2^{i+1}
		\end{equation}
	\end{ceqn} 
	On the other hand since $f(p) \notin \mathcal{A}$ we have
	\begin{ceqn}
		\begin{equation}
			\label{eqa:ANNCorrectness2}
			(1+\epsilon) \cdot d(q, f(p)) \geq (1+\epsilon) \cdot ( d(q, \psi) - d(\psi,f(p))) \geq (1+\epsilon) \cdot (d(q, \lambda) - 2^{i+1})
		\end{equation}
	\end{ceqn}
	Note that by line \ref{line:aknn:ifcondition} we have $\frac{2^{i+2}}{\epsilon} + 2^{i+1} \leq d(q, \lambda)$. It follows that 
	$2^{i+2} \leq \epsilon \cdot d(q, \lambda) - \epsilon \cdot 2^{i+1}$.
	Therefore we have:
	\begin{ceqn}
		\begin{equation}
			\label{eqa:ANNCorrectness3}
			d(q, \lambda) + 2^{i+1}  \leq d(q,\lambda) + 2^{i+2} - 2^{i+1} \leq (1+\epsilon) \cdot (d(q,\lambda) - 2^{i+1})
		\end{equation}
	\end{ceqn}
	By combining Equations (\ref{eqa:ANNCorrectness1}) - (\ref{eqa:ANNCorrectness3}) we obtain $d(q, p) \leq (1+\epsilon) \cdot d(q,f(p)) $.
	If the condition on line \ref{line:aknn:ifcondition} of Algorithm \ref{alg:cover_tree_k-nearest_approximate} is never satisfied, then the Algorithm finds real $k$-nearest neighbors of point $q$ in the end of the algorithm and therefore the claim holds. 
	
\end{proof}

\begin{thmm}[Time complexity of Algorithm~\ref{alg:cover_tree_k-nearest_approximate} ]
	\label{thm:approximate_k_nearestneighbors}
	In the notations of Definition \ref{dfn:ApproxKNearestNeighbor}, the complexity of Algorithm \ref{alg:cover_tree_k-nearest_approximate} is
	$$O\Big(  (c_m(R))^{8 + \lceil \log(2 + \frac{1}{\epsilon}) \rceil}  \cdot \log_2(k) \cdot \log_2(\Delta(R)) + k \Big ).$$
	\bs
\end{thmm}
\begin{proof}
	Similarly to Lemma \ref{lem:knn:time} it can be shown that Algorithm  \ref{alg:cover_tree_k-nearest_approximate}  is bounded by: 
	\begin{ceqn}
		\begin{equation}
			\label{eqa:boundapproximate}
			O((c_m(R))^4 \cdot \log_2(k) \cdot \max_i|R_i| \cdot |H(\T(R))| + \li{\ref{line:aknn:ifcondition} - \ref{line:aknn:endifcondition}})
		\end{equation}
	\end{ceqn}
	
    \noindent
	Note first that in lines \ref{line:aknn:ifcondition} - \ref{line:aknn:endifcondition} we loop over set $\C_i(R_{\eta(i)})$ and select $k$ points from it. Therefore $\li{\ref{line:aknn:ifcondition} - \ref{line:aknn:endifcondition}} = k + |\C_i(R_{\eta(i)})|$.

	\medskip
	
	\noindent
	Let us now bound the size of $R_i$. By line \ref{line:aknn:ifcondition} of Algorithm \ref{alg:cover_tree_k-nearest_approximate} either Algorithm \ref{alg:cover_tree_k-nearest_approximate} is launched that terminates the program or $\frac{2^{i+2}}{\epsilon} + 2^{i+1}  > d(q, \lambda)$. Let $C_i$ be the $i$th cover set of $\T(R)$.
	To bound $|R_i|$ we can assume the latter. Similarly to Theorem \ref{thm:knn_KR_time} we have:
	\begin{ceqn}
		\begin{align}
			\label{eqa:QBoundOne1}
			R_{i} &= \{r \in \C_i(R_{\eta(i)}) \mid d(p,q) \leq d(q,\lambda) + 2^{i+2}\} \\
			&= \bar{B}(q,d(q,\lambda)+2^{i+2}) \cap \C_i(R_{\eta(i)}) \\
			&\subseteq \bar{B}(q,d(q,\lambda)+2^{i+2}) \cap C_{i} \\
			&\subseteq \bar{B}(q,2^{i+2}(\frac{3}{2} + \frac{1}{\epsilon})) \cap C_{i} 
			\label{eqa:QboundTwo1}
		\end{align}
	\end{ceqn}
	Since the cover set $C_{i}$ is a $2^{i}$-sparse subset of the ambient metric space $X$, we can apply Lemma~\ref{lem:packing} with $t = 2^{i+2}(\frac{3}{2} + \frac{1}{\epsilon})$ and $\delta = 2^{i}$. 
	Since $4\frac{t}{\delta} + 1 = 2^4(\frac{3}{2} + \frac{1}{\epsilon}) + 1 \leq 2^4(2 + \frac{1}{\epsilon})$, we get $\max |R_i| \leq (c_m(R))^{4 + \lceil \log_2(2 + \frac{1}{\epsilon}) \rceil}$. 
	The final complexity is obtained by plugging the upper bound of $|R_i|$ above into (\ref{eqa:boundapproximate}).
\end{proof}

\begin{cor}[complexity for approximate $k$-nearest neighbors set $\AP$]
	\label{cor:approximate_k_nearestneighbors}
	In the notations of Definition \ref{dfn:ApproxKNearestNeighbor}, an approximate $k$-nearest neighbors set is found for all $q \in Q$ in time
	$O\Big( |Q| \cdot  (c_m(R))^{8 + \lceil \log(2 + \frac{1}{\epsilon}) \rceil} \cdot \log(k) \cdot \log_2(\Delta(R)) + |Q| \cdot k \Big ).$
	\bs
\end{cor}
\begin{proof}
	This corollary follows directly from Theorem \ref{thm:approximate_k_nearestneighbors} .
\end{proof}

\section{Discussions: current contributions and future steps}
\label{sec:Conclusions}

This paper rigorously proved the time complexity of the exact $k$-nearest neighbor search.
The motivations were the past gaps in the proofs of time complexities in \citet[Theorem~5]{beygelzimer2006cover}, \citet[Theorem~3.1]{ram2009linear}, \citet[Theorem~5.1]{march2010fast}.
Though \citet{elkin2022counterexamples} provided concrete counterexamples, no corrections were published.
Main Theorem~\ref{thm:knn_KR_time} and Corollary~\ref{cor:construction_time_KR} have finally filled the above gaps.
\medskip

\noindent
To overcome all past obstacles, first Definition~\ref{dfn:kNearestNeighbor} and Problem~\ref{pro:knn} rigorously dealt with a potential ambiguity of $k$-nearest neighbors at equal distances, which was not discussed in the past work.
\medskip

\noindent
A new compressed cover tree in Definition~\ref{dfn:cover_tree_compressed} substantially simplified the navigating nets \citet{krauthgamer2004navigating} and original cover trees  \citet{beygelzimer2006cover} by avoiding any repetitions of given data points.
This compression has substantially clarified the construction and search Algorithms~\ref{alg:cover_tree_k-nearest_construction_whole} and~\ref{alg:cover_tree_k-nearest}. 
\medskip

\noindent
Second, section~\ref{sec:minimized_exp_constant} showed that the new minimized expansion constant $c_m$ of any finite subset $R$ of a normed vector space $\R^{n}$ has the upper bound $2^{m}$. In the future, it can be similarly shown that if $R$ is uniformly distributed then classical expansion constant $c(R)$ is $2^{m}$ as well.
\medskip

\noindent
Third, sections~\ref{sec:ConstructionCovertree} and~\ref{sec:better_approach_knn_problem} corrected the approach of \citet{beygelzimer2006cover} as follows.
Assuming that expansion constants and aspect ratio of a reference set $R$ are fixed, Corollaries~\ref{cor:construction_time_KR} and~\ref{cor:cover_tree_knn_time} rigorously showed that the times are linear in the maximum size of $R,Q$ and near-linear $O(k\log k)$ in the number $k$ of neighbors. 
\medskip

\noindent
The future problem is to improve the complexity of $k$-nearest neighbor search to a pure linear time $O(c(R)^{O(1)}|R|$ by using cover trees on both sets $Q,R$. 
Since a similar approach \citet{ram2009linear} was shown to have incorrect proof in \citet[Counterexample~6.5]{elkin2022counterexamples} and \citet{curtin2015plug,elkin2022paired} used additional parameters $I, \theta$, this goal will require significantly more effort to understand if $O(c(R)^{O(1)}|R|)$ is achievable by using a compressed cover tree.
\medskip

\noindent
Corollary~\ref{cor:cover_tree_knn_time} allowed us to justify the near-linear time of generically complete PDD \citet{widdowson2021pointwise} invariants (Pointwise Distance Distributions), which recently distinguished  all (more than 660 thousand) periodic crystals in the world's largest database of real materials \citet{widdowson2022average}. 
Due to these ultra-fast invariants, more than 200 billion pairwise comparisons were completed over two days on a modest desktop while past tools were estimated to require over 34 thousand years \citet{widdowson2022resolving}. 
The huge speed of PDD is complemented by slower but provably complete invariant isosets \citet{anosova2021isometry} with continuous metrics that allow polynomial-time approximations \citet{anosova2022algorithms}.
\medskip

\noindent
For the purpose of open access, the authors applied a Creative Commons Attribution (CC BY) licence to any accepted version.


\end{document}